%% file: main.tex
\documentclass[journal,twocolumn,10pt,final]{IEEEtran}
\usepackage{caption}
\usepackage{subcaption}  %
\usepackage{algorithm}
\usepackage{algpseudocode}
\usepackage{algorithm} 
\usepackage[table]{xcolor}
\usepackage{longtable}
\usepackage{graphicx}
\usepackage[hidelinks]{hyperref}
\usepackage{bookmark}

\usepackage[T1]{fontenc}
\usepackage{lmodern} %
\usepackage{amssymb,stmaryrd,amsmath,amsfonts,rotating}
\usepackage{amsmath,amsthm,amssymb,cases}
\usepackage{tikz}
\usetikzlibrary{arrows}
\usetikzlibrary{tikzmark}
\usetikzlibrary{matrix}
\usetikzlibrary{positioning}
\usepackage{algpseudocode}
\usepackage{amsmath,cases}
\usepackage{bbm}
\usepackage{amsthm}
\usepackage{booktabs}
\usepackage{xcolor}
\usepackage{amssymb}
\usepackage{bm}
\usepackage{thmtools}

\theoremstyle{definition}%
\declaretheoremstyle[style=definition,qed=\openbox,]{ppstyle}

\newtheorem{teiri}{Theorem}
\newtheorem{lem}{Lemma}
\newtheorem{giron}{Discussion}

\newtheorem{pp}{Proposition}
\usepackage{mathtools}
\mathtoolsset{showonlyrefs}
\declaretheorem[name=Example, style=ppstyle,]{example}
\declaretheorem[name=Definition, style=ppstyle,]{df}
\declaretheorem[name=Requirement, style=ppstyle,]{yosei}

\newcommand{\rev}[1]{#1}

\input{notation}
\usepackage{dsfont}
\renewcommand{\I}{\mathds{1}}

\ifCLASSINFOpdf
\else
\fi

\begin{document}
\title{Quantum Error Correction Exploiting Degeneracy to Approach the Hashing Bound}
\author{Kenta Kasai%
\thanks{Kenta Kasai is with the Institute of Science Tokyo, Tokyo 152-8550, Japan (e-mail: kenta@ict.eng.isct.ac.jp).}\thanks{Author's accepted manuscript for \emph{IEEE Transactions on Information Theory}, DOI: \href{https://doi.org/10.1109/TIT.2026.3747884}{10.1109/TIT.2026.3747884}.
\copyright\ 2026 IEEE. Personal use of this material is permitted. Permission from IEEE must be obtained for all other uses, in any current or future media, including reprinting/republishing this material for advertising or promotional purposes, creating new collective works, for resale or redistribution to servers or lists, or reuse of any copyrighted component of this work in other works.}}
\maketitle

\begin{abstract}
Quantum error correction is essential for realizing scalable quantum computation.  
Among various approaches, low-density parity-check codes over higher-order Galois fields have shown promising performance due to their structured sparsity and compatibility with iterative decoding algorithms whose per-iteration computational complexity scales linearly with the number of physical qubits for fixed field size and degrees.  
In this work, we demonstrate that explicitly exploiting the degeneracy of quantum errors-i.e., the equivalence, up to a global phase on the code space, of errors differing by a stabilizer-can significantly enhance the decoding performance.  
\rev{Simulation results over the depolarizing channel indicate that the proposed method, at a coding rate of $1/3$, achieves an empirical frame error rate as low as $4.80\times10^{-5}$ at a physical error rate of $9.48\%$ for a code with $104{,}000$ logical qubits and $312{,}000$ physical qubits, thereby moving the operating point toward the quantum hashing-bound benchmark.}
These findings highlight the critical role of degeneracy in closing the gap to the fundamental limits of quantum error correction.
\end{abstract}
\IEEEpeerreviewmaketitle

\section{Introduction}

Quantum error correction  is a cornerstone for realizing fault-tolerant quantum computation, protecting quantum information from the inevitable decoherence and operational errors encountered in quantum hardware~\cite{shor95},\cite{steane96b}. 
Since Shor introduced the first quantum error-correcting codes in 1995, extensive research has led to various code families, including stabilizer codes~\cite{gottesmanthesis}, topological codes~\cite{kitaev2003fault},\cite{dennis2002topological}, and concatenated codes~\cite{knill1996concatenated}.

Among these, surface codes, introduced by Kitaev in the late 1990s~\cite{kitaev2003fault}, have emerged as particularly promising candidates due to their local structure and high error-correction thresholds. Surface codes have demonstrated error thresholds approaching $1\%$ under realistic noise models~\cite{fowler2012surface},\cite{google2023surface}, making them a leading architecture for scalable quantum computing. However, a major drawback of surface codes is their inherently low coding rate, typically near zero, which significantly increases the required number of physical qubits and associated hardware overhead.

To address the issue of low coding rates, quantum low-density parity-check (LDPC) codes have attracted considerable attention over the past two decades. Originally developed in classical coding theory by Gallager in 1962~\cite{Gallager1962}, LDPC codes were reintroduced and extensively studied in the 1990s due to their near-capacity performance and efficient iterative decoding algorithms~\cite{mackay1996near},\cite{richardson2001design}. Quantum LDPC codes were subsequently proposed as stabilizer codes characterized by sparse parity-check matrices, inheriting the classical advantages of efficient decoding~\cite{mackay2004sparse},\cite{4557323}.

Understanding and improving the minimum distance of quantum LDPC  codes is a central challenge in quantum error correction. Early constructions such as hypergraph product codes~\cite{tillich2013quantum} guaranteed a minimum distance that scales only as $O(\sqrt{n})$, which limits their ability to suppress logical errors at low physical error rates. Constructions based on fiber bundles~\cite{hastings2021fiber} and balanced products~\cite{breuckmann2021balanced} achieved minimum distances beyond the square-root scale. Subsequently, asymptotically good quantum LDPC codes with linear minimum distance and constant rate were constructed in~\cite{panteleev2022asymptotically}. These advances mark a significant step toward scalable quantum error correction, narrowing the gap between quantum and classical coding theory.

Sparse-graph codes with non-vanishing rate have shown that minimum distance alone is insufficient to fully characterize their practical performance~\cite[Section 13.8]{MacKay-2003},~\cite[Section 3.B]{mackay2004sparse}, and~\cite[Example 1.18]{Richardson:2008:MCT:1795974}.  
Iterative probabilistic decoding algorithms, such as belief propagation (BP), can achieve reliable error correction by effectively exploiting the structural sparsity of these codes.  
These algorithms perform well even in parameter regimes where conventional minimum-distance-based evaluations would predict failure.  
Therefore, assessing codes solely by their minimum distance overlooks key factors such as graph structure and decoder interaction-both of which are crucial for achieving near-capacity performance with practical decoding complexity.  

Classical LDPC codes are known to scale well with increasing block length, naturally exhibiting good scalability~\cite{ezri2007generalization}. In contrast, quantum LDPC codes have faced significant challenges in operating at non-vanishing constant rates with long block lengths. This difficulty has largely been attributed to the presence of degeneracy, as well as to the stringent constraints imposed on the Tanner graph by the orthogonality condition. Both factors interfere with the convergence of BP decoding and suppress the threshold behavior.

To improve the performance of BP decoding in quantum LDPC codes, various techniques have been developed. One approach combines BP with ordered statistics decoding (BP+OSD)~\cite{fossorier2002soft}, which significantly enhances decoding accuracy at the expense of increased complexity~\cite{panteleev2021degenerate}. Another approach focuses on identifying and mitigating trapping sets-small harmful substructures in the Tanner graph that lead to decoder failure-enabling targeted remedies to reduce the error floor~\cite{Raveendran2021TrappingSets,Chytas2025EnhancedMinSum}. More recently, Yao et al.~\cite{yao2024belief} proposed a guided-decimation decoder, where variable nodes are sequentially fixed based on BP likelihoods, and demonstrated that this method achieves performance comparable to BP+OSD while avoiding costly matrix inversions or stabilizer inactivation. Another notable line of work by Miao et al.~\cite{miao2024joint} proposes a code construction method that eliminates length-4 cycles in the joint Tanner graph, thereby mitigating short-cycle-induced decoding failures and improving BP convergence.

Another major advancement is the development of non-binary quantum LDPC codes over higher-order finite fields, motivated by the superior decoding performance observed in classical non-binary LDPC codes~\cite{DaveyMacKayGFq,bennatan2006design,4155118}. While binary LDPC codes typically exhibit optimal performance with column weight three, non-binary LDPC codes achieve their best performance with column weight two. Leveraging this property, Kasai et al.~\cite{6017122} introduced quantum LDPC codes constructed from circulant permutation matrices (CPMs), using non-binary LDPC codes with column weight two as the underlying structure. This construction demonstrated both improved decoding performance and enhanced structural flexibility. Building upon this approach, Komoto and Kasai~\cite{komoto2024quantumerrorcorrectionnear} further generalized the methodology by proposing a systematic framework based on {\itshape affine permutation matrices} (APMs), thereby enabling greater design freedom and decoding robustness. It should be noted that although the quantum error-correcting codes proposed in~\cite{6017122,komoto2024quantumerrorcorrectionnear} are constructed using non-binary LDPC codes, the resulting quantum codes are ultimately binary CSS codes.

GPU implementations are among the architectures studied for non-binary LDPC decoding~\cite{9606216}. For binary 5G LDPC codes, a separate implementation reports 27.6 Gbps on an RTX4090 using five layered iterations of normalized min-sum decoding~\cite{liu2024gpu}.

The high-rate codes (with row weight $L \ge 8$ and rate $R \ge 1/2$) proposed in~\cite{komoto2024quantumerrorcorrectionnear} achieve a sharp threshold using joint BP alone (see Figure~\ref{fig:comparison_ISIT}). Experimental results demonstrate an exponential decrease in error rates as the block length increases. However, as the code length increases and the rate decreases, these codes exhibit more pronounced error-floor behavior in the frame error rate (FER) curves at low noise levels (see Figure~\ref{fig:comparison_ISIT}). This error floor is attributed to the presence of short cycles in the Tanner graph.
\rev{In particular, while the non-binary constructions in~\cite{komoto2024quantumerrorcorrectionnear} already showed that, for $L \ge 8$, FER performance near the hashing-bound benchmark can be obtained around the $10^{-4}$ level, the smaller-$L$ cases, especially $L=6$ and $L=8$, continued to suffer from a pronounced error floor.}
For codes with row weight $L = 8$ and rate $R = 1/2$, the issue was addressed in~\cite{kasai2025quantumerrorcorrectiongirth16} by exploiting the non-commutativity of APMs to develop a modified construction method that eliminates such cycles. As a result, the FER was successfully reduced to $10^{-5}$ at a physical error rate of $p_D = 6.49\%$ over the depolarizing channel. Notably, this performance approaches the hashing bound for the same rate, which is given by $p_D^* = 7.44\%$.
\rev{That approach improves the $L=8$ case by enlarging the girth to 16 and thereby reducing the number of harmful short structures.}

\rev{The purpose of the present paper is different. We focus on codes with row weight $L = 6$ and rate $R = 1/3$, for which the error-floor behavior is more pronounced than in the $L = 8$, $R = 1/2$ case, and we target a FER of $10^{-4}$. In this regime, the girth is already constrained by the unavoidable length-12-cycle structure, so our objective is not to remove all short cycles but to control their algebraic effect: we assign nonzero symbols so that the inevitable length-12 cycles do not become codewords. In this sense, the present paper aims to suppress the error floor by cycle-to-codeword avoidance rather than by further girth enlargement.}

\rev{This viewpoint also clarifies why scaling with respect to the field size $q$ is nontrivial. When $q$ is small, the number of available nonzero symbols is limited, and it becomes much harder to satisfy the CSS orthogonality constraints while simultaneously assigning symbols so that the relevant length-12 cycles do not generate nonzero codewords. Thus, moving to smaller $q$ changes not only the decoder alphabet but also the feasibility of the code construction itself.}

\rev{In the failure examples studied here, the errors uncorrected by joint BP are concentrated on a small number of short cycles. This opens the possibility of correcting them through efficient post-processing~\cite{panteleev2021degenerate}. Preliminary results toward this goal were briefly explored in~\cite{kasai2025efficient}, which focused on a specific class of graph structures. In contrast, the present work significantly extends that direction by proposing a new code construction method together with theoretical justification and empirical validation.}

This paper presents a new code construction method for quantum error correction based on non-binary LDPC codes, extending the conventional framework~\cite{komoto2024quantumerrorcorrectionnear} with the following key innovations:
\begin{itemize}
    \item We identify inevitable cycle structures that arise from the orthogonality condition inherent to CSS-type constructions.
    \item We identify cycles whose codewords are stabilizers and assign labels that eliminate nonzero codewords on the remaining unavoidable cycles.
    \item \rev{We suppress the harmful effect of the inevitable short cycles by strategically assigning nonzero symbols over higher-order Galois fields so that the relevant length-12 cycles do not become codewords.}
    \item We develop a decoding algorithm that explicitly accounts for degeneracy, leading to substantial improvements in decoding performance.
    \item Simulation results over the depolarizing channel show that the proposed code, with $n = 312{,}000$ physical qubits, achieves an empirical frame error rate (FER) of $4.80\times10^{-5}$ at a physical error rate of $p_D = 9.48\%$, approaching the quantum hashing bound.
\end{itemize}
Together, these contributions provide a viable pathway toward practical, high-performance quantum error correction by balancing algebraic code design with decoding robustness.

We summarize the notation used throughout this paper as follows.  
For a natural number \(n\), we define \([n] = \{0, 1, \ldots, n-1\}\).  
The finite field with \(q\) elements is denoted by \(\Fb_q\), and the ring of integers modulo \(n\) is denoted by \(\Zb_n\), representing the set of residue classes modulo \(n\).  
Vectors and one-dimensional arrays are denoted using underlined symbols, such as 
\begin{align}
 \xU = (x_0, x_1, \ldots, x_{N-1}).  
\end{align}
Unless otherwise stated, all vectors are treated as column vectors, even when written in row form.  
To simplify notation, the transpose symbol is omitted wherever possible.  
Indices for arrays and matrices are assumed to start from zero, unless noted otherwise.  
Elements of the finite field \(\Fb_q\) are typically denoted using Greek letters such as \(\xi\).  

\section{Permutation Matrices for LDPC Codes}\label{011257_4Jun25}
This section provides a foundational framework for understanding how permutation matrices (PMs) are used in the construction of structured LDPC codes, particularly in the quantum setting.  
We begin by introducing permutation matrices and their algebraic properties, followed by two important subclasses: APMs and CPMs.  
We then discuss how arrays of PMs can be used to define binary LDPC parity-check matrices and classify such constructions as protograph-based, APM-based, or QC-based codes depending on the permutation class.  
To analyze the structure of these matrices, we introduce the concept of {\itshape block cycles}, define their composite functions, and explain how they relate to the existence of cycles in the Tanner graph.  
In particular, we focus on totally closed block cycles (TCBCs), which play a central role in determining the upper bound of girth.
We conclude the section by presenting an explicit example of a TCBC and a theorem establishing an upper bound on the girth for QC-LDPC codes, motivating the use of non-commutative APMs for our construction.
\subsection{Permutation Matrices (PMs)}
Let $P$ be a positive integer. Define $\mathcal{F}_P$ as the set of all permutations from $[P]$ to itself. 
The identity permutation is written as $\mathrm{id}$.
A permutation \( f \in \mathcal{F}_P \) is associated with a binary permutation matrix \( F \in \Fb_2^{P \times P} \) of size \( P \) through the rule:  
\begin{align}
  f(j) = i \iff F_{i,j} = 1, 
\end{align}
for $i,j\in [P]$. 
We denote this correspondence by $f \sim F$, or simply write $f = F$ when the context is clear.
With a slight abuse of notation, we will freely switch between viewing $f$ as a permutation and as its associated matrix $F$.

This correspondence preserves the algebraic structure of permutations in the matrix domain.  
Let $f, g \in \mathcal{F}_P$ and $F, G \in \Fb_2^{P\times P}$ with $f \sim F$ and $g \sim G$.  
Then the composition of permutations corresponds to matrix multiplication: the composed permutation $f \circ g$ satisfies $f \circ g \sim FG$.  
That is, applying $g$ first and then $f$ is equivalent to multiplying the corresponding permutation matrices $G$ and $F$ (from right to left).  
We say that two permutations $f$ and $g$ \emph{commute} if $f \circ g = g \circ f$.
Furthermore, the inverse of a permutation corresponds to the transpose of its matrix: $f^{-1} \sim F^\top$.  
This follows from the fact that transposing a permutation matrix swaps its rows and columns, effectively reversing the direction of the permutation.

The following lemma will be used several times later, so we assign it a number for easy reference.  
\begin{lem}\label{lem:194543_14Jun25}
If $f$ and $g$ commute, then each of the following pairs also commute: $f^{-1}$ and $g^{-1}$, $f$ and $g^{-1}$, and $f^{-1}$ and $g$.
\end{lem}
\begin{proof}
Suppose $f \circ g = g \circ f$.  
Then we have
\[
f \circ g \circ g^{-1} = g \circ f \circ g^{-1},
\quad \text{so} \quad
f = g \circ f \circ g^{-1}.
\]
Composing both sides on the left with $g^{-1}$ gives
\[
g^{-1} \circ f = f \circ g^{-1}.
\]
Thus, $f$ and $g^{-1}$ commute.  
In the same way, we obtain
\[
f^{-1} \circ g = g \circ f^{-1}
\quad \text{and} \quad
f^{-1} \circ g^{-1} = g^{-1} \circ f^{-1}.
\]
\end{proof}
\subsection{Affine and Circulant Permutation Matrices}\label{subsec:144040_31May25}
For integers \( a, b \in \Zb_P \), we define a permutation \( f: \Zb_P \to \Zb_P \) of the form  
\[
f(j) = a j + b \mod P.
\]
Note that the sets \(\Zb_P\) and \([P]\) are identical as sets.  
It is known that \( f \in \Fc \) if and only if \( \gcd(a, P) = 1 \)~\cite{myung2006combining}.  
A permutation matrix arising from a permutation of the above form is called an \emph{affine permutation matrix} (APM).  
In this paper, we identify permutation functions with their corresponding permutation matrices, and hence refer to such matrices as APMs.  
We denote the set of all such matrices by \( \Ac_P \).  
When \( a = 1 \), the affine permutation reduces to a simple shift: \( f(j) = j + b \mod P \).  
The corresponding permutation matrix is called a \emph{circulant permutation matrix} (CPM), and we denote the set of all such matrices by \( \Qc_P \).  
Any two elements \( f, g \in \Qc_P \) commute with each other; that is, \( f \circ g = g \circ f \) always holds.  
In contrast, two affine permutations \( f, g \in \Ac_P \) do not necessarily commute.
\begin{example}
Consider the affine permutation over $\Zb_5$ defined by
\[
f(j) = 2j + 1 \mod 5.
\]
Note that both row and column indices are taken modulo $5$. The corresponding permutation matrix $F \in \Fb_2^{5 \times 5}$ is given by
\[
\renewcommand{\arraystretch}{0.8}
\setlength{\tabcolsep}{0.1pt}
F = \begin{pmatrix}
0 & 0 & 1 & 0 & 0 \\
1 & 0 & 0 & 0 & 0 \\
0 & 0 & 0 & 1 & 0 \\
0 & 1 & 0 & 0 & 0 \\
0 & 0 & 0 & 0 & 1
\end{pmatrix}.
\]
The inverse permutation is given by
\[
f^{-1}(i) = 3i + 2 \mod 5,
\]
which can be verified directly by checking that $f \circ f^{-1} = \mathrm{id}$.
The inverse permutation corresponds to the transpose of the matrix:
\[
\renewcommand{\arraystretch}{0.8}
F^\top = \begin{pmatrix}
0 & 1 & 0 & 0 & 0 \\
0 & 0 & 0 & 1 & 0 \\
1 & 0 & 0 & 0 & 0 \\
0 & 0 & 1 & 0 & 0 \\
0 & 0 & 0 & 0 & 1
\end{pmatrix}.
\]
\end{example}
\subsection{Low-Density Parity-Check Matrices}\label{045449_15Feb25}
We represent the column and row weights by integers $J \ge 2$ and even $L \ge 4$, respectively.  
Let $\HH = (f_{ij}) \in \Fc_P^{J \times L}$ be a permutation array whose entries are taken from a set of permutations $\Fc_P$.  
This array can also be interpreted as a binary parity-check matrix in $\HH=(f_{ij})=(F_{ij})\in \Fb_2^{JP \times LP}$.  
We refer to each $P \times P$ matrix appearing in $\HH$ as a \emph{block}.  
LDPC codes defined by such parity-check matrices are conventionally referred to as protograph codes, APM-LDPC codes, and quasi-cyclic LDPC (QC-LDPC) codes when the permutations are taken from $\Fc_P$, $\Ac_P$, and $\Qc_P$, respectively.  
As an example, see the parity-check matrices $\HH_X$ and $\HH_Z$ in Example~\ref{ex:021639_17Apr25}.  
\subsection{Block Cycles}
Let $\HH = (f_{ij})$ be a parity-check matrix, where each $f_{ij} \in \Fc_P$.  
We consider a path that moves alternately in horizontal and vertical directions across the blocks of $\HH$, beginning and ending at the same block.  
Such a path is represented as
\[
\cU:=f_{11} \to f_{12} \to f_{21} \to f_{22} \to \cdots \to f_{n1} \to f_{n2} \to f_{11}
\]
within the following submatrix $S$ of $\HH$:
\begin{align}
S= \left(
    \begin{array}{llllllll}
  f_{11}        & f_{12}        & *                & *                & *                & *                \\
  *             & f_{21}        & f_{22}           & *                & *                & *                \\
  *             & *             & \ddots           & \ddots           & *                & *                \\
  *             & *             & *                & *                & f_{n-1,1}           & f_{n-1,2}           \\
  f_{n2}        & *             & *                & \cdots           & *                & f_{n1}           
    \end{array}
\right).
\end{align}
In this construction, each step in the path moves either to a different row or to a different column, but not both simultaneously.  
That is, consecutive blocks must differ in exactly one of their row or column indices.  
For instance, the column indices of $f_{11}$ and $f_{12}$ must be different, and likewise, the row indices of $f_{12}$ and $f_{21}$ must differ.
Such a path $\cU$ is referred to as a \emph{block cycle}.

We define the \emph{composite function} associated with the block cycle as
\begin{align}
f_{\cU}(j) := \bigl(f_{n2}^{-1} f_{n1} \cdots f_{22}^{-1} f_{21} f_{12}^{-1} f_{11}\bigr)(j),\label{143950_31May25}
\end{align}
for $j \in \Zb_P$, where the composition operator $\circ$ is omitted for simplicity.  
Equivalently, the inverse function can be written as
\begin{align}
f_{\cU}^{-1}(j) = \bigl(f_{11}^{-1} f_{12} \cdots f_{n1}^{-1} f_{n2}\bigr)(j), \quad j \in \Zb_P.
\end{align}
We call the block cycle \emph{closed} if $f_{\cU}(j) = j$ for some $j \in \Zb_P$, and \emph{open} otherwise.  
In particular, if $f_{\cU}(j) = j$ for all $j \in \Zb_P$, that is, if $f_{\cU} = \mathrm{id}$, the block cycle is said to be \emph{totally closed}.

For a block cycle $\cU$ in $\HH$, there are at most $P$ corresponding Tanner cycles in $\HH$, each denoted by $\Ct$.  
We say that $\Ct$ is contained in $\cU$, and write $\cU \in \Ct$.  
We use $\cU$ also to denote the set of all such Tanner cycles contained in the block cycle $\cU$.  
The existence of a closed block cycle implies the presence of a cycle in the corresponding Tanner graph~\cite{myung2006combining,yoshida2019linear}.  
Conversely, if no closed block cycle exists, then the Tanner graph contains no cycles.
To avoid confusion with block cycles, we explicitly refer to cycles in the Tanner graph as \emph{Tanner cycles} throughout this paper.  
We define the \emph{girth} of a given parity-check matrix as the length of the shortest closed block cycle contained in it.  
This length coincides with the length of the shortest cycle in the Tanner graph.

When column weight $J = 2$, each step in the block cycle alternates between the first and second row blocks, and hence its length is always a multiple of 4.  
In such cases, if we assume that the cycle starts from a block in the upper row and that the first move is in the horizontal direction, the block cycle can be uniquely specified by the sequence of column indices.
Moreover, if every element in the first row block appears in a unique column (i.e., no two entries in the same row block share a column), then the entire block cycle can be fully described by the sequence of permutation elements.
The following example illustrates a concrete instance of a block cycle and demonstrates how the composite function $f_{\cU}$ is evaluated.  
In particular, it shows that the cycle is totally closed, thereby providing a constructive example of a totally closed block cycle (TCBC).
\begin{example}\label{ex:215134_16May25}
We consider the following example of a parity-check matrix $\HH$, which corresponds to the one shown in~\eqref{162406_8Jun25} and illustrated in Example~\ref{ex:021639_17Apr25}.  
\begin{align}
\HH_X = \left(
\begin{array}{ccc||ccc}
f_0 & f_1 & f_2 & g_0 & g_1 & g_2 \\\hline
f_2 & f_0 & f_1 & g_2 & g_0 & g_1
\end{array}
\right).  
\end{align}
Let us consider a block cycle $\cU$ starting from the $(0,0)$-th block in the matrix $\HH$.  
\begin{align}
\cU ={}& 
\overline{f}_0 \to \overline{g}_0 \to g_2 \to f_0 \to \overline{f}_1 \notag\\ &\to \overline{g}_2 \to g_1 \to f_1 \to \overline{f}_2 \notag\\ &\to \overline{g}_1 \to g_0 \to f_2 \to \overline{f}_0.
\end{align}
where overlined symbols indicate blocks in the upper row.  
 Since the entries in the upper row are distinct, each one determines a unique column index.  
This traversal can be compactly expressed by the sequence:
 \begin{align}
\cU= f_0\Rightarrow g_0\Rightarrow f_1\Rightarrow g_2\Rightarrow f_2\Rightarrow g_1\Rightarrow f_0.
 \end{align}
The notation $\to$ indicates the full sequence of block-level transitions, whereas $\Rightarrow$ is used in the compact form, omitting the intermediate transitions between upper and lower rows.  
From the values in~\eqref{103357_10Jun25}, evaluating the composite function $f_{\cU}$ associated with this block cycle yields $f_{\cU}(x) = x$.  
Therefore, this example constitutes a TCBC.  
\end{example}

From the following theorem, it follows that any QC-LDPC code whose parity-check matrix contains  commutative submatrices of size at least $2 \times 3$ has girth at most 12.
\begin{teiri}[\cite{1317123}]\label{144806_5Apr25}
Let $a, b, c, d, e, f \in \mathcal{F}_P$ be mutually commuting elements.  
In particular, this commutativity condition is automatically satisfied when $a, b, c, d, e, f \in \Qc_P$.
Consider the following $2 \times 3$ subarray of a larger permutation array $H$:
\begin{align}
\left( 
\begin{array}{ccc}
a & b & c \\
d & e & f
\end{array}
\right).
\end{align}
Then, the block cycle of length 12  ${\cU}:=a \Rightarrow b \Rightarrow c \Rightarrow a \Rightarrow b \Rightarrow c \Rightarrow a$
is closed. Moreover, it forms a TCBC.  
\end{teiri}
\begin{proof}
Since all the elements commute, the following composition reduces to the identity function:
\begin{align}
f_{\cU}^{-1}= a^{-1}b  e^{-1} f  c^{-1} a  d^{-1} e  b^{-1} c  f^{-1} d   = \mathrm{id}.
\end{align}
Hence, the composite function associated with this block cycle is the identity, confirming that it is a TCBC.
\end{proof}

This theorem implies that when the permutation array is constructed using mutually commuting blocks such as CPMs, the girth of the Tanner graph is upper bounded by $12$.
In~\cite{kasai2025quantumerrorcorrectiongirth16}, a novel construction was proposed using APMs that do not necessarily commute. By exploiting this non-commutativity, the authors succeeded in constructing Calderbank--Shor--Steane (CSS) codes whose Tanner graphs achieve girth $16$.

\section{Conventional Construction}\label{192200_2Jun25}
In the remainder of this paper, we restrict our attention to the case where the column weight is $J = 2$.
In~\cite{komoto2024quantumerrorcorrectionnear}, the construction method for CSS codes based on non-binary LDPC codes-originally proposed for $\Qc_P$-valued arrays in~\cite{6017122}-was generalized to support $\Fc_P$-valued arrays.  
This generalized method \cite{komoto2024quantumerrorcorrectionnear} is referred to as the \emph{conventional construction} in this paper.

Within this framework, we ultimately construct the binary parity-check matrices $H_X$ and $H_Z$ that define a CSS code from two arrays of permutations 
\begin{align}
 \fU := (f_0, \ldots, f_{L/2-1}) \tAND \gU := (g_0, \ldots, g_{L/2-1}) 
\end{align}
in $\Fc_P^{L/2}$. As intermediate steps, we generate auxiliary binary matrices $\HH_X$ and $\HH_Z$, along with non-binary matrices $H_\Gamma$ and $H_\Delta$, which are subsequently used in the decoding process. Since this construction serves as the foundation for the proposed code construction presented in Section~\ref{sec:150611_18May25}, we provide a detailed overview below.

Section~\ref{101943_12Jun25} describes a construction of orthogonal binary matrix pairs $(\HH_X, \HH_Z)$.  
Section~\ref{sec:231308_5Jun25} explains how to convert these binary matrices into orthogonal $\Fb_q$-valued matrix pairs $(H_\Gamma, H_\Delta)$.  
Finally, Section~\ref{092546_7Jun25} presents a construction of binary matrix pairs $(H_X, H_Z)$ that are equivalent to $(H_\Gamma, H_\Delta)$.
\subsection{Construction of $\HH_X$ and $\HH_Z$}
\label{101943_12Jun25}
The following requirement on $\fU$ and $\gU$ is used to ensure that the resulting matrices $\HH_X$ and $\HH_Z$ are orthogonal \cite{komoto2024quantumerrorcorrectionnear}.
\begin{yosei}\label{yosei:orthogonal}
We require that the following commutativity condition holds:
\begin{align}
  g_{\ell-j} \, f_{k-\ell} = f_{k-\ell} \, g_{\ell-j} \quad \text{for all } \ell \in [L/2],\ j,k \in [J]. \label{162154_15Feb25}
\end{align}
 Here, the indices of $f$ and $g$ are understood modulo $L/2$. 
\end{yosei}
This condition~\eqref{162154_15Feb25} is equivalently expressed by any of the following three statements.  
\begin{align}
& f_{\ell - j} \, g_{k - \ell} = g_{k - \ell} \, f_{\ell - j} \quad \text{for all } \ell \in [L/2],\ j, k \in [J], \label{192555_14Jun25}\\
&\begin{aligned}&\text{$f_\ell$ commutes with $g_{- \ell + j}$ for all $\ell \in [L/2]$}\\&\text{and $j \in \{0,\pm 1,\ldots,\pm (J-1)\}$,}\end{aligned} \label{191912_14Jun25}\\
 &\begin{aligned}&\text{$g_\ell$ commutes with $f_{- \ell + j}$ for all $\ell \in [L/2]$}\\&\text{and $j \in \{0,\pm 1,\ldots,\pm (J-1)\}$.}\end{aligned}
\end{align}

This commutativity condition guarantees that each pair of corresponding blocks in $\HH_X$ and $\HH_Z$ ensures the orthogonality.
As shown in Table~\ref{fig:commutativity_conditions}, the table lists the $(\fU, \gU)$ pairs for which commutativity is required when $L = 6,8$ and $10$. 
A value of ``1'' indicates that the pair must satisfy the commutativity condition $f_i g_j = g_j f_i$, while ``--'' denotes that no such requirement is imposed.
When $L = 6$, it is known that all $f_i$ and $g_j$ commute.
\begin{table*}[htbp]\label{table:145006_12Jun25}
  \centering
  \renewcommand{\arraystretch}{1.2}
  \caption{Commutativity matrices for various values of $L$ with fixed $J=2$.}
  \label{fig:commutativity_conditions}
  \begin{subfigure}[t]{0.3\textwidth}
    \caption{$J=2$, $L = 6$}
    \centering
    \[
    \begin{array}{c|ccc}
        & g_0 & g_1 & g_2 \\
        \hline
        f_0 & 1 & 1 &  1 \\
        f_1 & 1 & 1 &  1 \\
        f_2 & 1 & 1 &  1 \\
    \end{array}
    \]
  \end{subfigure}
  \begin{subfigure}[t]{0.3\textwidth}
    \caption{$J=2$, $L = 8$}
    \centering
    \[
    \begin{array}{c|cccc}
        & g_0 & g_1 & g_2 & g_3 \\
        \hline
        f_0 & 1 & 1 & - & 1 \\
        f_1 & 1 & - & 1 & 1 \\
        f_2 & - & 1 & 1 & 1 \\
        f_3 & 1 & 1 & 1 & - \\
    \end{array}
    \]
  \end{subfigure}
  \begin{subfigure}[t]{0.3\textwidth}
    \caption{$J=2$, $L = 10$}
    \centering
    \[
    \begin{array}{c|ccccc}
        & g_0 & g_1 & g_2 & g_3 & g_4 \\
        \hline
        f_0 & 1 & 1 & - & - & 1 \\
        f_1 & 1 & - & - & 1 & 1 \\
        f_2 & - & - & 1 & 1 & 1 \\
        f_3 & - & 1 & 1 & 1 & - \\
        f_4 & 1 & 1 & 1 & - & - \\
    \end{array}
    \]
  \end{subfigure}
\end{table*}

We show that the matrices $\HH_X$ and $\HH_Z$, defined below using $\fU$ and $\gU$ that satisfy Requirement \ref{yosei:orthogonal}, are orthogonal; that is, $\HH_X \HH_Z^\top = 0$. Let us begin by defining these matrices.
\begin{df}{Construction of $\HH_X$ and $\HH_Z$}\label{141246_12Jun25}
Let $f_0, \dots, f_{L/2-1}$ and $g_0, \dots, g_{L/2-1}$ be elements of $\mathcal{F}_P$, and assume that the indices are extended cyclically, i.e., $f_k = f_{k \bmod (L/2)}$ and similarly for $g_k$.

We first define the following four types of $\mathcal{F}_P$-valued arrays of length $L/2$:
\begin{align}
  \begin{array}{llll}
  &(\HH^{(\Lrm,j)}_X)_{\ell} = f_{ \ell-j}, \quad &(\HH^{(\Rrm,j)}_X)_{\ell} = g_{ \ell-j}, 
\\&(\HH^{(\Lrm,j)}_Z)_{\ell} = g^{-1}_{-(\ell-j)}, \quad  &(\HH^{(\Rrm,j)}_Z)_{\ell} = f^{-1}_{-(\ell-j)},
 \end{array}
\end{align}
for $j\in [L/2]$ and $\ell \in [L/2]$.
The arrays $\HH_X^{(\Lrm,j)}$, $\HH_X^{(\Rrm,j)}$, $\HH_Z^{(\Lrm,j)}$, and $\HH_Z^{(\Rrm,j)}$ serve as building blocks for constructing the permutation arrays $\hat{H}_X$ and $\hat{H}_Z$ used in our CSS code.

Next, we define the $\Fc_P$-valued $J \times L$ array $\HH_X$ as the horizontal concatenation of the left and right halves:
 \begin{align}
  \HH_X &\defeq (\HH^{(\Lrm)}_X\mid \HH^{(\Rrm)}_X) 
\\&\defeq
\left(\begin{array}{l|l} 
\HH^{(\Lrm,0)}_X&\HH^{(\Rrm,0)}_X\\
\HH^{(\Lrm,1)}_X&\HH^{(\Rrm,1)}_X\\
 \vdots&\vdots\\
\HH^{(\Lrm,J-1)}_X&\HH^{(\Rrm,J-1)}_X
 \end{array}
\right)
\end{align}
Alternatively, letting $N := LP$ and $M := JP$, the matrix $\HH_X$ can equivalently be regarded as a binary matrix of size $M \times N$.
The matrix $\HH_Z$ is defined analogously, using the row vectors $\HH^{(\Lrm,j)}_Z$ and $\HH^{(\Rrm,j)}_Z$ in place of their $X$-counterparts.
\end{df}

Alternatively, if we define the $j$-th row of $\HH_X$ for $j\in [J]$ as
\begin{align}
\HH^{(j)}_X \coloneqq (\HH^{(\Lrm,j)}_X\ |\ \HH^{(\Rrm,j)}_X),\label{025034_18May25}
\end{align}
then the matrix can be compactly expressed as
\begin{align}
\HH_X =
\left(
\begin{array}{c}
\HH^{(0)}_X \\
\HH^{(1)}_X \\
\vdots \\
\HH^{(J-1)}_X
\end{array}
\right).\label{113347_20May25}
\end{align}

Strictly speaking, the index range $j \in [J]$ is sufficient for the definition of $\HH_X$.
However, for the case $L=6$ considered later, we make use of $\HH_X^{(2)}$, which coincides with $\HH_X^{(L/2 - 1)}$
due to the $L/2$-periodicity of the permutation indices $f_i$ and $g_j$.
Therefore, we adopt the extended index range $j \in [L/2]$ in the current definition.

\begin{example}\label{ex:021639_17Apr25}
Let $P = 8$, $J = 2$, and $L = 6$.
Using the APMs or CPMs $f_0, \ldots, f_2$ and $g_0, \ldots, g_2$ specified as follows:
\begin{align}
\fU &= (5X + 7,\quad 5X + 3,\quad 1X + 6), \\
\gU &= (5X + 7,\quad 5X + 5,\quad 5X + 7),
\end{align}
we compute their inverses:
\begin{align}
(f_0^{-1},\ f_1^{-1},\ f_2^{-1}) &= (5X + 5,\quad 5X + 1,\quad 1X + 2), \\
(g_0^{-1},\ g_1^{-1},\ g_2^{-1}) &= (5X + 5,\quad 5X + 7,\quad 5X + 5).
\end{align}

Based on these permutations, we define the following $J \times L$ protograph-based permutation array $\HH_X$:
\begin{figure*}[!t]
\begin{align}
\HH_X 
&= \left(
\begin{array}{lll||lll}
f_0 & f_1 & f_2 & g_0 & g_1 & g_2 \\\hline
f_{-1} & f_0 & f_1 & g_{-1} & g_0 & g_1
\end{array}
\right) \\
&= \left(
\begin{array}{ccc||ccc}
f_0 & f_1 & f_2 & g_0 & g_1 & g_2 \\\hline
f_2 & f_0 & f_1 & g_2 & g_0 & g_1
\end{array}
\right) \label{162406_8Jun25}\\
&= \left(
\begin{array}{ccc||ccc}
5X + 7 & 5X + 3 & 1X + 6 & 5X + 7 & 5X + 5 & 5X + 7 \\\hline
1X + 6 & 5X + 7 & 5X + 3 & 5X + 7 & 5X + 7 & 5X + 5
\end{array}
\right).\label{103357_10Jun25}
\end{align}
\end{figure*}

The corresponding matrix $\HH_Z$ is defined analogously by replacing each block in $\HH_X$ with its inverse, arranged according to the shifted index pattern:
\begin{figure*}[!t]
\begin{align}
\HH_Z 
&= \left(
\begin{array}{ccc||ccc}
g_0^{-1} & g_{-1}^{-1} & g_{-2}^{-1} & f_0^{-1} & f_{-1}^{-1} & f_{-2}^{-1} \\\hline
g_1^{-1} & g_0^{-1} & g_{-1}^{-1} & f_1^{-1} & f_0^{-1} & f_{-1}^{-1}
\end{array}
\right) \\
&= \left(
\begin{array}{ccc||ccc}
g_0^{-1} & g_{2}^{-1} & g_{1}^{-1} & f_0^{-1} & f_{2}^{-1} & f_{1}^{-1} \\\hline
g_1^{-1} & g_0^{-1} & g_{2}^{-1} & f_1^{-1} & f_0^{-1} & f_{2}^{-1}
\end{array}
\right) \\
&= \left(
\begin{array}{ccc||ccc}
5X + 5 & 5X + 5 & 5X + 7 & 5X + 5 & 1X + 2 & 5X + 1 \\\hline
5X + 7 & 5X + 5 & 5X + 5 & 5X + 1 & 5X + 5 & 1X + 2
\end{array}
\right).
\end{align}
\end{figure*}

Let $M = JP = 16$ and $N = LP = 48$. The $M \times N$ binary matrices corresponding to $\HH_X$ and $\HH_Z$ are obtained by replacing each affine permutation with its associated $P \times P$ binary permutation matrix. The resulting binary matrices are shown below. 
The zero entries are represented as blanks.  
To facilitate understanding of the correspondence with the Tanner or factor graphs introduced in later sections, each block in the matrix is color-coded to match the corresponding edge in the graph. These matrices contain no cycles of length~4, but do contain cycles of length~8.

\input{matrix}

The binary matrices corresponding to $\HH_X^{(2)}$ and $\HH_Z^{(2)}$ are shown below.
Each affine permutation is expanded into a $P \times P$ binary permutation matrix with $P = 8$.
\begin{figure*}[!t]
\begin{align}
\HH_{X}^{(2)} 
& =\left(
 \begin{array}{c@{\hspace{6mm}}c@{\hspace{6mm}}c@{\hspace{3.5mm}}||@{\hspace{3mm}}c@{\hspace{6mm}}c@{\hspace{6mm}}c@{\hspace{3mm}}}
    f_{1} & f_{2} & f_0 & g_{1} & g_{2} & g_0 
 \end{array}\right) \\ 
&=
\left(\begin{array}{ccc||ccc}
5X+3&1X+6&5X+7&5X+5&5X+7&5X+7
\end{array} 
\right)\\
&=
\left(\begin{array}{c|c|c||c|c|c}
\phantom{0}1\phantom{0}\phantom{0}\phantom{0}\phantom{0}\phantom{0}\phantom{0}&\phantom{0}\phantom{0}1\phantom{0}\phantom{0}\phantom{0}\phantom{0}\phantom{0}&\phantom{0}\phantom{0}\phantom{0}\phantom{0}\phantom{0}1\phantom{0}\phantom{0}&\phantom{0}\phantom{0}\phantom{0}\phantom{0}\phantom{0}\phantom{0}\phantom{0}1&\phantom{0}\phantom{0}\phantom{0}\phantom{0}\phantom{0}1\phantom{0}\phantom{0}&\phantom{0}\phantom{0}\phantom{0}\phantom{0}\phantom{0}1\phantom{0}\phantom{0}\\
\phantom{0}\phantom{0}\phantom{0}\phantom{0}\phantom{0}\phantom{0}1\phantom{0}&\phantom{0}\phantom{0}\phantom{0}1\phantom{0}\phantom{0}\phantom{0}\phantom{0}&\phantom{0}\phantom{0}1\phantom{0}\phantom{0}\phantom{0}\phantom{0}\phantom{0}&\phantom{0}\phantom{0}\phantom{0}\phantom{0}1\phantom{0}\phantom{0}\phantom{0}&\phantom{0}\phantom{0}1\phantom{0}\phantom{0}\phantom{0}\phantom{0}\phantom{0}&\phantom{0}\phantom{0}1\phantom{0}\phantom{0}\phantom{0}\phantom{0}\phantom{0}\\
\phantom{0}\phantom{0}\phantom{0}1\phantom{0}\phantom{0}\phantom{0}\phantom{0}&\phantom{0}\phantom{0}\phantom{0}\phantom{0}1\phantom{0}\phantom{0}\phantom{0}&\phantom{0}\phantom{0}\phantom{0}\phantom{0}\phantom{0}\phantom{0}\phantom{0}1&\phantom{0}1\phantom{0}\phantom{0}\phantom{0}\phantom{0}\phantom{0}\phantom{0}&\phantom{0}\phantom{0}\phantom{0}\phantom{0}\phantom{0}\phantom{0}\phantom{0}1&\phantom{0}\phantom{0}\phantom{0}\phantom{0}\phantom{0}\phantom{0}\phantom{0}1\\
1\phantom{0}\phantom{0}\phantom{0}\phantom{0}\phantom{0}\phantom{0}\phantom{0}&\phantom{0}\phantom{0}\phantom{0}\phantom{0}\phantom{0}1\phantom{0}\phantom{0}&\phantom{0}\phantom{0}\phantom{0}\phantom{0}1\phantom{0}\phantom{0}\phantom{0}&\phantom{0}\phantom{0}\phantom{0}\phantom{0}\phantom{0}\phantom{0}1\phantom{0}&\phantom{0}\phantom{0}\phantom{0}\phantom{0}1\phantom{0}\phantom{0}\phantom{0}&\phantom{0}\phantom{0}\phantom{0}\phantom{0}1\phantom{0}\phantom{0}\phantom{0}\\
\phantom{0}\phantom{0}\phantom{0}\phantom{0}\phantom{0}1\phantom{0}\phantom{0}&\phantom{0}\phantom{0}\phantom{0}\phantom{0}\phantom{0}\phantom{0}1\phantom{0}&\phantom{0}1\phantom{0}\phantom{0}\phantom{0}\phantom{0}\phantom{0}\phantom{0}&\phantom{0}\phantom{0}\phantom{0}1\phantom{0}\phantom{0}\phantom{0}\phantom{0}&\phantom{0}1\phantom{0}\phantom{0}\phantom{0}\phantom{0}\phantom{0}\phantom{0}&\phantom{0}1\phantom{0}\phantom{0}\phantom{0}\phantom{0}\phantom{0}\phantom{0}\\
\phantom{0}\phantom{0}1\phantom{0}\phantom{0}\phantom{0}\phantom{0}\phantom{0}&\phantom{0}\phantom{0}\phantom{0}\phantom{0}\phantom{0}\phantom{0}\phantom{0}1&\phantom{0}\phantom{0}\phantom{0}\phantom{0}\phantom{0}\phantom{0}1\phantom{0}&1\phantom{0}\phantom{0}\phantom{0}\phantom{0}\phantom{0}\phantom{0}\phantom{0}&\phantom{0}\phantom{0}\phantom{0}\phantom{0}\phantom{0}\phantom{0}1\phantom{0}&\phantom{0}\phantom{0}\phantom{0}\phantom{0}\phantom{0}\phantom{0}1\phantom{0}\\
\phantom{0}\phantom{0}\phantom{0}\phantom{0}\phantom{0}\phantom{0}\phantom{0}1&1\phantom{0}\phantom{0}\phantom{0}\phantom{0}\phantom{0}\phantom{0}\phantom{0}&\phantom{0}\phantom{0}\phantom{0}1\phantom{0}\phantom{0}\phantom{0}\phantom{0}&\phantom{0}\phantom{0}\phantom{0}\phantom{0}\phantom{0}1\phantom{0}\phantom{0}&\phantom{0}\phantom{0}\phantom{0}1\phantom{0}\phantom{0}\phantom{0}\phantom{0}&\phantom{0}\phantom{0}\phantom{0}1\phantom{0}\phantom{0}\phantom{0}\phantom{0}\\
\phantom{0}\phantom{0}\phantom{0}\phantom{0}1\phantom{0}\phantom{0}\phantom{0}&\phantom{0}1\phantom{0}\phantom{0}\phantom{0}\phantom{0}\phantom{0}\phantom{0}&1\phantom{0}\phantom{0}\phantom{0}\phantom{0}\phantom{0}\phantom{0}\phantom{0}&\phantom{0}\phantom{0}1\phantom{0}\phantom{0}\phantom{0}\phantom{0}\phantom{0}&1\phantom{0}\phantom{0}\phantom{0}\phantom{0}\phantom{0}\phantom{0}\phantom{0}&1\phantom{0}\phantom{0}\phantom{0}\phantom{0}\phantom{0}\phantom{0}\phantom{0}
\end{array}\right),
\end{align}
\end{figure*}
\begin{figure*}[!t]
\begin{align}
\HH_{Z}^{(2)} 
&=
 \left(
 \begin{array}{lll||lll}
 g_{2}^{-1} & g_{1}^{-1} & g_0^{-1} & f_{2}^{-1} & f_{1}^{-1} & f_0^{-1}
 \end{array}
 \right)
\\&=\left(\begin{array}{ccc||ccc}
   5X+5&5X+7&5X+5&1X+2&5X+1&5X+5
\end{array}
\right)\\
&=
\left(\begin{array}{c|c|c||c|c|c}
 \phantom{0}\phantom{0}\phantom{0}\phantom{0}\phantom{0}\phantom{0}\phantom{0}1&\phantom{0}\phantom{0}\phantom{0}\phantom{0}\phantom{0}1\phantom{0}\phantom{0}&\phantom{0}\phantom{0}\phantom{0}\phantom{0}\phantom{0}\phantom{0}\phantom{0}1&\phantom{0}\phantom{0}\phantom{0}\phantom{0}\phantom{0}\phantom{0}1\phantom{0}&\phantom{0}\phantom{0}\phantom{0}1\phantom{0}\phantom{0}\phantom{0}\phantom{0}&\phantom{0}\phantom{0}\phantom{0}\phantom{0}\phantom{0}\phantom{0}\phantom{0}1\\
 \phantom{0}\phantom{0}\phantom{0}\phantom{0}1\phantom{0}\phantom{0}\phantom{0}&\phantom{0}\phantom{0}1\phantom{0}\phantom{0}\phantom{0}\phantom{0}\phantom{0}&\phantom{0}\phantom{0}\phantom{0}\phantom{0}1\phantom{0}\phantom{0}\phantom{0}&\phantom{0}\phantom{0}\phantom{0}\phantom{0}\phantom{0}\phantom{0}\phantom{0}1&1\phantom{0}\phantom{0}\phantom{0}\phantom{0}\phantom{0}\phantom{0}\phantom{0}&\phantom{0}\phantom{0}\phantom{0}\phantom{0}1\phantom{0}\phantom{0}\phantom{0}\\
 \phantom{0}1\phantom{0}\phantom{0}\phantom{0}\phantom{0}\phantom{0}\phantom{0}&\phantom{0}\phantom{0}\phantom{0}\phantom{0}\phantom{0}\phantom{0}\phantom{0}1&\phantom{0}1\phantom{0}\phantom{0}\phantom{0}\phantom{0}\phantom{0}\phantom{0}&1\phantom{0}\phantom{0}\phantom{0}\phantom{0}\phantom{0}\phantom{0}\phantom{0}&\phantom{0}\phantom{0}\phantom{0}\phantom{0}\phantom{0}1\phantom{0}\phantom{0}&\phantom{0}1\phantom{0}\phantom{0}\phantom{0}\phantom{0}\phantom{0}\phantom{0}\\
 \phantom{0}\phantom{0}\phantom{0}\phantom{0}\phantom{0}\phantom{0}1\phantom{0}&\phantom{0}\phantom{0}\phantom{0}\phantom{0}1\phantom{0}\phantom{0}\phantom{0}&\phantom{0}\phantom{0}\phantom{0}\phantom{0}\phantom{0}\phantom{0}1\phantom{0}&\phantom{0}1\phantom{0}\phantom{0}\phantom{0}\phantom{0}\phantom{0}\phantom{0}&\phantom{0}\phantom{0}1\phantom{0}\phantom{0}\phantom{0}\phantom{0}\phantom{0}&\phantom{0}\phantom{0}\phantom{0}\phantom{0}\phantom{0}\phantom{0}1\phantom{0}\\
 \phantom{0}\phantom{0}\phantom{0}1\phantom{0}\phantom{0}\phantom{0}\phantom{0}&\phantom{0}1\phantom{0}\phantom{0}\phantom{0}\phantom{0}\phantom{0}\phantom{0}&\phantom{0}\phantom{0}\phantom{0}1\phantom{0}\phantom{0}\phantom{0}\phantom{0}&\phantom{0}\phantom{0}1\phantom{0}\phantom{0}\phantom{0}\phantom{0}\phantom{0}&\phantom{0}\phantom{0}\phantom{0}\phantom{0}\phantom{0}\phantom{0}\phantom{0}1&\phantom{0}\phantom{0}\phantom{0}1\phantom{0}\phantom{0}\phantom{0}\phantom{0}\\
 1\phantom{0}\phantom{0}\phantom{0}\phantom{0}\phantom{0}\phantom{0}\phantom{0}&\phantom{0}\phantom{0}\phantom{0}\phantom{0}\phantom{0}\phantom{0}1\phantom{0}&1\phantom{0}\phantom{0}\phantom{0}\phantom{0}\phantom{0}\phantom{0}\phantom{0}&\phantom{0}\phantom{0}\phantom{0}1\phantom{0}\phantom{0}\phantom{0}\phantom{0}&\phantom{0}\phantom{0}\phantom{0}\phantom{0}1\phantom{0}\phantom{0}\phantom{0}&1\phantom{0}\phantom{0}\phantom{0}\phantom{0}\phantom{0}\phantom{0}\phantom{0}\\
 \phantom{0}\phantom{0}\phantom{0}\phantom{0}\phantom{0}{\setlength{\fboxsep}{1pt}\colorbox{red!71!black}{\textcolor{white}{1}}}\phantom{0}\phantom{0}&\phantom{0}\phantom{0}\phantom{0}{\setlength{\fboxsep}{1pt}\colorbox{red!71!black}{\textcolor{white}{1}}}\phantom{0}\phantom{0}\phantom{0}\phantom{0}&\phantom{0}\phantom{0}\phantom{0}\phantom{0}\phantom{0}{\setlength{\fboxsep}{1pt}\colorbox{red!71!black}{\textcolor{white}{1}}}\phantom{0}\phantom{0}&\phantom{0}\phantom{0}\phantom{0}\phantom{0}{\setlength{\fboxsep}{1pt}\colorbox{red!71!black}{\textcolor{white}{1}}}\phantom{0}\phantom{0}\phantom{0}&\phantom{0}{\setlength{\fboxsep}{1pt}\colorbox{red!71!black}{\textcolor{white}{1}}}\phantom{0}\phantom{0}\phantom{0}\phantom{0}\phantom{0}\phantom{0}&\phantom{0}\phantom{0}\phantom{0}\phantom{0}\phantom{0}{\setlength{\fboxsep}{1pt}\colorbox{red!71!black}{\textcolor{white}{1}}}\phantom{0}\phantom{0}\\
 \phantom{0}\phantom{0}1\phantom{0}\phantom{0}\phantom{0}\phantom{0}\phantom{0}&1\phantom{0}\phantom{0}\phantom{0}\phantom{0}\phantom{0}\phantom{0}\phantom{0}&\phantom{0}\phantom{0}1\phantom{0}\phantom{0}\phantom{0}\phantom{0}\phantom{0}&\phantom{0}\phantom{0}\phantom{0}\phantom{0}\phantom{0}1\phantom{0}\phantom{0}&\phantom{0}\phantom{0}\phantom{0}\phantom{0}\phantom{0}\phantom{0}1\phantom{0}&\phantom{0}\phantom{0}1\phantom{0}\phantom{0}\phantom{0}\phantom{0}\phantom{0}
\end{array}\right).
 \end{align}
\end{figure*}
\end{example}

This construction generalizes the original Hagiwara--Imai construction~\cite{4557323,6017122} to a more flexible framework involving general PMs~\cite{komoto2024quantumerrorcorrectionnear}, by extending the class of permutations from $\Qc_P$ to $\Fc_P$.
Under the condition specified in Requirement~\ref{yosei:orthogonal}, the matrices $\HH_X, \HH_Z \in \mathbb{F}_2^{JP \times LP}$ are orthogonal, i.e.,
\begin{align}
 \HH_X \HH_Z^\top = O.\label{222558_20May25}
\end{align}
This orthogonality follows from the structure of the circulant blocks. Specifically, the $(j,k)$-th block of the product is given by
\begin{align*}
(\HH_X \HH_Z^\top)_{j,k}
&= \sum_{\ell} F_{\ell-j} G_{k-\ell} + \sum_{\ell} G_{\ell-j} F_{k-\ell} \\
&= \sum_{\ell} F_{\ell-j} G_{k-\ell} + \sum_{\ell} F_{k-\ell} G_{\ell-j} \\
&= \sum_{\ell} F_{\ell} G_{k-j-\ell} + \sum_{\ell} F_{-\ell} G_{\ell-j+k} \\
&= \sum_{\ell} F_{\ell} G_{k-j-\ell} + \sum_{\ell} F_{\ell} G_{k-j-\ell} \\
&= O,
\end{align*}
In all summations above, the index $\ell$ ranges over $[L/2]$.
The first equality follows from Requirement~\ref{yosei:orthogonal}, and the final equality holds since the two sums are identical over $\mathbb{F}_2$ and thus cancel out.  Therefore, $\HH_X$ and $\HH_Z$ are orthogonal by construction.

\subsection{Construction of Non-binary Parity-Check Matrices $H_\Gamma$ and $H_\Delta$}\label{sec:231308_5Jun25}
In the previous section, we constructed a pair of orthogonal $(J, L)$-regular binary LDPC matrices, $\HH_X$ and $\HH_Z$.  
Let $q = 2^e$ for a positive integer $e$.  
In this section, we focus on the case $J = 2$ and outline the conventional method~\cite{komoto2024quantumerrorcorrectionnear}, which generalizes the earlier approach by Kasai et al.~\cite{6017122} for constructing a pair of orthogonal LDPC matrices $H_\Gamma$ and $H_\Delta$ over $\Fb_q$.  
These matrices share the same support as $\HH_X$ and $\HH_Z$, and serve as the basis for the modified construction in Section~\ref{sec:150611_18May25}.  

We consider $H_\Gamma, H_\Delta \in \mathbb{F}_q^{M \times N}$, where $M=JP$ and $N=LP$, with the same support as $\HH_X$ and $\HH_Z$, respectively.
The logarithmic construction below assumes that the support of each row of either binary matrix selects a simple Tanner cycle of length $2L$ in the other matrix. This condition is established for the proposed $J=2$, $L=6$ construction in Theorem~\ref{053905_28Apr25}.
Our goal is to randomly generate such matrices satisfying the orthogonality condition  
\begin{align}
H_\Gamma H_\Delta^\top = O. \label{212247_11Feb25}
\end{align}
When $J \ge 3$, this condition leads to a system of quadratic equations over $\mathbb{F}_q$, which is generally intractable.  
However, for $J = 2$, the problem can be decomposed into two linear subproblems, allowing for efficient solution via linear algebra.

We denote the entries of the matrices as follows:  
\begin{align}
H_\Gamma = (\gamma_{ij}), \quad H_\Delta = (\delta_{ij}).  
\end{align}
To express the algorithm concisely, we define a one-dimensional vector representation in place of $(\gamma_{ij})$ and $(\delta_{ij})$.  
Let us define a vector representation $\gammaU = (\gamma_0, \ldots, \gamma_{2N - 1}) \in \mathbb{F}_q^{2N}$ associated with the parity-check matrix $H_\Gamma$ as follows.  
For each position $(i,j)$ such that $(\HH_X)_{ij} \ne 0$, we define:
\begin{align}
\gamma_{i,j} =
\begin{cases}
\gamma_j & \text{if } i < P, \\
\gamma_{j + PL} & \text{if } i \ge P.
\end{cases}
\end{align}
Since $H_\Gamma$ contains exactly $2N$ nonzero entries, the length of the vector $\gammaU$ is also $2N$.
The corresponding vector $\deltaU$ associated with $H_\Delta$ is defined analogously.

For example, when $P = 8$ and $L = 6$, we have:
\begin{align*}
\gamma_{1,12} &= \gamma_{12}, \\
\gamma_{9,12} &= \gamma_{60}.
\end{align*}

Fix a primitive element $\alpha$ of the finite field $\mathbb{F}_q$.
Then, every nonzero element $\gamma \in \mathbb{F}_q^\times$ can be uniquely expressed as $\gamma = \alpha^i$ for some $i \in [q - 1]$.
We define the discrete logarithm of $\gamma$ as $\log_\alpha \gamma := i$, where the base $\alpha$ is omitted when there is no ambiguity.
Accordingly, we define the \emph{logarithmic vector representation} of a vector $\gammaU = (\gamma_0, \ldots, \gamma_{2N-1}) \in \mathbb{F}_q^{2N}$ as
\begin{align}
\log \gammaU := (\log \gamma_0, \ldots, \log \gamma_{2N-1}),\label{235116_14Jun25}
\end{align}
where each $\log \gamma_i$ is an integer in $[q - 1]$ such that $\gamma_i = \alpha^{\log \gamma_i}$.

We wish to determine vectors $\gammaU$ and $\deltaU$ that satisfy the following conditions.  
To ensure the existence of matrices $H_\Gamma$ and $H_\Delta$ that satisfy the orthogonality condition~\eqref{212247_11Feb25}, the vector $\log \gammaU$ must satisfy the linear congruence relation:
\begin{align}
A_{01} \log \gammaU &= \zeroU \pmod{q - 1}. \label{015905_13Feb25} 
\end{align}
Here, the coefficient matrix $A_{01}\in\{0, \pm 1\}^{2P \times 2N}$ is defined as
\begin{align}
A_{01} :&= 
\left(
\setlength{\arraycolsep}{2pt}\begin{array}{c|c|c|c}
-\HH_Z^{(\Lrm,0)} & \HH_Z^{(\Rrm,0)} & \HH_Z^{(\Lrm,0)} & -\HH_Z^{(\Rrm,0)}\\\hline
-\HH_Z^{(\Lrm,1)} & \HH_Z^{(\Rrm,1)} & \HH_Z^{(\Lrm,1)} & -\HH_Z^{(\Rrm,1)}
\end{array} 
\right)\\
&=(
\setlength{\arraycolsep}{2pt}\begin{array}{c|c|c|c}
-\HH_Z^{(\Lrm)} & \HH_Z^{(\Rrm)} & \HH_Z^{(\Lrm)} & -\HH_Z^{(\Rrm)}
\end{array} 
). 
\end{align}
The proof of this result is given in Appendix~\ref{224018_19May25}.

We first solve the system of congruences in~\eqref{015905_13Feb25} and randomly select one solution $\log \gammaU$ from the solution space.  
This uniquely determines the corresponding vector $\gammaU$ and hence the matrix $H_\Gamma$.
Next, for the fixed matrix $H_\Gamma$, the orthogonality condition \eqref{212247_11Feb25}: $H_\Gamma H_\Delta^\top = O$ 
can be viewed as a system of linear equations in the entries of $H_\Delta$.  
We solve this system and then randomly generate one valid solution for $H_\Delta$ with the prescribed support from the resulting solution space.

Although \eqref{212247_11Feb25} exhibits a clear symmetry between $\gammaU$ and $\deltaU$,  
\eqref{015905_13Feb25} does not explicitly reflect this symmetry.  
Some readers may find this asymmetry puzzling.  
However, when both \eqref{212247_11Feb25} and \eqref{015905_13Feb25} are satisfied,  
it can be shown that the following equation-an analogue of \eqref{015905_13Feb25} with respect to $\deltaU$-also holds automatically:
\begin{align}
B_{01} \log \deltaU &= \zeroU \pmod{q - 1},\label{095733_27May25}
\end{align}
where the coefficient matrix $B_{01} \in \{0, \pm 1\}^{2P \times 2N}$ is defined as
\begin{align}
B_{01} :&= 
\left(
\setlength{\arraycolsep}{2pt}\begin{array}{c|c|c|c}
-\HH_X^{(\Lrm,0)} & \HH_X^{(\Rrm,0)} & \HH_X^{(\Lrm,0)} & -\HH_X^{(\Rrm,0)}\\\hline
-\HH_X^{(\Lrm,1)} & \HH_X^{(\Rrm,1)} & \HH_X^{(\Lrm,1)} & -\HH_X^{(\Rrm,1)}
\end{array} 
\right)\\
&=(
\setlength{\arraycolsep}{2pt}\begin{array}{c|c|c|c}
-\HH_X^{(\Lrm)} & \HH_X^{(\Rrm)} & \HH_X^{(\Lrm)} & -\HH_X^{(\Rrm)}
\end{array}
).\label{152405_18May25}
\end{align}
This reveals that, although the symmetry between $\gammaU$ and $\deltaU$ is not manifest in the form of \eqref{015905_13Feb25}, it is nevertheless preserved at a structural level.  
The pair of equations \eqref{015905_13Feb25} and \eqref{095733_27May25} together reestablish the symmetry, confirming that the treatment of $\gammaU$ and $\deltaU$ is fundamentally dual.

\begin{example}\label{035757_2Jun25}
We continue with the previous example, now extending to the finite field $\mathbb{F}_{q}$ with $q=2^e, e = 8$.  
Let $H_\Gamma$ and $H_\Delta$ be parity-check matrices defined over $\mathbb{F}_{q}$, constructed to satisfy the orthogonality condition $H_\Gamma H_\Delta^\top = 0$.
One possible pair of such matrices is given as follows (only nonzero entries are shown using hexadecimal notation; blank entries represent zero):
\begin{figure*}[!t]
\begin{align}
& H_\Gamma=
{\scriptsize 
\left(\begin{array}{c@{}c@{}c@{}c@{}c@{}c@{}c@{}c@{}|c@{}c@{}c@{}c@{}c@{}c@{}c@{}c@{}|c@{}c@{}c@{}c@{}c@{}c@{}c@{}c@{}||c@{}c@{}c@{}c@{}c@{}c@{}c@{}c@{}|c@{}c@{}c@{}c@{}c@{}c@{}c@{}c@{}|c@{}c@{}c@{}c@{}c@{}c@{}c@{}c@{}c@{}}
  &  &  &  &  &\tikzmark{0A}{\setlength{\fboxsep}{0.8pt}\colorbox{red!70!black}{\textcolor{white}{$\mathtt{3F}$}}}&  &  &  &\mathtt{4F}&  &  &  &  &  &  &  &  &\mathtt{40}&  &  &  &  &  &  &  &  &  &  &\mathtt{89}&  &  &  &  &  &  &  &  &  &\mathtt{0A}&  &  &  &  &  &\tikzmark{0B}{\setlength{\fboxsep}{0.8pt}\colorbox{red!70!black}{\textcolor{white}{$\mathtt{F1}$}}}&  &  &\\
  &  &\mathtt{D9}&  &  &  &  &  &  &  &  &  &  &  &\mathtt{68}&  &  &  &  &\mathtt{0D}&  &  &  &  &  &  &\mathtt{BC}&  &  &  &  &  &  &  &  &  &\mathtt{4B}&  &  &  &  &  &\mathtt{18}&  &  &  &  &  &\\
  &  &  &  &  &  &  &\mathtt{A9}&  &  &  &\tikzmark{0E}{\setlength{\fboxsep}{0.8pt}\colorbox{red!70!black}{\textcolor{white}{$\mathtt{4C}$}}}&  &  &  &  &  &  &  &  &\mathtt{63}&  &  &  &  &  &  &  &  &  &  &\mathtt{81}&  &\tikzmark{0F}{\setlength{\fboxsep}{0.8pt}\colorbox{red!70!black}{\textcolor{white}{$\mathtt{78}$}}}&  &  &  &  &  &  &  &  &  &  &  &  &  &\mathtt{5C}&\\
  &  &  &  &\mathtt{79}&  &  &  &\mathtt{0B}&  &  &  &  &  &  &  &  &  &  &  &  &\tikzmark{0I}{\setlength{\fboxsep}{0.8pt}\colorbox{red!70!black}{\textcolor{white}{$\mathtt{E1}$}}}&  &  &  &  &  &  &\tikzmark{0J}{\setlength{\fboxsep}{0.8pt}\colorbox{red!70!black}{\textcolor{white}{$\mathtt{89}$}}}&  &  &  &  &  &  &  &  &  &\mathtt{98}&  &  &  &  &  &\mathtt{77}&  &  &  &\\
  &\mathtt{B9}&  &  &  &  &  &  &  &  &  &  &  &\mathtt{FE}&  &  &  &  &  &  &  &  &\mathtt{5F}&  &  &\mathtt{86}&  &  &  &  &  &  &  &  &  &\mathtt{AA}&  &  &  &  &  &\mathtt{0B}&  &  &  &  &  &  &\\
  &  &  &  &  &  &\mathtt{01}&  &  &  &\mathtt{F6}&  &  &  &  &  &  &  &  &  &  &  &  &\mathtt{30}&  &  &  &  &  &  &\mathtt{41}&  &\mathtt{4B}&  &  &  &  &  &  &  &  &  &  &  &  &  &\mathtt{C7}&  &\\
  &  &  &\mathtt{A9}&  &  &  &  &  &  &  &  &  &  &  &\mathtt{D1}&\mathtt{3F}&  &  &  &  &  &  &  &  &  &  &\mathtt{C7}&  &  &  &  &  &  &  &  &  &\mathtt{5A}&  &  &  &  &  &\mathtt{EA}&  &  &  &  &\\
{\setlength{\fboxsep}{0.8pt}\colorbox{blue!70!black}{\textcolor{white}{$\mathtt{C8}$}}}&  &  &  &  &  &  &  &  &  &  &  &{\setlength{\fboxsep}{0.8pt}\colorbox{blue!70!black}{\textcolor{white}{$\mathtt{C8}$}}}&  &  &  &  &{\setlength{\fboxsep}{0.8pt}\colorbox{blue!70!black}{\textcolor{white}{$\mathtt{AA}$}}}&  &  &  &  &  &  &{\setlength{\fboxsep}{0.8pt}\colorbox{blue!70!black}{\textcolor{white}{$\mathtt{EE}$}}}&  &  &  &  &  &  &  &  &  &{\setlength{\fboxsep}{0.8pt}\colorbox{blue!70!black}{\textcolor{white}{$\mathtt{A7}$}}}&  &  &  &  &  &{\setlength{\fboxsep}{0.8pt}\colorbox{blue!70!black}{\textcolor{white}{$\mathtt{5F}$}}}&  &  &  &  &  &  &  &\\
\hline\hline
  &  &\mathtt{BD}&  &  &  &  &  &  &  &  &  &  &\mathtt{68}&  &  &  &\mathtt{3B}&  &  &  &  &  &  &  &  &  &  &  &\mathtt{E7}&  &  &  &  &  &  &  &\mathtt{37}&  &  &  &  &  &  &  &  &  &\mathtt{B1}&\\
  &  &  &\mathtt{61}&  &  &  &  &  &  &\mathtt{88}&  &  &  &  &  &  &  &  &  &  &  &\mathtt{02}&  &  &  &\mathtt{44}&  &  &  &  &  &  &  &\mathtt{80}&  &  &  &  &  &  &  &  &  &\mathtt{D7}&  &  &  &\\
  &  &  &  &\mathtt{CA}&  &  &  &  &  &  &  &  &  &  &\mathtt{5D}&  &  &  &\mathtt{6F}&  &  &  &  &  &  &  &  &  &  &  &\mathtt{F3}&  &  &  &  &  &  &  &\mathtt{C7}&  &\mathtt{43}&  &  &  &  &  &  &\\
  &  &  &  &  &\tikzmark{0L}{\setlength{\fboxsep}{0.8pt}\colorbox{red!70!black}{\textcolor{white}{$\mathtt{8C}$}}}&  &  &  &  &  &  &\mathtt{4E}&  &  &  &\mathtt{D3}&  &  &  &  &  &  &  &  &  &  &  &\tikzmark{0K}{\setlength{\fboxsep}{0.8pt}\colorbox{red!70!black}{\textcolor{white}{$\mathtt{16}$}}}&  &  &  &  &  &  &  &\mathtt{CD}&  &  &  &  &  &  &  &  &  &\mathtt{6D}&  &\\
  &  &  &  &  &  &\mathtt{97}&  &  &\mathtt{BE}&  &  &  &  &  &  &  &  &  &  &  &\tikzmark{0H}{\setlength{\fboxsep}{0.8pt}\colorbox{red!70!black}{\textcolor{white}{$\mathtt{A4}$}}}&  &  &  &\mathtt{2D}&  &  &  &  &  &  &  &\tikzmark{0G}{\setlength{\fboxsep}{0.8pt}\colorbox{red!70!black}{\textcolor{white}{$\mathtt{85}$}}}&  &  &  &  &  &  &  &  &  &\mathtt{B7}&  &  &  &  &\\
  &  &  &  &  &  &  &\mathtt{0E}&  &  &  &  &  &  &\mathtt{A7}&  &  &  &\mathtt{AC}&  &  &  &  &  &  &  &  &  &  &  &\mathtt{92}&  &  &  &  &  &  &  &\mathtt{0C}&  &\mathtt{58}&  &  &  &  &  &  &  &\\
\mathtt{D7}&  &  &  &  &  &  &  &  &  &  &\tikzmark{0D}{\setlength{\fboxsep}{0.8pt}\colorbox{red!70!black}{\textcolor{white}{$\mathtt{F3}$}}}&  &  &  &  &  &  &  &  &  &  &  &\mathtt{7C}&  &  &  &\mathtt{89}&  &  &  &  &  &  &  &\mathtt{08}&  &  &  &  &  &  &  &  &  &\tikzmark{0C}{\setlength{\fboxsep}{0.8pt}\colorbox{red!70!black}{\textcolor{white}{$\mathtt{D4}$}}}&  &  &\\
  &\mathtt{C3}&  &  &  &  &  &  &\mathtt{75}&  &  &  &  &  &  &  &  &  &  &  &\mathtt{30}&  &  &  &\mathtt{4D}&  &  &  &  &  &  &  &\mathtt{87}&  &  &  &  &  &  &  &  &  &\mathtt{5C}&  &  &  &  &  &\\
\end{array}
\right), 
}
\begin{tikzpicture}[remember picture, overlay]
  \draw[thick, red!50, ->] (pic cs:0A) -- (pic cs:0B);
  \draw[thick, red!50, ->] (pic cs:0B) -- (pic cs:0C);
  \draw[thick, red!50, ->] (pic cs:0C) -- (pic cs:0D);
  \draw[thick, red!50, ->] (pic cs:0D) -- (pic cs:0E);
  \draw[thick, red!50, ->] (pic cs:0E) -- (pic cs:0F);
  \draw[thick, red!50, ->] (pic cs:0F) -- (pic cs:0G);
  \draw[thick, red!50, ->] (pic cs:0G) -- (pic cs:0H);
  \draw[thick, red!50, ->] (pic cs:0H) -- (pic cs:0I);
  \draw[thick, red!50, ->] (pic cs:0I) -- (pic cs:0J);
  \draw[thick, red!50, ->] (pic cs:0J) -- (pic cs:0K);
  \draw[thick, red!50, ->] (pic cs:0K) -- (pic cs:0L);
  \draw[thick, red!50, ->] (pic cs:0L) -- (pic cs:0A);
\end{tikzpicture}
\end{align}
\end{figure*}
\begin{figure*}[!t]
\begin{align}
& H_\Delta=
{\scriptsize 
\left(\begin{array}{c@{}c@{}c@{}c@{}c@{}c@{}c@{}c@{}|c@{}c@{}c@{}c@{}c@{}c@{}c@{}c@{}|c@{}c@{}c@{}c@{}c@{}c@{}c@{}c@{}||c@{}c@{}c@{}c@{}c@{}c@{}c@{}c@{}|c@{}c@{}c@{}c@{}c@{}c@{}c@{}c@{}|c@{}c@{}c@{}c@{}c@{}c@{}c@{}c@{}c@{}}
  &  &  &  &  &  &  &\mathtt{ED}&  &  &  &  &  &  &  &\mathtt{AC}&  &  &  &  &  &\mathtt{A6}&  &  &  &  &  &  &  &  &  &\mathtt{16}&  &  &  &  &  &  &\mathtt{EF}&  &  &  &  &\mathtt{93}&  &  &  &  &\\
  &  &  &  &\mathtt{C4}&  &  &  &  &  &  &  &\tikzmark{1A}{\setlength{\fboxsep}{1pt}\colorbox{blue!70!black}{\textcolor{white}{$\mathtt{7C}$}}}&  &  &  &  &  &\mathtt{91}&  &  &  &  &  &  &  &  &  &\mathtt{B4}&  &  &  &  &  &  &  &  &  &  &\mathtt{C7}&\tikzmark{1B}{\setlength{\fboxsep}{1pt}\colorbox{blue!70!black}{\textcolor{white}{$\mathtt{E5}$}}}&  &  &  &  &  &  &  &\\
  &\mathtt{D6}&  &  &  &  &  &  &  &\mathtt{78}&  &  &  &  &  &  &  &  &  &  &  &  &  &\mathtt{2E}&  &\mathtt{0A}&  &  &  &  &  &  &\mathtt{13}&  &  &  &  &  &  &  &  &  &  &  &  &\mathtt{D5}&  &  &\\
  &  &  &  &  &  &\mathtt{B1}&  &  &  &  &  &  &  &\mathtt{5C}&  &  &  &  &  &\mathtt{D8}&  &  &  &  &  &  &  &  &  &\mathtt{71}&  &  &\mathtt{C3}&  &  &  &  &  &  &  &  &\mathtt{AC}&  &  &  &  &  &\\
       &  &  &\mathtt{34}&  &  &  &  &  &  &  &\mathtt{AB}&  &  &  &  &  &\tikzmark{1E}{\setlength{\fboxsep}{1pt}\colorbox{blue!70!black}{\textcolor{white}{$\mathtt{12}$}}}&  &  &  &  &  &  &  &  &  &\mathtt{16}&  &  &  &  &  &  &\tikzmark{1F}{\setlength{\fboxsep}{1pt}\colorbox{blue!70!black}{\textcolor{white}{$\mathtt{15}$}}}&  &  &  &  &  &  &  &  &  &  &  &  &\mathtt{9B}&\\
\tikzmark{1I}{\setlength{\fboxsep}{1pt}\colorbox{blue!70!black}{\textcolor{white}{$\mathtt{3E}$}}}&  &  &  &  &  &  &  &\mathtt{EF}&  &  &  &  &  &  &  &  &  &  &  &  &  &\mathtt{59}&  &\tikzmark{1J}{\setlength{\fboxsep}{1pt}\colorbox{blue!70!black}{\textcolor{white}{$\mathtt{18}$}}}&  &  &  &  &  &  &  &  &  &  &\mathtt{0E}&  &  &  &  &  &  &  &  &\mathtt{83}&  &  &  &\\
  &  &  &  &  &\mathtt{D0}&  &  &  &  &  &  &  &\mathtt{06}&  &  &  &  &  &\mathtt{CD}&  &  &  &  &  &  &  &  &  &\mathtt{86}&  &  &  &  &  &  &\mathtt{8F}&  &  &  &  &\mathtt{F9}&  &  &  &  &  &  &\\
  &  &\mathtt{7F}&  &  &  &  &  &  &  &\mathtt{58}&  &  &  &  &  &\mathtt{21}&  &  &  &  &  &  &  &  &  &\mathtt{9C}&  &  &  &  &  &  &  &  &  &  &\mathtt{06}&  &  &  &  &  &  &  &  &\mathtt{87}&  &\\
\hline\hline
  &  &  &  &  &\mathtt{3D}&  &  &  &  &  &  &  &  &  &\mathtt{DC}&  &  &  &  &  &  &  &\mathtt{F3}&  &  &  &\mathtt{E6}&  &  &  &  &  &  &  &  &  &  &  &\mathtt{72}&  &  &  &  &  &  &\mathtt{5C}&  &\\
  &  &\mathtt{0F}&  &  &  &  &  &  &  &  &  &\tikzmark{1L}{\setlength{\fboxsep}{1pt}\colorbox{blue!70!black}{\textcolor{white}{$\mathtt{1D}$}}}&  &  &  &  &  &  &  &\mathtt{14}&  &  &  &\tikzmark{1K}{\setlength{\fboxsep}{1pt}\colorbox{blue!70!black}{\textcolor{white}{$\mathtt{F6}$}}}&  &  &  &  &  &  &  &  &  &  &  &\mathtt{9D}&  &  &  &  &  &  &  &  &  &  &\mathtt{1B}&\\
  &  &  &  &  &  &  &\mathtt{58}&  &\mathtt{50}&  &  &  &  &  &  &  &\tikzmark{1D}{\setlength{\fboxsep}{1pt}\colorbox{blue!70!black}{\textcolor{white}{$\mathtt{C2}$}}}&  &  &  &  &  &  &  &  &  &  &  &\mathtt{16}&  &  &  &\mathtt{89}&  &  &  &  &  &  &\tikzmark{1C}{\setlength{\fboxsep}{1pt}\colorbox{blue!70!black}{\textcolor{white}{$\mathtt{0E}$}}}&  &  &  &  &  &  &  &\\
  &  &  &  &\mathtt{14}&  &  &  &  &  &  &  &  &  &\mathtt{59}&  &  &  &  &  &  &  &\mathtt{47}&  &  &  &\mathtt{05}&  &  &  &  &  &  &  &  &  &  &  &\mathtt{F4}&  &  &\mathtt{9B}&  &  &  &  &  &  &\\
  &\mathtt{5A}&  &  &  &  &  &  &  &  &  &\mathtt{7D}&  &  &  &  &  &  &  &\mathtt{CC}&  &  &  &  &  &  &  &  &  &  &  &\mathtt{48}&  &  &  &\mathtt{69}&  &  &  &  &  &  &\mathtt{C1}&  &  &  &  &  &\\
  &  &  &  &  &  &\mathtt{C0}&  &\mathtt{88}&  &  &  &  &  &  &  &\mathtt{4C}&  &  &  &  &  &  &  &  &  &  &  &\mathtt{0A}&  &  &  &\mathtt{76}&  &  &  &  &  &  &  &  &  &  &\mathtt{A0}&  &  &  &  &\\
  &  &  &\mathtt{21}&  &  &  &  &  &  &  &  &  &\mathtt{3F}&  &  &  &  &  &  &  &\mathtt{40}&  &  &  &\mathtt{B7}&  &  &  &  &  &  &  &  &  &  &  &\mathtt{70}&  &  &  &  &  &  &\mathtt{AA}&  &  &  &\\
\tikzmark{1H}{\setlength{\fboxsep}{1pt}\colorbox{blue!70!black}{\textcolor{white}{$\mathtt{6F}$}}}&  &  &  &  &  &  &  &  &  &\mathtt{88}&  &  &  &  &  &  &  &\mathtt{24}&  &  &  &  &  &  &  &  &  &  &  &\mathtt{3E}&  &  &  &\tikzmark{1G}{\setlength{\fboxsep}{1pt}\colorbox{blue!70!black}{\textcolor{white}{$\mathtt{90}$}}}&  &  &  &  &  &  &  &  &  &  &\mathtt{72}&  &  &\\
\end{array}
\right).
}
\end{align}
\begin{tikzpicture}[remember picture, overlay]
  \draw[thick, blue!50, ->] (pic cs:1A) -- (pic cs:1B);
  \draw[thick, blue!50, ->] (pic cs:1B) -- (pic cs:1C);
  \draw[thick, blue!50, ->] (pic cs:1C) -- (pic cs:1D);
  \draw[thick, blue!50, ->] (pic cs:1D) -- (pic cs:1E);
  \draw[thick, blue!50, ->] (pic cs:1E) -- (pic cs:1F);
  \draw[thick, blue!50, ->] (pic cs:1F) -- (pic cs:1G);
  \draw[thick, blue!50, ->] (pic cs:1G) -- (pic cs:1H);
  \draw[thick, blue!50, ->] (pic cs:1H) -- (pic cs:1I);
  \draw[thick, blue!50, ->] (pic cs:1I) -- (pic cs:1J);
  \draw[thick, blue!50, ->] (pic cs:1J) -- (pic cs:1K);
  \draw[thick, blue!50, ->] (pic cs:1K) -- (pic cs:1L);
  \draw[thick, blue!50, ->] (pic cs:1L) -- (pic cs:1A);
\end{tikzpicture}
\end{figure*}
Each nonzero entry such as $\mathtt{0F}$ represents the field element $\alpha^{15}$, where $\alpha$ is a fixed primitive element of $\mathbb{F}_{q}$, and the exponent is the decimal equivalent of the hexadecimal value.
The nonzero support (positions) of $H_\Gamma$ and $H_\Delta$ matches those of the corresponding binary matrices $\HH_X$ and $\HH_Z$.
The matrices $H_\Gamma$ and $H_\Delta$ are constructed to be orthogonal.  
For example, let us examine the orthogonality between the 7th row of $H_\Gamma$ and the 5th row of $H_\Delta$.  
These rows contain the following intersecting non-zero entries:
\[
H_\Gamma[7, \cdot] = (\alpha^{200}, \alpha^{238}), \quad
H_\Delta[5, \cdot] = (\alpha^{62}, \alpha^{24}).
\]
Then, their inner product is given by:
\[
\begin{aligned}
&\alpha^{200} \cdot \alpha^{62} + \alpha^{238} \cdot \alpha^{24}\\
&= \alpha^{(200+62) \bmod 255} + \alpha^{(238+24) \bmod 255}
\\&= \alpha^7 + \alpha^7 = 0.
\end{aligned}
\]
This confirms that the two rows are orthogonal over $\mathbb{F}_{256}$.
\end{example}

\subsection{Construction of $H_X$ and $H_Z$}
\label{092546_7Jun25}
In this section, we construct orthogonal $\Fb_2$-valued matrices $H_X$ and $H_Z$ of size $m \times n$ from the orthogonal $\Fb_q$-valued matrices $H_\Gamma$ and $H_\Delta$ of size $M \times N$ obtained in the previous section, where $n = eN$ and $m = eM$. This construction follows the method used in \cite{6017122} and \cite{komoto2024quantumerrorcorrectionnear}. 
The matrices $H_X$ and $H_Z$ are constructed from $H_\Gamma$ and $H_\Delta$ as follows. The theoretical justification and proofs, which were omitted in \cite{6017122} and \cite{komoto2024quantumerrorcorrectionnear}, are provided in the appendices.

The companion matrix $A(\gamma)\in \Fb_2^{e\times e}$  for $\gamma \in \Fb_q$ are defined in Appendices~\ref{234926_5Apr25}. Using the mapping $A$, we construct the binary parity-check matrices $(H_X, H_Z)$ as follows:
\begin{align}
    H_X &= \bigl(A(\gamma_{i,j})\bigr), \quad
    H_Z = \bigl(A(\delta_{i,j})^\top\bigr).
\end{align}
From the properties of the companion matrix, it can be verified that $(H_X, H_Z)$ are orthogonal:
\begin{align}
    H_X H_Z^{\top} = O. \label{224009_19May25}
\end{align}
A proof of~\eqref{224009_19May25} is provided in Appendix~\ref{164912_4Jun25}.

We denote by $C_\Gamma$ and $C_\Delta$ the $\Fb_q$-linear spaces defined as the null spaces of the matrices $H_\Gamma$ and $H_\Delta$, respectively:
\begin{align}
    C_\Gamma &= \{\xiU \in \Fb_q^N \mid H_\Gamma \xiU = 0\},\quad     C_\Delta = \{\xiU \in \Fb_q^N \mid H_\Delta \xiU = 0\}.
\end{align}
Likewise, we denote by $C_X$ and $C_Z$ the $\Fb_2$-linear spaces defined as the null spaces of the binary matrices $H_X$ and $H_Z$, respectively:
\begin{align}
    C_X &= \{\xU \in \Fb_2^n \mid H_X \xU = 0\}, \quad    C_Z = \{\xU \in \Fb_2^n \mid H_Z \xU = 0\}.
\end{align}
The following theorem shows that the codes $C_\Gamma$ and $C_\Delta$, originally defined over $\Fb_q$, can be equivalently represented as binary codes $C_X$ and $C_Z$ via the transformations $\vU$ and $\wU$, respectively. 
The definitions of the mappings $\vU(\cdot)$ and $\wU(\cdot)$ are provided in Appendices~\ref{234926_5Apr25} and~\ref{234956_5Apr25}.
\begin{teiri}\label{thm:binary_mapping}
Let $\xiU = (\xi_0,\ldots,\xi_{N-1})$ be a vector over $\Fb_q$. Then the following statements hold:
\begin{enumerate}
 \item $\bigl(\vU(\xi_0),\ldots,\vU(\xi_{N-1})\bigr) \in C_X$ if and only if $\xiU \in C_\Gamma$.
 \item $\bigl(\wU(\xi_0),\ldots,\wU(\xi_{N-1})\bigr) \in C_Z$ if and only if $\xiU \in C_\Delta$.
\end{enumerate}
\end{teiri}
\begin{proof}
From Appendix~\ref{234926_5Apr25} and Appendix~\ref{234956_5Apr25}, we obtain the following identities for any vector $\xiU=(\xi_0,\ldots,\xi_{N-1})\in\Fb_q^N$:
\begin{align}
\sum_j(H_X)_{ij} \vU(\xi_j) &= \sum_j A(\gamma_{ij}) \vU(\xi_j) \notag\\&= \sum_j \vU(\gamma_{ij} \xi_j) = \vU\left(\sum_j \gamma_{ij} \xi_j\right), \\
\sum_j(H_Z)_{ij} \wU(\xi_j) &= \sum_j A(\delta_{ij})^\T \wU(\xi_j) \notag\\&= \sum_j \wU(\delta_{ij} \xi_j) = \wU\left(\sum_j \delta_{ij} \xi_j\right).
\end{align}
These equalities show that applying $H_X$ to the bit-level expansion $\bigl(\vU(\xi_0), \ldots, \vU(\xi_{N-1})\bigr)$ is equivalent to applying $H_\Gamma$ to $\xiU$ over $\Fb_q$ and then mapping the result via $\vU$.  
Therefore, $\bigl(\vU(\xi_0), \ldots, \vU(\xi_{N-1})\bigr)$ belongs to $C_X$ if and only if $\xiU \in C_\Gamma$.  
The second statement is proved in the same way.
\end{proof}

Write $\vU^N$ and $\wU^N$ for the componentwise extensions of the two maps. The theorem gives
\begin{align}
C_X=\vU^N(C_\Gamma),\qquad C_Z=\wU^N(C_\Delta).
\end{align}
The corresponding dual spaces satisfy
\begin{align}
\vU^N(C_\Delta^\perp)&=\operatorname{row}_{\Fb_2}(H_Z)=C_Z^\perp,\notag\\
\wU^N(C_\Gamma^\perp)&=\operatorname{row}_{\Fb_2}(H_X)=C_X^\perp.\label{eq:dual_binary_images}
\end{align}
Indeed, the rows of the blocks $A(\delta_{ij})^\top$, transposed, are $\vU(\delta_{ij}\alpha^t)$ for $t\in[e]$. Since $\{\alpha^t:t\in[e]\}$ is an $\Fb_2$-basis of $\Fb_q$, their binary row span gives the first identity. For the second, let $\beta_t=\wU^{-1}(\eU^{(t)})$, where $\eU^{(t)}$ is the $t$-th standard basis vector. The map $\wU$ is an $\Fb_2$-linear isomorphism, as shown in Appendix~\ref{234956_5Apr25}, so the $\beta_t$ also form a basis. The transposed $t$-th row of $A(\gamma_{ij})$ equals $\wU(\gamma_{ij}\beta_t)$, proving the second identity. Consequently, $C_X^\perp\subset C_Z$ is equivalent to $C_\Gamma^\perp\subset C_\Delta$ under these maps.

The binary CSS code defined by $(C_X, C_Z)$ is used for quantum error correction.  
As explained in Section~\ref{sec:171047_25May25}, the decoder does not directly use the binary matrices $H_X$ and $H_Z$ as the binary matrices that define the Tanner graphs for BP decoding.  
Instead, it uses a partitioned form of $H_X$ and $H_Z$ divided into $e$-bit blocks, or equivalently, the corresponding nonbinary matrices $H_\Gamma$ and $H_\Delta$.  
Accordingly, when we refer to the girth of $C_X$ or $C_Z$, we mean the girth of $H_\Gamma$ or $H_\Delta$, respectively; that is, the girth of $\HH_X$ or $\HH_Z$.

\label{003230_3Jun25}
The minimum distance $d$ is defined as follows, where $w_H(\cdot)$ denotes the Hamming weight of a binary vector:  
\begin{align}
d_X &= \min_{\xU \in C_X \setminus C_Z^\perp}  w_H(\xU) \\
&= \min_{\xiU \in C_\Gamma \setminus C_\Delta^\perp}  w_H\bigl(\vU(\xi_0),\vU(\xi_1),\ldots,\vU(\xi_{N-1})\bigr), \label{def:d_X} \\
d_Z &= \min_{\zU \in C_Z \setminus C_X^\perp} w_H(\zU) \\
&= \min_{\zetaU \in C_\Delta \setminus C_\Gamma^\perp} w_H\bigl(\wU(\zeta_0),\wU(\zeta_1),\ldots,\wU(\zeta_{N-1})\bigr), \label{def:d_Z} \\
d &= \min\{d_X, d_Z\}.
\end{align}
\section{Structure of Cycles in the Conventional Method}\label{sec:property_komoto_kasai}
From this section onward, we assume $J=2$ and $L=6$, unless stated otherwise.
To emphasize the structural role of $L$ in the design and analysis of the code, we avoid substituting expressions involving $L$ (such as $2L$ or $L/2$) with their explicit numeric values (e.g., $2L = 12$, $L/2 = 3$).
Instead, we consistently retain symbolic expressions in terms of $L$ throughout the presentation.

In this section, we investigate properties of the Tanner graphs associated with the parity-check matrices of CSS codes constructed using the conventional method.  
In Section~\ref{sec:235443_19May25}, we show that the matrices $H_X$ and $H_Z$ (or equivalently $\HH_X$ and $\HH_Z$) inevitably contain three types of TCBCs of length $2L$, denoted by $\uU(j)$ for $j = \{ 0,1,2 \}$.  
Next, in Section~\ref{231534_3Jun25}, we demonstrate that the structures of $\uU(j)$ for $j=\{0,1\}$ in $H_X$ and $H_Z$ correspond to codewords in $C_Z^\perp$ and $C_X^\perp$, respectively.  
Furthermore, the structure of $\uU(2)$ in $H_X$ and $H_Z$ may correspond to codewords in $C_X \setminus C_Z^\perp$ and $C_Z \setminus C_X^\perp$, respectively.
\subsection{Unavoidable Totally Closed Block Cycles}\label{sec:235443_19May25}
Suppose that $\HH_X$ and $\HH_Z$ are constructed from $\fU$ and $\gU$ according to the method described in Definition \ref{141246_12Jun25}.  
The following theorem states that, under Requirement~\ref{yosei:orthogonal}, the existence of length-$2L$ TCBCs is inevitable in the construction of the parity-check matrices $\HH_X$ and $\HH_Z$.

\begin{teiri}[Unavoidable TCBCs~\cite{komoto2024quantumerrorcorrectionnear}]\label{teiri:002234_11Feb25}
Let $P \ge 1$, $J=2$, and $L=6$. Let $\fU$ and $\gU$ be permutations in $\Fc_P$ that satisfy Requirement~\ref{yosei:orthogonal}.
For each integer $j\in \{0,1,2\}$, consider a following block cycle $\uU(j)$ of length $2L$ in $\HH_X$, starting from position $(0,0)$:
\begin{align}
\uU(j) ={}& f_0 \Rightarrow g_j \Rightarrow f_1 \Rightarrow g_{j-1} \notag\\ &\Rightarrow \cdots \Rightarrow f_{L/2-1} \Rightarrow g_{j - (L/2 - 1)} \Rightarrow f_0. \label{055301_4Feb25}
\end{align}
Each block cycle $\uU(j)$ forms a TCBC in $\HH_X$.
Note that the subscripts of $f$ and $g$ are to be interpreted modulo $L/2$, i.e., as elements of $\mathbb{Z}_{L/2}$.

A symmetric statement holds for a block cycle in $\HH_Z$, which is obtained by interchanging the roles of $X$ and $Z$. 
The corresponding block cycle is given by:
\begin{align}
\uU(j) ={}& g_0^{-1} \Rightarrow f_{-j}^{-1} \Rightarrow g_{-1}^{-1} \Rightarrow f_{-(j-1)}^{-1} \notag\\&\Rightarrow \cdots \Rightarrow g_{-(L/2 - 1)}^{-1} \notag\\&\Rightarrow f_{-(j - (L/2 - 1))}^{-1} \Rightarrow g_0^{-1}. \label{055302_4Feb25}
\end{align}
Each block cycle $\uU(j)$ forms a TCBC in $\HH_Z$.
\end{teiri}
\begin{proof}
 Let $f_{\uU(j)}$ denote the composed function corresponding to $\uU(j)$ in~\eqref{055301_4Feb25}.  
The simplified block cycle in~\eqref{055301_4Feb25} can be expanded into the full block cycle form as follows:
\begin{align}
\uU(j)={}& f_0\to g_j \to g_{j-1}\to f_0\to f_1\notag\\&\to g_{j-1} \to g_{j-2}\to f_1\to \cdots \notag\\&\to f_{L/2-1}\to g_{j-{L/2-1}} \notag\\&\to g_{j-({L/2})}\to f_{L/2-1} \to f_0
\end{align}
Therefore, the composite function of $\uU(j)$ is calculate
\begin{align}
 f_{\uU(j)}^{-1}&= (f_0^{-1}g_j g_{j-1}^{-1}f_0)(f_1^{-1}g_{j-1} g_{j-2}^{-1}f_1) \notag\\&\qquad\cdots (f_\ell^{-1}g_{j-\ell} g_{j-(\ell+1)}^{-1}f_\ell) \notag\\&\qquad\cdots (f_{L/2-1}^{-1}g_{j-{L/2-1}} g_{j-({L/2})}^{-1}f_{L/2-1})
\\&=(g_j g_{j-1}^{-1}f_0 f_0^{-1}) (g_{j-1} g_{j-2}^{-1}f_1 f_1^{-1}) \notag\\&\qquad\cdots(g_{j-\ell} g_{j-(\ell+1)}^{-1}f_\ell^{-1}f_\ell)\notag\\&\qquad\cdots (g_{j-{L/2-1}} g_{j-({L/2})}^{-1}f_{L/2-1} f_{L/2-1}^{-1})
\\&=(g_j g_{j-1}^{-1}) (g_{j-1} g_{j-2}^{-1}) \notag\\&\qquad\cdots (g_{j-\ell} g_{j-(\ell+1)}^{-1}) \notag\\&\qquad\cdots(g_{j-{L/2-1}} g_{j-({L/2})}^{-1})
\\&=g_j (g_{j-1}^{-1} g_{j-1}) (g_{j-2}^{-1}g_{j-2}) \notag\\&\qquad\cdots (g_{j-\ell}^{-1}g_{j-\ell}) \notag\\&\qquad\cdots (g_{j-{L/2-1}}^{-1}g_{j-{L/2-1}}) g_{j-({L/2})}^{-1}
\\&=g_j  g_{j-({L/2})}^{-1}
\\&=g_j  g_{j}^{-1}
\\&=\mathrm{id}.
\end{align}
Here, we used Lemma \ref{lem:194543_14Jun25} and condition~\eqref{191912_14Jun25}, which is equivalent to Requirement~\ref{yosei:orthogonal}.  
Under this condition, the following identities hold.  
\begin{align}
& f_\ell^{-1} g_{j - \ell} = g_{j - \ell} f_\ell^{-1}, \\
& f_\ell g_{j - (\ell + 1)}^{-1} = g_{j - (\ell + 1)}^{-1} f_\ell,
\end{align}
for $\ell\in [L/2]$. In particular, $f_\ell$ and $g_{\ell^{\prime}}$ commute for all $\ell, \ell^{\prime} \in [L/2]$ when $L = 6$.  
In the same manner, it can be shown that the composition of the functions along the block cycle in~\eqref{055302_4Feb25} also yields the identity function.  
\end{proof}
In~\cite{komoto2024quantumerrorcorrectionnear}, the target codes had $L \ge 8$, and the corresponding closure argument was applied for $j=0,1$.  
In contrast, the present paper focuses on the case $L = 6$, for which the theorem also holds for $j = 2$.  
The newly identified TCBC $\uU(2)$ was found to be a dominant source of performance degradation specifically when $L = 6$.  
In this work, we propose  non-binary label assignment scheme that neutralizes the harmful effect of $\uU(2)$, thereby enabling improved decoding performance.

Thus, while Theorem~\ref{teiri:002234_11Feb25} guarantees the existence of TCBCs denoded by $\uU(j)$, it leaves open the possibility that other types of TCBC may exist.
For the permutation pairs listed in Table~\ref{tab:015542_17Apr25}, direct enumeration verifies that there are no Tanner cycles shorter than $2L$ and that all length-$2L$ cycles belong to $\uU(0)$, $\uU(1)$, or $\uU(2)$.

We refer to these as \emph{unavoidable TCBCs} (UTCBCs), denoted by~\eqref{055301_4Feb25} and~\eqref{055302_4Feb25} in $\HH_X$ and $\HH_Z$, respectively. 
When the context is clear, we simply denote them by $\uU(j)$.
Depending on the context, the UTCBC $\uU(j)$ in $\HH_X$ may be regarded as the set of Tanner cycles contained within it.
Moreover, when we refer to the UTCBC $\uU(j)$ in $H_\Gamma$, we mean the submatrix in $H_\Gamma$ corresponding to the Tanner cycles that comprise $\uU(j)$ in $\HH_X$, including their non-binary matrix components.
Likewise, the UTCBC $\uU(j)$ in $H_X$ refers to the corresponding submatrix in $H_X$ that aligns with the UTCBC $\uU(j)$ in $H_\Gamma$.

\begin{example}\label{ex:040047_23Apr25}
In the matrices $\HH_X$ and $\HH_Z$ provided in Example~\ref{ex:021639_17Apr25}, two Tanner cycles contained in each $\uU(j)$ for $j = 2$ and $j = 0$ is highlighted in red and blue, respectively.
In the matrices $H_\Gamma$ and $H_\Delta$ provided in Example~\ref{035757_2Jun25}, two Tanner cycles contained in each $\uU(j)$ for $j = 2$ and $j = 0$ is highlighted in red and blue, respectively.

We consider the parity-check matrices $\HH_X$ and $\HH_Z$ with block dimensions $J = 2$ and $L = 6$. Each matrix comprises $J \times L$ blocks, where each block is of size $P \times P$. Under this configuration, each $\uU(j)$ traverses all blocks exactly once.
For $j \in  \{0, 1, 2\}$, we illustrate the sequence in which the blocks are visited in $\uU(j)$, starting from position $(0,0)$. The blocks are labeled in the order $0, 1, 2, \ldots, 10, 11$. The following matrices represent the visitation order over the $2 \times 6$ grid of blocks:
\begin{align}
\left(
\begin{array}{lll|lll}
0 & 4 & 8 & 1 & 9 & 5 \\\hline
11 & 3 & 7 & 2 & 10 & 6
\end{array}
\right),
\\
\left(
\begin{array}{lll|lll}
0 & 4 & 8 & 5 & 1 & 9 \\\hline
11 & 3 & 7 & 6 & 2 & 10
\end{array}
\right),
\\
\left(
\begin{array}{lll|lll}
0 & 4 & 8 & 9 & 5 & 1 \\\hline
11 & 3 & 7 & 10 & 6 & 2
\end{array}
\right)
\end{align}
Each matrix corresponds to $j = 0$, $j = 1$, and $j = 2$, respectively.
It can be verified that each of these block cycles $\uU(j)$ visits all $J \times L$ blocks exactly once by circulating $L/2$ steps clockwise through the four quadrants of the matrix, each consisting of half rows or half columns. This structure ensures that the entire matrix is covered uniformly in a single traversal.
\end{example}

Theorem~\ref{053905_28Apr25} concerns the structure of the UTCBCs $\uU(k)$ for $k\in \{0,1,2\}$ in $\HH_X$ and $\HH_Z$.  
Before stating the theorem, we first present the following illustrative example.
Let $\hU_{kr}^\top$ denote the $r$-th row of $\HH_X^{(k)}$. 
From the orthogonality condition~\eqref{222558_20May25}, this row is expected to be orthogonal to $\HH_Z$, satisfying the following relation:
\begin{align}
\HH_Z \hU_{kr} = \zeroU. \label{144630_21May25}
\end{align}

Let us recall the matrices $\HH_X$ and $\HH_Z$ discussed in Example~\ref{ex:021639_17Apr25}.  
For the case $k = 0$ and $r = 7$, the nonzero entries involved in the equation are located at columns 0, 12, 17, 24, 34, and 40, and at rows 1, 4, 5, 9, 10, and 15. These entries are extracted and presented as follows:
\begin{align}
\left(\begin{array}{c@{\hspace{1mm}}c@{\hspace{1mm}}c@{\hspace{1mm}}c@{\hspace{1mm}}c@{\hspace{1mm}}c@{\hspace{1mm}}}
&{\setlength{\fboxsep}{1pt}\colorbox{blue!70!black}{\textcolor{white}{$\mathtt{1}$}}}&&&&{\setlength{\fboxsep}{1pt}\colorbox{blue!70!black}{\textcolor{white}{$\mathtt{1}$}}}\\
&&{\setlength{\fboxsep}{1pt}\colorbox{blue!70!black}{\textcolor{white}{$\mathtt{1}$}}}&&{\setlength{\fboxsep}{1pt}\colorbox{blue!70!black}{\textcolor{white}{$\mathtt{1}$}}}\\
{\setlength{\fboxsep}{1pt}\colorbox{blue!70!black}{\textcolor{white}{$\mathtt{1}$}}}&&&{\setlength{\fboxsep}{1pt}\colorbox{blue!70!black}{\textcolor{white}{$\mathtt{1}$}}}\\
&{\setlength{\fboxsep}{1pt}\colorbox{blue!70!black}{\textcolor{white}{$\mathtt{1}$}}}&&{\setlength{\fboxsep}{1pt}\colorbox{blue!70!black}{\textcolor{white}{$\mathtt{1}$}}}\\
&&{\setlength{\fboxsep}{1pt}\colorbox{blue!70!black}{\textcolor{white}{$\mathtt{1}$}}}&&&{\setlength{\fboxsep}{1pt}\colorbox{blue!70!black}{\textcolor{white}{$\mathtt{1}$}}}\\
{\setlength{\fboxsep}{1pt}\colorbox{blue!70!black}{\textcolor{white}{$\mathtt{1}$}}}&&&&{\setlength{\fboxsep}{1pt}\colorbox{blue!70!black}{\textcolor{white}{$\mathtt{1}$}}}
\end{array}\right)
\left(\begin{array}{c}
 {\setlength{\fboxsep}{1pt}\colorbox{blue!70!black}{\textcolor{white}{$\mathtt{1}$}}}\\
 {\setlength{\fboxsep}{1pt}\colorbox{blue!70!black}{\textcolor{white}{$\mathtt{1}$}}}\\
 {\setlength{\fboxsep}{1pt}\colorbox{blue!70!black}{\textcolor{white}{$\mathtt{1}$}}}\\
 {\setlength{\fboxsep}{1pt}\colorbox{blue!70!black}{\textcolor{white}{$\mathtt{1}$}}}\\
 {\setlength{\fboxsep}{1pt}\colorbox{blue!70!black}{\textcolor{white}{$\mathtt{1}$}}}\\
 {\setlength{\fboxsep}{1pt}\colorbox{blue!70!black}{\textcolor{white}{$\mathtt{1}$}}}
\end{array}
\right)
=
\zeroU
\end{align}
The submatrix of $\HH_Z$ that appears in the above equation is an $L \times L$ matrix forming a Tanner cycle of length $2L$.  
We refer to this submatrix as $S$.  
Moreover, it can be verified that $S$ is contained in the UTCBC $\uU(0)$ in $\HH_Z$.  
Theorem~\ref{053905_28Apr25} asserts that this property holds for any $j \in [J]$ and $r \in [P]$, under Requirement~\ref{yosei:orthogonal} together with Requirement~\ref{yosei:square}. 

Requirement~\ref{yosei:square} is a condition imposed specifically on the permutation pair $(\fU, \gU)$ to ensure that the submatrix $S$ is square.
\begin{yosei}\label{yosei:square}
For every $j\in[J]$, $k\in\{0,1,2\}$, and distinct $\ell,\ell'\in[L/2]$, require
\begin{align}
 (f_{\ell-k}\circ g_{j-\ell})(r)
 \ne (f_{\ell'-k}\circ g_{j-\ell'})(r)
 \quad\text{for every }r\in[P].\label{163532_15Feb25}
\end{align}
All subscripts are interpreted modulo $L/2$.
\end{yosei}
As discussed in the proof of the following theorem, if Requirement~\ref{yosei:square} is not satisfied, then cycles of length less than $2L$ will appear in $\HH_X$ and $\HH_Z$.  
If the permutations $\fU$ and $\gU$ are constructed so that the girth is exactly $2L$, then Requirement~\ref{yosei:square} is automatically satisfied.  
When $L$ is small-particularly in the case $L = 6$, which is the main focus of this paper-the girth is explicitly set to $2L$.  
Therefore, it is not necessary to explicitly impose Requirement~\ref{yosei:square} in the code design process.

In the following, we restrict our attention to those permutations $(\fU, \gU)$ and the corresponding parity-check matrices $(\HH_X, \HH_Z)$ that satisfy both Requirement~\ref{yosei:orthogonal} and Requirement~\ref{yosei:square}.

\begin{teiri}[{Properties of the UTCBCs}]\label{053905_28Apr25}
Assume $J=2$, $L=6$, and Requirements~\ref{yosei:orthogonal} and~\ref{yosei:square}. For any $k \in \{0,1,2\}$, let the $r$-th row of the extended block $\HH_X^{(k)}$ be denoted by $\hU_{kr}^\top$, where $r \in [P]$. Here, $\hU_{kr}$ is a binary vector of length $N = LP$. Cross commutativity holds for every $f_\ell,g_{\ell' }$ when $L=6$, so the orthogonality calculation~\eqref{222558_20May25} also applies to the extended block $k=2$. Thus
\begin{align}
\HH_Z \hU_{kr} = \zeroU. \label{222820_20May25}
\end{align}
Let $S$ be obtained by restricting $\HH_Z$ to the columns in $\supp(\hU_{kr})$ and deleting its zero rows.
Then, the submatrix $S$ is square, forms a Tanner cycle of length $2L$, and moreover, it belongs to the UTCBC $\uU(k)$ in $\HH_Z$, as defined in Theorem~\ref{teiri:002234_11Feb25}.  
The same conclusion holds when the roles of $\HH_Z$ and $\HH_X$ are interchanged.  
\end{teiri}
\begin{proof}
While the proof is not particularly difficult, it is unfortunately not easy to present in a concise form.  
Therefore, we recommend the reader to refer to Example~\ref{023404_18May25}, which illustrates concrete instances of the key expressions used in the proof.

For $j\in [J]$ and $r\in [P]$, we denote the $r$-th row within the $j$-th row block by $[r,j]$, which corresponds to the absolute row index $jP + r$ in the matrix $\HH_X$ or $\HH_Z$.
For column indices, we use the column block index $\ell \in [L/2]$ to refer to columns within the left and right halves of $\HH_X$ or $\HH_Z$, which each consist of $L/2$ blocks.
For $c\in [P]$, the $c$-th column within the $\ell$-th column block of the {left half} is denoted by $[c,\ell]_\Lrm$, which corresponds to the absolute column index $\ell P + c$.
 The $c$-th column within the $\ell$-th column block of the {right half} is also denoted by $[c,\ell]_\Rrm$, but corresponds to the absolute column index $(\ell+ L/2)P + c$.

We begin by analyzing the columns of the submatrix $S$; that is, we examine the column indices corresponding to the nonzero entries of $\hU_{kr}$.  
Recall that $\hU_{kr}^\top$ is the $r$-th row of the matrix  $\HH_X^{(k)}$.  
In $\HH_X^{(k)}$, the $\ell$-th column block in the left and right halves corresponds to the permutations $f_{\ell - k}$ and $g_{\ell - k}$, respectively.  
In each half, each column $c \in [P]$ contains a $1$ at position $\bigl([f_{\ell - k}(c), k], [c, \ell]_\Lrm\bigr)$ in the left half and $\bigl([g_{\ell - k}(c), k], [c, \ell]_\Rrm\bigr)$ in the right half.  
Equivalently, in the left and right halves, respectively, for each row index $r \in [P]$, the entries at positions $\bigl([r,k], [f_{\ell - k}^{-1}(r),\ell]_\Lrm\bigr)$ and $\bigl([r,k], [g_{\ell - k}^{-1}(r),\ell]_\Rrm\bigr)$ are equal to $1$.  
By extending this idea to the right half as well, we now describe the positions of $1$'s in the binary vector $\hU_{kr}$.  
It can be seen that one column of $S$ appears in each column block  among the $L$ column blocks of $\HH_Z$.  
That is, $S$ has one column in each of the $L$ blocks.
\begin{figure*}[!t]
\begin{align}
\begin{array}{c|c|c|c||c|c|c|c}
  c_{\Lrm,0} & c_{\Lrm,1} & \cdots & c_{\Lrm,L/2 - 1} \quad & \quad c_{\Rrm,0} & c_{\Rrm,1} & \cdots & c_{\Rrm,L/2 - 1} \label{155227_19May25}
\end{array}
\end{align}
\end{figure*}
Here, the column indices are defined as:
\begin{align}
c_{\Lrm,\ell} = f_{\ell-k}^{-1}(r), \qquad c_{\Rrm,\ell} = g_{\ell-k}^{-1}(r).
\end{align}
Each $c_{\Lrm,\ell}$ (respectively $c_{\Rrm,\ell}$) corresponds to the index of the column within the $\ell$-th column block of the left (respectively right) half of $\HH_X^{(k)}$ that contains a $1$ in the $r$-th row.

Next, we analyze the rows of the submatrix $S$. In the $(j,\ell)$-th block of the left half of $\HH_Z$ for $j\in [J]$ and $\ell\in [L/2]$, which corresponds to the permutation $g_{-(\ell-j)}^{-1}$, the entry at position $\bigl(g_{-(\ell-j)}^{-1}(c_{\Lrm,\ell}), c_{\Lrm,\ell}\bigr)$ is equal to $1$. Proceeding similarly, the entries in $\HH_Z^{(0)}$ and $\HH_Z^{(1)}$ that intersect with the support of $\hU_{kr}$ are located at the following positions:
\begin{align}
 \begin{array}{c||c}
\bigl(r_{0,\ell},\; c_{\Lrm,\ell}\bigr)_{\ell=0,1,\ldots,L/2-1} & \bigl(s_{0,\ell},\; c_{\Rrm,\ell}\bigr)_{\ell=0,1,\ldots,L/2-1} \\\hline
\bigl(r_{1,\ell},\; c_{\Lrm,\ell}\bigr)_{\ell=0,1,\ldots,L/2-1} & \bigl(s_{1,\ell},\; c_{\Rrm,\ell}\bigr)_{\ell=0,1,\ldots,L/2-1}
\end{array}
\label{233431_18May25}
\end{align}
where the row indices are defined as
\begin{align}
r_{j,\ell} &= g_{-(\ell - j)}^{-1}(c_{\Lrm,\ell})= (g_{-(\ell - j)}^{-1}f_{\ell-k}^{-1})(r), \notag\\s_{j,\ell} &= f_{-(\ell - j)}^{-1}(c_{\Rrm,\ell})= (f_{-(\ell - j)}^{-1}g_{\ell-k}^{-1})(r).
\end{align}
Since every $f_a$ commutes with every $g_b$ for $L=6$, Lemma~\ref{lem:194543_14Jun25} applies also when $k=2$. Hence
the left and right row indices in~\eqref{233431_18May25} are equal as sets.  
More precisely, we have  
\begin{align}
 &\{r_{j,0}, r_{j,1}, \ldots, r_{j,L/2-1}\}\notag\\
 & = \{s_{j,0}, s_{j,1}, \ldots, s_{j,L/2-1}\} \quad \text{for all } j \in [J],
\end{align}
with the correspondence given explicitly by
\begin{align}
r_{j,\ell} = s_{j, j + k - \ell}, \qquad r_{j, j + k - \ell} = s_{j,\ell}.
\end{align}
Here $r_{j,\ell}=(f_{\ell-k}\circ g_{j-\ell})^{-1}(r)$. Two permutations have an intersection if and only if their inverses do, so Requirement~\ref{yosei:square} implies that each set consists of $L/2$ distinct row indices.  
Consequently, the submatrix $S$ is an $L \times L$ square matrix in which every row and every column has weight two.  
It remains to show that this matrix forms a cycle of length $2L$.  
If the requirement fails for particular $j,k$, two inverse compositions agree at some row value $r$. The corresponding closed walk then repeats a check node before completing $2L$ steps and contains a shorter cycle.  

We now describe a specific closed cycle in $S$, which is part of a block cycle winding around the four quadrants of $\HH_Z$.  
This cycle traverses the four regions into which $\HH_Z$ is divided by its horizontal and vertical halves, proceeding clockwise through them over $L/2$ steps.
Consider the entry
\[
\bigl( [r_{0,0},\ 0],\ [c_{\Lrm,0},\ 0]_\Lrm \bigr)
\]
in block $(0,0)$, located in the upper-left quadrant. This is the starting point of the Tanner cycle. During the $\ell$-th round ($\ell = 0,1,\ldots,L/2 - 1$), the traversal proceeds as follows:
\begingroup
\setlength{\arraycolsep}{2pt}
\begin{align}
\renewcommand{\arraystretch}{1.3}
\begin{array}{lllll}
\text{Start:}          & \bigl( [r_{0,\ell},& 0],\ [c_{\Lrm,\ell},& \ell]_\Lrm &\bigr)\\
\text{Right move:}     & \bigl( [r_{0,\ell} = s_{0,k-\ell},& 0],\ [c_{\Rrm,k-\ell},& k-\ell]_\Rrm &\bigr) \\
\text{Down move:}      & \bigl( [s_{1,k-\ell},& 1],\ [c_{\Rrm,k-\ell},& k-\ell]_\Rrm &\bigr) \\
\text{Left move:}      & \bigl( [r_{1,\ell+1} = s_{1,k-\ell},& 1],\ [c_{\Lrm,\ell+1},& \ell+1]_\Lrm &\bigr) \\
\text{Up move:}        & \bigl( [r_{0,\ell+1},& 0],\ [c_{\Lrm,\ell+1},& \ell+1]_\Lrm &\bigr)
\end{array}
\end{align}
\endgroup
By viewing $\HH_Z$ as consisting of four quadrants (upper-left, upper-right, lower-right, and lower-left), we see that the block cycle rotates {clockwise}, completing $L/2$ full rounds.
The sets of column indices in the left and right halves of $\HH_Z$ are given by
\begin{align}
C_{\Lrm} &= \bigl( [c_{\Lrm,\ell},\, \ell]_\Lrm \bigr)_{\ell = 0, 1, \ldots, L/2 - 1}, \notag\\C_{\Rrm} &= \bigl( [c_{\Rrm,\ell},\, \ell]_\Rrm \bigr)_{\ell = 0, 1, \ldots, L/2 - 1}.
\end{align}
The sets of row indices in the upper and lower halves of $\HH_Z$ are defined as
\begin{align}
R_0 &= \bigl( [r_{0,\ell},\, 0] \bigr)_{\ell = 0, 1, \ldots, L/2 - 1}, \\
R_1 &= \bigl( [r_{1,\ell},\, 1] \bigr)_{\ell = 0, 1, \ldots, L/2 - 1}.
\end{align}
The corresponding submatrix of $\HH_Z$ that connects these row and column sets, with entries corresponding to $1$s, can be compactly represented as follows:
\begin{align}
\renewcommand{\arraystretch}{1.4} 
\begin{array}{c||c|c}
         & C_{\Lrm} & C_{\Rrm} \\
\hline\hline
R_0      & I        & M^{(k)} \\
\hline
R_1      & I        & M^{(k+1)}
\end{array}
\label{181207_19May25}
\end{align}
Here, $I$ denotes the identity matrix of size $L/2 \times L/2$, and $M^{(k)}$ denotes a permutation matrix of size $L/2 \times L/2$ defined by
\begin{align}
M^{(k)}_{ij} =
\begin{cases}
1 & \text{if } i + j \equiv k \pmod{L/2}, \\
0 & \text{otherwise}.
\end{cases}
\end{align}
The Tanner cycle in question is contained in a block cycle, which can be expressed in a simplified form for $J = 2$ as follows:
\begin{align}
\uU(k) ={}& g_0^{-1} \Rightarrow f_{-k}^{-1} \Rightarrow g_{-1}^{-1} \Rightarrow f_{-(k-1)}^{-1} \notag\\&\Rightarrow \cdots \Rightarrow g_{-(L/2 - 1)}^{-1} \notag\\&\Rightarrow f_{-(k - (L/2 - 1))}^{-1} \Rightarrow g_0^{-1}. 
\end{align}
This sequence coincides with  the UTCBC $\uU(k)$ in $\HH_Z$, as defined in \eqref
{055302_4Feb25} of Theorem~\ref{teiri:002234_11Feb25}.
A similar argument applies to the case obtained by interchanging $X$ and $Z$.
\end{proof}
\begin{example}\label{023404_18May25}
Let us recall the example discussed in Example~\ref{ex:021639_17Apr25}. For $k = 0$ and $r = 7$, we consider the submatrix $S$ involved in the equation
\begin{align}
\HH_Z \hU_{kr} = \zeroU,
\end{align}
where $\hU_{kr}^\top$ is the $r$-th row of $\HH_X^{(k)}$. The sets in~\eqref{155227_19May25} are instantiated as follows:
\begin{align}
\begin{array}{c|c|c||c|c|c}
c_{\Lrm,0} & c_{\Lrm,1} & c_{\Lrm,2} & c_{\Rrm,0} & c_{\Rrm,1} & c_{\Rrm,2} \\
\hline
0 & 4 & 1 & 0 & 2 & 0
\end{array}
\end{align}
The positions in~\eqref{233431_18May25} are given explicitly as:
\begin{figure*}[!t]
\begin{align}
\begin{array}{c|c|c||c|c|c}
(r_{0,0}, c_{\Lrm,0}) & (r_{0,1}, c_{\Lrm,1}) & (r_{0,2}, c_{\Lrm,2}) & (s_{0,0}, c_{\Rrm,0}) & (s_{0,1}, c_{\Rrm,1}) & (s_{0,2}, c_{\Rrm,2}) \\
\hline
(5,0) & (1,4) & (4,1) & (5,0) & (4,2) & (1,0) \\
(7,0) & (1,4) & (2,1) & (1,0) & (7,2) & (2,0)
\end{array}
\end{align}
\end{figure*}
The submatrix representation in~\eqref{181207_19May25} is instantiated as:
 \begin{align}
\renewcommand{\arraystretch}{1.4}
\setlength{\arraycolsep}{4.8pt}\begin{array}{c||c|c|c||c|c|c}
       & [0,0]_\Lrm & [4,1]_\Lrm & [1,2]_\Lrm & [0,0]_\Rrm & [2,1]_\Rrm & [0,2]_\Rrm \\
\hline\hline
{[5,0]} & \tikzmark{A}1 &         &           & \tikzmark{B}1 &         & \\
{[1,0]} &   & \tikzmark{E}1       &           &               &         & \tikzmark{F}1 \\
{[4,0]} &   &         &\tikzmark{I}1         &               & \tikzmark{J}1       & \\
\hline
{[7,1]} & \tikzmark{L}1 &         &           &               & \tikzmark{K}1       & \\
{[1,1]} &   &\tikzmark{D} 1       &           & \tikzmark{C}1             &         & \\
{[2,1]} &   &         & \tikzmark{H}1         &               &         & \tikzmark{G}1 \\
\end{array}
 \end{align}
\begin{tikzpicture}[remember picture, overlay]
  \draw[thick, blue!50, ->] (pic cs:A) -- (pic cs:B);
  \draw[thick, blue!50, ->] (pic cs:B) -- (pic cs:C);
  \draw[thick, blue!50, ->] (pic cs:C) -- (pic cs:D);
  \draw[thick, blue!50, ->] (pic cs:D) -- (pic cs:E);
  \draw[thick, blue!50, ->] (pic cs:E) -- (pic cs:F);
  \draw[thick, blue!50, ->] (pic cs:F) -- (pic cs:G);
  \draw[thick, blue!50, ->] (pic cs:G) -- (pic cs:H);
  \draw[thick, blue!50, ->] (pic cs:H) -- (pic cs:I);
  \draw[thick, blue!50, ->] (pic cs:I) -- (pic cs:J);
  \draw[thick, blue!50, ->] (pic cs:J) -- (pic cs:K);
  \draw[thick, blue!50, ->] (pic cs:K) -- (pic cs:L);
  \draw[thick, blue!50, ->] (pic cs:L) -- (pic cs:A);
\end{tikzpicture}
\end{example}

\subsection{Impact of UTCBCs on Minimum Distance}\label{231534_3Jun25}
Following common conventions in the literature on LDPC codes, we identify the Tanner graph $\Gt$ corresponding to the matrix $H_\Gamma$ with $H_\Gamma$ itself. Accordingly, columns correspond to variable nodes, rows to check nodes, and submatrices to subgraphs, allowing us to treat these elements interchangeably. For example, the Tanner graph associated with a submatrix $S$ of $H_\Gamma$ is also denoted by $S$. Similarly, a subgraph of $H_\Gamma$ is denoted by $\Gt'$, and the corresponding submatrix is also referred to as $\Gt'$.

Let us consider a Tanner cycle $\Ct$ of length $2L$ contained in $H_\Gamma$.
Let $(H_\Gamma)_\Ct$ denote the submatrix of $H_\Gamma$ obtained by zeroing out all columns not involved in the cycle $\Ct$.  
Similarly, let $\xiU_\Ct$ denote the vector obtained from $\xiU \in \Fb_q^N$ by zeroing out all coordinates not corresponding to columns in $\Ct$.
We define the subcode associated with the cycle $\Ct$ as
\begin{align}
 N(\Ct;H_\Gamma) \defeq \left\{ \xiU_\Ct \in \Fb_q^N\mid (H_\Gamma)_\Ct \, \xiU_\Ct = \zeroU,\ \xiU \in \Fb_q^N \right\}.
 \label{042010_2Jun25}
\end{align}
The set $N(\Ct;H_\Gamma)$ is a subcode of $C_\Gamma$, supported only on the positions involved in the cycle $\Ct$.  
It captures the contribution of the cycle $\Ct$ to the code $C_\Gamma$ by isolating codewords whose nonzero positions lie solely within $\Ct$.
For cycles in $\uU(0)$ or $\uU(1)$, Theorem~\ref{192838_2Jun25} shows that $N(\Ct;H_\Gamma)$ is contained in $C_\Delta^\perp \subset C_\Gamma$, and therefore such codewords do not affect the minimum distance $d_X$ (see~\eqref{def:d_X}).

\begin{teiri}\label{192838_2Jun25}
Assume $J=2$, $L=6$, Requirements~\ref{yosei:orthogonal} and~\ref{yosei:square}, and nonzero labels on the prescribed supports. Let $H_\Gamma H_\Delta^\top=O$. Let $\Ct$ be a Tanner cycle in the UTCBC $\uU(k)$ in $H_\Gamma$ for $k = 0,1$, which is orthogonal to the row vector $\hU_{kr}^\top$, in the sense that
\begin{align}
 (H_\Gamma)_\Ct \, \hU_{kr} = \zeroU. \label{144858_5Jun25}
\end{align}
Here, $\hU_{kr}^\top$ denotes the $r$-th row of $H_\Delta^{(k)}$.
Then, the associated subcode $N(\Ct;H_\Gamma)$ is a one-dimensional subspace over $\Fb_q$ given by
\begin{align}
N(\Ct;H_\Gamma) = \{ \xi \hU_{kr} \mid \xi \in \Fb_q \},\label{145346_5Jun25}
\end{align}
where $\xi$ is a scalar in $\Fb_q$.
Moreover, $N(\Ct;H_\Gamma)$ is a subspace of $C_\Delta^\perp \subset C_\Gamma$.
A similar statement holds when $H_\Delta$ is considered in place of $H_\Gamma$, i.e., when the roles of $X$ and $Z$ are interchanged.
\end{teiri}
\begin{proof}
We first observe that $\supp(\hU_{kr}) = \supp(\Ct)$.  
By Theorem~\ref{053905_28Apr25}, the cycle $\Ct$ is square-shaped and can be arranged into an upper-triangular form except for one row. Therefore, its rank is at least $L - 1$.  
Since the column weight is $J = 2$, the rank of the submatrix $(H_\Gamma)_\Ct$ coincides with that of the cycle $\Ct$.  
From~\eqref{144858_5Jun25}, a nonzero solution $\hU_{kr}$ exists in the kernel of $(H_\Gamma)_\Ct$, implying that its rank is at most $L - 1$.
Hence, the rank of $(H_\Gamma)_\Ct$ is exactly $L - 1$, and this proves~\eqref{145346_5Jun25}.
Since $\hU_{kr}^\top$ is a row of the parity-check matrix $H_\Delta$, it follows that $N(\Ct;H_\Gamma)$ is a subspace of $C_\Delta^\perp$.
\end{proof}

\begin{example}
 As an example, consider the cycle $\Ct$ in $H_\Delta$ highlighted in blue in Example~\ref{035757_2Jun25}.
 Consequently, $N(\Ct;H_\Delta)$ is a subspace of $C_\Gamma^\perp$.
 Here, $\hU_{kr}^\top$ is the $r$-th row of $H_\Gamma^{(k)}$ for $k=0, r=7$, which is also highlighted in blue:
\begin{figure*}[!t]
 \begin{align}
\input{blue_row}
 \end{align}
\end{figure*}
We extract only the nonzero components involved in the equation $H_\Delta \hU_{kr} = \zeroU$.
The relevant portion can be represented as the following matrix-vector product:
\input{blue_matrix}
\end{example}

\begin{giron}\label{234642_5Jun25}
 In contrast to the cases $k = 0, 1$, the UTCBCs $\uU(k)$ for $k = 2$ can negatively affect the minimum distance $d_X$.  
 As an example, consider the cycle $\Ct$ in $H_\Gamma$ highlighted in red in Example~\ref{035757_2Jun25}.  The cycle $\Ct$ corresponds to a square matrix of size $L$:
 \input{red_matrix}
In this instance, the matrix $\Ct$ happens to be of full rank, and hence the associated subcode $N(\Ct;H_\Gamma)$ is the trivial (zero-dimensional) subspace consisting only of the all-zero codeword.  That is, $N(\Ct;H_\Gamma) = \{\zeroU\}$.  
 However, if the arrangement of nonzero symbols in $H_\Gamma$ were unfortunate and $\Ct$ failed to be full rank, then $N(\Ct;H_\Gamma)$ would become a one-dimensional subspace, containing a nonzero codeword.  
 Such a codeword, if it lies outside $C_\Delta^\perp$, corresponds to an uncorrectable logical error and can degrade the minimum distance $d_X$ (see~\eqref{def:d_X}).
\end{giron}

Requiring these cycle matrices to have full rank eliminates all nonzero codewords supported on each individual $\uU(2)$ cycle, including any nontrivial logical operators on that support.

\section{Proposed Code Construction }\label{sec:150611_18May25}

This section proposes an improved construction of the matrices $\HH_X$, $\HH_Z$, $H_\Gamma, H_\Delta, H_X$ and $H_Z$, based on the observations made in the previous section.

As shown in Theorem~\ref{teiri:002234_11Feb25}, the conventional codes presented in \cite{komoto2024quantumerrorcorrectionnear} inevitably contain length-$2L$ cycles $\Ct$ that belong to the UTCBCs $\uU(k)$ for 
$k\in \{0,1,2\}$.  Among these, cycles in $\uU(0)$ and $\uU(1)$ are harmless to the minimum distance of the code, as established in Theorem~\ref{192838_2Jun25}.  
In contrast, as explained in Discussion~\ref{234642_5Jun25},  
cycles in UTCBC $\uU(2)$ may introduce low-weight codewords, depending on the specific values of the nonzero entries along the cycle, and therefore can degrade the minimum distance of $C_\Gamma$.  
The same argument applies to $H_\Delta$ by symmetry, by interchanging the roles of $H_\Gamma$ and $H_\Delta$.  

In a related work \cite{kasai2025quantumerrorcorrectiongirth16}, the error floor was successfully reduced for a code with $L = 8$ and rate $R = 1/2$ by eliminating all length-$2L$ cycles.  
However, for the target parameters considered in this paper, namely $J=2, L=6$ and $R = 1/3$, length-2L cycles in UTCBCs are unavoidable due to Theorem~\ref{teiri:002234_11Feb25}.

\subsection{Guiding Principles for the Construction}\label{sec:000508_6Jun25}
The objective of this section is to refine the conventional construction method of~\cite{komoto2024quantumerrorcorrectionnear} (see Section~\ref{192200_2Jun25}) so that the resulting codes satisfy the following three properties:
\begin{enumerate}
 \item \label{000948_6Jun25} The matrices $\HH_X$ and $\HH_Z$ contain no cycles of length less than $2L$.
 \item \label{000952_6Jun25} All length-$2L$ cycles in $\HH_X$ and $\HH_Z$ belong to one of the UTCBCs $\uU(k)$ with $k \in \{0,1,2\}$.  
 \item \label{000957_6Jun25} For all length-$2L$ cycles $\Ct$ belonging to the UTCBC $\uU(2)$ in $H_\Gamma$ and $H_\Delta$,  
 the null spaces $N(\Ct;H_\Gamma)$ and $N(\Ct;H_\Delta)$ are equal to $\{\zeroU\}$, respectively.  
 In other words, each such cycle $\Ct$ must have full rank $L$ (see Discussion~\ref{234642_5Jun25}).  
\end{enumerate}
Although length-12 cycles are unavoidable for $L=6$ (see Theorem~\ref{teiri:002234_11Feb25}), our goal is to achieve robust performance even when $L = 6$ by carefully managing the structure and impact of UTCBCs.

In this section, we explicitly construct permutations $\fU$ and $\gU$, along with the associated matrices $\HH_X$, $\HH_Z$, $H_\Gamma$, and $H_\Delta$.
Our goal is not merely to construct a single pair $(\HH_X, \HH_Z)$ that satisfies the required conditions, but rather to randomly sample such pairs from the space of all valid constructions.  
This randomized approach is essential for harnessing the performance advantages of LDPC codes viewed as random codes \cite{Gallager1962}.

In Section~\ref{sec:property_komoto_kasai}, we investigated the structural properties of the conventional construction over a broad class of permutations, assuming that $\fU$ and $\gU$ are drawn from general subsets of $\Fc_P$.  
Although it is difficult to randomly generate permutations from $\Fc_P$ that satisfy the necessary conditions, the construction becomes significantly more efficient by leveraging the algebraic structure of APMs.  
While exploring the full class $\Fc_P$ may offer greater randomness and potentially better-performing codes, we restrict ourselves to the structured subset $\mathcal{A}_P \subset \Fc_P$ to ensure algorithmic feasibility.
Therefore, in the remainder of this paper, we focus on permutations selected from $\mathcal{A}_P$.

\subsection{Construction of $\HH_X$ and $\HH_Z$}

\begin{table*}[t]
\caption{Constructed $f_i(X)$ and $g_i(X)$ polynomials for each $P$}
\label{tab:015542_17Apr25}
\centering
\begin{align}
 \begin{array}{@{}c|ccc|ccc@{}}
\toprule
P & f_0(X) & f_1(X) & f_2(X) & g_0(X) & g_1(X) & g_2(X) \\
\midrule
384 =2^7\cdot 3  & 221X+358 & 101X+314 & 217X+92  & 199X+303 & 169X+324 & 343X+375 \\
768 =2^8\cdot 3 & 235X+723 & 127X+345 & 277X+6   & 565X+374 & 725X+166 & 709X+366 \\
1536 =2^9\cdot 3 & 1003X+723 & 91X+219 & 1045X+6 & 1333X+1142 & 65X+1248 & 473X+1012 \\
3072 =2^{10}\cdot 3 & 2155X+1773 & 1165X+1110 & 1237X+2010 & 2957X+1238 & 1885X+638 & 2425X+2908 \\
6144 =2^{11}\cdot 3 & 1099X+1665 & 5875X+69 & 1153X+5952 & 2957X+974 & 2173X+4838 & 1973X+2386 \\
6500 =2^2 \cdot 5^3 \cdot 13 ;& 1X+ 2998&  1501X+ 3518&  5501X+ 2346 &3251X+ 4459&  3251X+ 3900 &    1X+  988 \\
12288 =2^{12}\cdot 3 & 3433X+3987 & 10801X+9018 & 10177X+6408 & 6065X+5770 & 3169X+2932 & 10193X+8070 \\
\bottomrule
\end{array}
\end{align}
\end{table*}

In this section, we aim to construct a pair of arrays of APMs
\[
\fU = (f_0, f_1, \ldots, f_{L/2 - 1}), \quad \gU = (g_0, g_1, \ldots, g_{L/2 - 1}),
\]
that satisfy the following criteria:
\begin{itemize}
  \item[(a)] Every required cross pair $(f_i,g_j)$ must satisfy the commutativity condition specified in Requirement~\ref{yosei:orthogonal}.
  \item[(b)] To prevent the formation of $2 \times 3$ totally closed block cycles (see Theorem~\ref{144806_5Apr25}), complete commutativity among the $f_i$'s or among the $g_i$'s must be avoided.  
  Specifically, for some $i \ne j$, the permutations $f_i$ and $f_j$ should not commute, and likewise for $g_i$ and $g_j$.
  \item[(c)] Other than the UTCBCs $\uU(k)$ for $k\in \{ 0,1,2\}$, no additional closed block cycles  of length at most $2L$ should appear in either $\HH_X$ or $\HH_Z$.
\end{itemize}
\textit{Note.} Although condition~(b) may appear to be implied by condition~(c), it plays a crucial role in the construction process.  
In practice, if condition~(b) is not explicitly enforced, it becomes extremely difficult to find sequences that satisfy condition~(c) via random sampling.  
We do not explicitly include Requirement~\ref{yosei:square} in the above criteria, since it is automatically satisfied once condition~(c) holds.

To construct such APMs $\fU$ and $\gU$, we employ a sequential randomized algorithm that incrementally builds each sequence by adding one APM at a time while checking the required constraints.  
Starting from an empty list, candidates for $f_i$ and $g_i$ are iteratively sampled from $\mathcal{A}_P$ and evaluated to ensure that the resulting partial sequences satisfy all applicable conditions~(a), (b), and (c). Conditions involving unassigned maps are deferred, and condition~(b) is checked only when the corresponding family is complete.  
If a candidate satisfies all constraints, it is accepted and appended to the sequence; otherwise, a new candidate is drawn.  
This process is repeated until valid sequences $\fU$ and $\gU$ of length $L/2$ are obtained.  
The full procedure is detailed in Algorithm~\ref{alg:fg_construction}.

\renewcommand\algorithmicindent{.5em}
\begin{algorithm}[thbp]
\caption{Construction Algorithm of $\fU$ and $\gU$}
\label{alg:fg_construction}
\begin{algorithmic}[1]
\State Initialize empty list $\mathcal{S} \gets [\ ]$
\For{$i = 0$ to $L/2 - 1$}
    \Repeat
        \State Randomly generate a candidate $f_i$ from $\mathcal{A}_P$
        \State Temporarily set $\mathcal{S}' \gets \mathcal{S} \cup \{f_i\}$
        \If{$\mathcal{S}'$ passes all applicable checks (a), (b), and (c)}
            \State Accept $f_i$: $\mathcal{S} \gets \mathcal{S}'$
            \State \textbf{break}
        \EndIf
 \Until{a valid $f_i$ is found}
    \Repeat
        \State Randomly generate a candidate $g_i$ from $\mathcal{A}_P$
        \State Temporarily set $\mathcal{S}' \gets \mathcal{S} \cup \{g_i\}$
        \If{$\mathcal{S}'$ passes all applicable checks (a), (b), and (c)}
            \State Accept $g_i$: $\mathcal{S} \gets \mathcal{S}'$
            \State \textbf{break}
        \EndIf
 \Until{a valid $g_i$ is found}
\EndFor
\State \textbf{return} $\mathcal{S}$
\end{algorithmic}
\end{algorithm}

Using the APMs $\fU$ and $\gU$ determined by Algorithm~\ref{alg:fg_construction}, we compute the binary matrices $\HH_X$ and $\HH_Z$.
\subsection{Determination of $H_\Gamma$ and $H_\Delta$}\label{sec:se235947_5Jun25}
In the previous section, we constructed matrices $\HH_X$ and $\HH_Z$ that satisfy  
Properties~\ref{000948_6Jun25} and~\ref{000952_6Jun25} listed in Section~\ref{sec:000508_6Jun25}. 
In this section, we aim to construct orthogonal matrices $H_\Gamma$ and $H_\Delta$ that share the same support as $\HH_X$ and $\HH_Z$, respectively,  
and satisfy Property~\ref{000957_6Jun25} stated in Section~\ref{sec:000508_6Jun25}.
Formally, in addition to the orthogonality conditions~\eqref{015905_13Feb25} and~\eqref{095733_27May25},  
we seek vector representations $\gammaU$ and $\deltaU$ of $H_\Gamma$ and $H_\Delta$ over $\mathbb{F}_q$ that satisfy the following constraints:
\begin{figure*}[!t]
\begin{align}
  A_{2} \log \gammaU &:= \left[-\HH_Z^{(2,L)} \mid \HH_Z^{(2,R)} \mid \HH_Z^{(2,L)} \mid -\HH_Z^{(2,R)}\right] \log \gammaU = \cU \mod (q - 1), \label{160947_1Jun25} \\
  B_{2} \log \deltaU &:= \left[-\HH_X^{(2,L)} \mid \HH_X^{(2,R)} \mid \HH_X^{(2,L)} \mid -\HH_X^{(2,R)}\right] \log \deltaU = \dU \mod (q - 1), \label{160950_1Jun25}
\end{align}
\end{figure*}
where $c_i \ne 0$ and $d_i \ne 0$ for all $i \in [P]$.
These constraints are derived in Appendix~\ref{161030_1Jun25}.  
Each system consists of $P$ equations in $2N$ variables, specifically the entries $\gamma_i$ and $\delta_i$ for $i = 0, 1, \ldots, 2N - 1$.

By performing Gaussian elimination with unit pivots over $\mathbb Z_{q-1}$, we can obtain a solution $\log \gammaU$ that satisfies $A_{01} \log \gammaU = \zeroU$.  
In this process, we decompose the vector $\log \gammaU$ satisfying~\eqref{015905_13Feb25}, i.e., $A_{01} \log \gammaU = \zeroU$, into a bound part and a free part, denoted as $\log \gammaU = (\log \gammaU_{\brm} \mid \log \gammaU_{\frm})$.  
Let $G_{01}$ be the linear transformation from the free part to the bound part:  
\begin{align}
 \log \gammaU_{\brm} = G_{01} \log \gammaU_{\frm},\label{005601_6Jun25}
\end{align}
where $G_{01}$ is a matrix over $\mathbb{Z}_{q-1}$, and $B$ and $F$ denote the lengths of the bound and free parts, respectively.

Such pivots exist for $A_{01}$: negating its second block row gives a signed incidence matrix, and a spanning forest yields pivots in $\{+1,-1\}$, which are invertible modulo $q-1$.

We employ the following two randomized algorithms-Algorithm~\ref{alg:HGamma} and Algorithm~\ref{alg:HDelta}-to construct the parity-check matrices $H_\Gamma$ and $H_\Delta$, respectively, with the desired properties.

Algorithm~\ref{alg:HGamma} constructs a matrix $H_\Gamma$ such that $A_{01} \log \gammaU=\zeroU$  and $A_{2} \log \gammaU$ contains no zero entries.  
The algorithm begins by randomly generating the free part $\log \gammaU_{\frm}$ and computing the corresponding bound part $\log \gammaU_{\brm}$ using \eqref{005601_6Jun25}.
Next, we compute a $\Zb_{q-1}$-valued vector $\cU^{(0)} := A_2 \log \gammaU$ of length $P$.  
If all entries of $\cU^{(0)}$ are nonzero, then the corresponding vector $\gammaU$ is accepted.  
In this work, we set $q = 2^e$ with $e = 8$.
If any entry of $\cU^{(0)}$ is zero, the algorithm perturbs a single free logarithmic variable and updates the bound part accordingly.  
This process is repeated until all entries of $A_{2} \log \gammaU$ are nonzero.

Given $\gammaU$, or equivalently $H_\Gamma$, Algorithm~\ref{alg:HDelta} constructs a matrix $H_\Delta$ such that $H_\Gamma H_\Delta^\top = \zeroU$ and $B_{2} \log \deltaU$ contains no zero entries.  
As a first step, the algorithm solves the linear equation $H_\Gamma H_\Delta^\top = \zeroU$ and samples each row of $H_\Delta$ as a nonzero scalar multiple of its prescribed cycle-kernel generator.  
Next, we compute $\dU^{(0)} := B_2 \log \deltaU$.  
If any zero entry in $\dU^{(0)}$ is detected, the algorithm perturbs a free variable in $\log \gammaU_\frm$, recomputes the bound part $\log \gammaU_\brm$, and updates $H_\Delta$ accordingly.  
This process continues until $B_{2} \log \deltaU$ is fully nonzero while retaining the $A_2$ condition. Since $H_\Gamma$ may change, Algorithm~\ref{alg:HDelta} returns the final pair. At each step, $H_\Gamma$ is obtained by exponentiating the current $\log\gammaU$, including after a rejected perturbation is undone.
\begin{algorithm}[t]
\caption{Construction of $H_\Gamma$ such that $A_{01} \log \gammaU=\zeroU$ and $A_{2} \log \gammaU$ is nonzero in all entries}
\label{alg:HGamma}
\begin{algorithmic}[1]
\State Randomly generate $\log \gammaU_{\frm}\in \Zb_{q-1}^{F}$
\State Compute $\log \gammaU_{\brm} \gets G_{01} \log \gammaU_{\frm}$
\State Compute $\cU^{(0)} \gets A_{2} \log \gammaU$
\If{all entries of $\cU^{(0)}$ are nonzero}
  \State \Return $H_\Gamma$ and terminate
\EndIf
\State $i \gets 0$
\While{true}
  \State Save $\log\gammaU$; perturb one coordinate of $\log\gammaU_{\frm}$
  \State Recompute $\log \gammaU_{\brm} \gets G_{01} \log \gammaU_{\frm}$
  \State Compute $\cU^{(i+1)} \gets A_{2} \log \gammaU$
  \If{all entries of $\cU^{(i+1)}$ are nonzero}
    \State \Return $H_\Gamma$ 
\ElsIf{the number of zeros in $\cU^{(i+1)}$ is greater than or equal to that in $\cU^{(i)}$}
    \State Restore $\log\gammaU$; set $\cU^{(i+1)}\gets\cU^{(i)}$ 
  \EndIf
  \State $i \gets i + 1$
\EndWhile
\end{algorithmic}
\end{algorithm}

\begin{algorithm}[t]
\caption{Construction of $H_\Delta$ such that $H_\Gamma H_\Delta^\top = O$ and $B_{2} \log \deltaU$ is nonzero in all entries}
\label{alg:HDelta}
\begin{algorithmic}[1]
\Require $H_\Gamma$ from Algorithm~\ref{alg:HGamma}
\State Sample each row of $H_\Delta$ as a nonzero multiple of its prescribed cycle-kernel generator
\State Compute $\dU^{(0)} \gets B_{2} \log \deltaU$
\If{all entries of $\dU^{(0)}$ are nonzero}
  \State \Return $(H_\Gamma,H_\Delta)$ and terminate
\EndIf
\State $i \gets 0$
\While{true}
  \State Save $(\log\gammaU,H_\Delta)$; perturb one coordinate of $\log\gammaU_{\frm}$
  \State Recompute $\log \gammaU_{\brm} \gets G_{01} \log \gammaU_{\frm}$
  \State Sample each row of $H_\Delta$ as a nonzero multiple of its prescribed cycle-kernel generator
  \State Compute $\cU \gets A_{2} \log \gammaU$
  \State Compute $\dU^{(i+1)} \gets B_{2} \log \deltaU$
  \If{$\cU$ contains any zero entries \textbf{or} the number of zeros in $\dU^{(i+1)}$ is greater than or equal to those in $\dU^{(i)}$}
    \State Restore $(\log\gammaU,H_\Delta)$; set $\dU^{(i+1)}\gets\dU^{(i)}$ 
  \ElsIf{all entries of $\dU^{(i+1)}$ are nonzero}
    \State \Return $(H_\Gamma,H_\Delta)$ 
  \EndIf
  \State $i \gets i + 1$
\EndWhile
\end{algorithmic}
\end{algorithm}

This paragraph discusses a trade-off in the proposed code construction between the probability of satisfying algebraic constraints and the decoding complexity.
Assuming that the prescribed nonzero $\Fb_q$-valued labels assigned to each cycle are independently and uniformly drawn from $\Fb_q^\times$, the probability that a single cycle in the UTCBC $\uU(2)$ is full rank is $1 - 1/(q-1)$, which increases with $q$.  
Since there are $P$ such cycles, and assuming their labels are independently assigned, the probability that all of them are full rank is $(1 - 1/(q-1))^P$.  
These probabilities apply to the independent-label model, not to the coupled orthogonality-constrained construction. In that model, increasing $q$ improves the chance that all cycles are full rank, whereas increasing $P$ reduces it.
On the other hand, the joint BP algorithm includes message updates costing $O(q^2)$.  
This is specifically required in the message updates in equations~\eqref{013318_27May25} and~\eqref{013321_27May25}.  
As $q = 2^e$, the complexity grows rapidly with increasing $e$.  
This creates a practical constraint: while a larger $q$ improves the algebraic robustness of the code, it also significantly increases decoding cost.  
Therefore, to keep the decoder complexity manageable, $q$ must be kept reasonably small, while the achievable lift size $P$ also depends on the constrained construction search.

\subsection{Construction of $H_X$ and $H_Z$ from $H_\Gamma$ and $H_\Delta$}
Finally, the binary parity-check matrices $H_X$ and $H_Z$ are constructed from the non-binary matrices $H_\Gamma$ and $H_\Delta$ in the same manner as described in Section~\ref{092546_7Jun25}. Specifically, each entry $\gamma_{ij} \in \Fb_q$ in $H_\Gamma$ is mapped to a binary $e \times e$ matrix $A(\gamma_{ij})$ using the companion matrix representation, and similarly, each entry $\delta_{ij} \in \Fb_q$ in $H_\Delta$ is mapped to $A(\delta_{ij})^\top$. This yields the binary matrices $H_X$ and $H_Z$ of size $m \times n$, where $m = eM$ and $n = eN$.

\section{Joint Belief Propagation and Post-Processing Decoding}
\label{sec:171047_25May25}
This section presents a decoding method designed to fully exploit the performance of the proposed codes described in the previous section.  
The primary goal is to concisely define a post-processing algorithm that enhances the conventional decoding method.

Section~\ref{013101_9Jun25} provides a general formulation of syndrome decoding using the binary vector representation of Pauli errors.  
Section~\ref{013101_9Jun25} defines the depolarizing channel, which serves as the physical noise model assumed in our numerical experiments.  
Section~\ref{153729_6Jun25} explains the degeneracy of quantum noise vectors and the conditions under which syndrome decoding can successfully recover the codeword.  
Section~\ref{subsec:205908_25May25} describes the joint BP decoding algorithm over binary vectors, which serves as the basis of our approach.  
Section~\ref{sec:143825_6Jun25} reformulates this decoding procedure over a finite field, enabling a more compact representation and facilitating the post-processing method.  
Section~\ref{150214_6Jun25} then analyzes the behavior of the conventional decoder in the error floor regime.  
Section~\ref{150240_6Jun25} presents an algorithm for estimating a set of cycles responsible for decoder stagnation.  
Finally, Section~\ref{150247_6Jun25} describes the proposed post-processing algorithm that refines the decoder output by solving a small linear system over the estimated support.
\subsection{Pauli Channel and Binary Vector Representation of Errors}\label{013101_9Jun25}
The Pauli group $\mathcal{P}_n$ is non-commutative; however, by ignoring the global phase factor $\alpha$ of Pauli operators, we obtain the quotient group $\mathcal{P}_n /\{\pm I, \pm i I\}$ by the central subgroup $\{\pm I, \pm i I\}$.
Here, $I$ represents the identity operator on the space $\mathbb{C}^{2^n}$, and $\alpha \in \{\pm 1, \pm i\}$ denotes a global phase factor that does not affect the measurement outcome of quantum states.  
This quotient group is isomorphic to the commutative group $\mathbb{F}_2^{2n}$, and the isomorphism for $E \in \mathcal{P}_n / \{\pm I, \pm i I\}$ is given as follows:  
\begin{align}
 &E=\alpha \bigotimes_{i=0}^{n-1} X^{x_i} Z^{z_i} \leftrightarrow (\xU|\zU)\notag\\&\hphantom{E}=\left(x_0, \ldots, x_{n-1} \mid z_0, \ldots, z_{n-1}\right)^\top.\label{003014_6Jan25}
\end{align}
Using this correspondence, we identify the Pauli error $E$ with its binary representation $(\xU|\zU)$. Whether the vector is a row vector or column vector should be understood from the context.

A very important subclass of stabilizer codes \cite{gottesmanthesis} is the CSS codes \cite{calderbank96,steane96}. CSS codes are a type of stabilizer code characterized by the property that the non-identity components of each stabilizer generator within a tensor product are all equal to $X$ or all equal to $Z$. 

The check matrix of the stabilizer group $S$ is  expressed as binary vectors arranged in rows. 
Without loss of generality, the check matrix $H$ of a CSS code of length $n$ can be expressed as follows:
\begin{align}
 H = 
\begin{pmatrix}
 H_X & O \\
 O & H_Z
 \end{pmatrix}\in \Fb_2^{(m_X+m_Z)\times 2n},\label{222417_12Jan25}
\end{align}
where $H_X, H_Z$ are binary matrices of sizes $m_X\times n$ and $m_Z\times n$, respectively. The commutativity of the stabilizer generators implies that $ H_X H_Z^\T = O$. For simplicity, this paper assumes $m=m_X=m_Z$. The condition $ H_X H_Z^\T = O$ is equivalent to the following:
\begin{align}
 C_Z^{\perp} \subset C_X, \quad C_X^{\perp} \subset C_Z,\label{013431_22Jan25}
\end{align}
where $C_X,C_Z$ are the code spaces defined by $H_X,H_Z$ when considered as parity-check matrices. 
Furthermore, the dimension $k$ of the CSS code of length $n$ can be determined by the following formula:
\[
k = n - \operatorname{rank} H_X - \operatorname{rank} H_Z.
\]

Let us consider the codeword state $|\psi\>$ with a Pauli error $E\leftrightarrow (\xU\mid \zU)$ applied, resulting in the state $E|\psi\>$. The decoder considered in this paper performs decoding according to the following steps:
\begin{enumerate}
 \item Measure the syndrome $(\sU=H_Z\xU, \tU=H_X\zU)\in \Fb_2^{2m}$.
 \item Based on the syndrome, estimate the noise as $\EH\leftrightarrow(\xUH\mid \zUH)$ that satisfies the following:
 \begin{align}
 H_Z\xUH=\sU, \quad H_X\zUH=\tU\label{030159_23Jan25}
 \end{align}
 \item 
The codeword $|\psi\>$ is affected by noise $E$, resulting in the state $E|\psi\>$. Applying $\EH^\dagger$ to this state yields $\EH^\dagger E|\psi\>$.
If $\EH^\dagger E\in S$, then $\EH^\dagger E|\psi\>\propto |\psi\>$, and the codeword is recovered. 
\end{enumerate}

The condition $\EH^\dagger E\in S$ is equivalent to the existence of $\aU=(\aU_X\mid \aU_Z)=(a_0,\ldots,a_{2m-1})^\T \in\Fb_2^{2m}$ such that 
\begin{align}
 \prod_{i=0}^{2m-1} S_i^{a_i}=\EH^\dagger E, 
\end{align}
where $S_i$ are the stabilizer generators for $i=0,\ldots,2m-1$. Moreover, since $\EH^\dagger E\leftrightarrow(\xU+\xUH\mid \zU+\zUH)$, this is equivalent to:
\begin{align}
\sum_{i=0}^{2m-1} a_i\sU_i=(\xU+\xUH\mid \zU+\zUH). 
\end{align}
where $S_i\leftrightarrow\sU_i$. Writing this in terms of matrices and vectors, we get $  H^\T\aU=(\xU+\xUH\mid \zU+\zUH)$. Specifically, for CSS codes, since the parity-check matrix is given by \eqref{222417_12Jan25}, we use $\aU=(\aU_X\mid\aU_Z)$ to obtain:
\begin{align}
  & (H_X)^\T \aU_X=\xU+\xUH, \quad (H_Z)^\T \aU_Z=\zU+\zUH\label{025725_23Jan25}
\end{align}
which is equivalent to 
\begin{align}
&G_X\left(\xU+\xUH\right)=0, \quad G_Z\left(\zU+\zUH\right)=0,\label{041406_21Jan25}
\end{align}
or equivalently, 
\begin{align}
& \xUH+\xU\in C_X^\perp,\quad \zUH+\zU\in  C_Z^\perp, \label{021354_20Jan25}
\end{align}
where $G_X,G_Z$ are generator matrices of $C_X,C_Z$, respectively. 
From \eqref{013431_22Jan25}, it follows that \eqref{021354_20Jan25} implies
\begin{align}
 H_X\left(\zU+\zUH\right)=0, \quad H_Z\left(\xU+\xUH\right)=0,\label{022112_22Jan25}
\end{align}
which is equivalent to \eqref{030159_23Jan25}. 

We assume a depolarizing channel with physical error rate \(p_D\). The hashing bound for the depolarizing channel with \(p_D\) is given by
\begin{align}
 1 - H_2(p_D) - p_D \log_2(3),
\end{align}
where \(H_2(\cdot)\) is the binary entropy function. 
\rev{Throughout this paper, this quantity is used only as an achievability benchmark derived from random-coding arguments for stabilizer/CSS codes, not as a converse bound. Solving this expression for the coding rates treated in this paper gives the corresponding depolarizing-channel reference values $p_D^* \simeq 0.10835$ ($10.835\%$) for $R = 1/3$, $p_D^* = 7.439\%$ for $R = 1/2$, $p_D^* = 5.598\%$ for $R = 3/5$, and $p_D^* = 4.454\%$ for $R = 2/3$. These values are also drawn in the figures below as hashing-bound benchmarks and are used only as asymptotically motivated references; finite-length simulations by themselves do not establish asymptotic attainment of the hashing bound.}
In this section, we describe the probability distribution of \(X\) and \(Z\) errors induced by the channel.

Let \(M = PJ\), \(N = PL\), and \(n = eN\). For \(n = eN\) qubits, the vector representations of \(X\) and \(Z\) errors are denoted by \(\xU\) and \(\zU\), respectively. We divide \(\xU\) and \(\zU\) into \(e\)-bit segments and write
\begin{align}
    \xU &= (x_0, \ldots, x_{N-1})^\T, \quad \zU = (z_0, \ldots, z_{N-1})^\T.
\end{align}
Each \(x_j\) and \(z_j\) is a binary vector of length \(e\). The probability of occurrence for \(\xU, \zU \in \Fb_2^{eN}\) in this channel is given by
\begin{align}
    p(\xU,\zU) &= \prod_{j=0}^{N-1}p(x_{j},z_{j}),\\
    p(x_{j},z_{j})&=\prod_{k=0}^{e-1} p(x_{j,k},z_{j,k}),  \\
    p(x_{j,k},z_{j,k}) &=
    \begin{cases}
        1-p_D, &  (x_{j,k},z_{j,k})=(0,0),\\
        \frac{p_D}{3}, & (x_{j,k},z_{j,k})=(0,1),(1,0),(1,1),
    \end{cases}
\end{align}
where \(x_{j,k}, z_{j,k} \in \Fb_2\) denote the \(k\)-th bit of \(x_j\) and \(z_j\), respectively, for \(k = 0,1,\ldots,e-1\).

\subsection{Degeneracy and Logical Correctness}\label{153729_6Jun25}
The decoding process involves estimating the noise vectors \(\xU, \zU \in \Fb_2^{n}\) from the syndromes $\sU,\tU\in \Fb_2^m$
\begin{align}
 \sU = H_Z \xU \tAND \tU = H_X \zU. \label{031203_14May25}
\end{align}
The decoder attempts to find a likely noise estimate $\xUH$ and $ \zUH$ that satisfy
\begin{align}
H_Z \xUH = \sU, \quad H_X \zUH = \tU, \label{163133_10May25}
\end{align}
that is, a vector consistent with the syndrome. Naturally, the true noise vectors \(\xU\) and \(\zU\) also satisfy this condition.

Such estimates $\xUH$ and $\zUH$ are not unique: the numbers of solutions are $|C_Z|$ and $|C_X|$, respectively. Indeed, Eq.~\eqref{163133_10May25} is equivalent to
\begin{align}
\xU + \xUH \in C_Z,\quad \zU+ \zUH \in C_X.   \label{180124_5Apr25}
\end{align}
Hence, the set of all solutions can be written as a coset:
\begin{align}
 \xU + C_Z &= \{ \xU + \xU' \mid \xU' \in C_Z \}, \notag\\\zU + C_X &= \{ \zU + \zU' \mid \zU' \in C_X \}. \label{180124_5Apr25_2}
\end{align}

In classical communication, syndrome decoding aims to identify the exact noise vector; the estimate must match the true error. However, in quantum error correction, exact identification is not necessary. Instead, if the following condition holds:
\begin{align}
\xU + \xUH \in C_X^\perp, \quad \zU + \zUH \in C_Z^\perp, \label{204646_19Mar25}
\end{align}
then a recovery operator constructed from \(\xUH\) and \(\zUH\) will correctly restore the quantum state, despite potential mismatches between the estimated and true noise vectors. An error that satisfies \eqref{163133_10May25} but not \eqref{204646_19Mar25}, in other words,
\begin{align}
\xU + \xUH &\in C_Z \setminus C_X^\perp, \notag\\
\text{or}\quad \zU + \zUH &\in C_X \setminus C_Z^\perp,
\end{align}
is called a logical error, as it transforms the state into a different codeword state.  
That is, the recovery still satisfies the stabilizer constraints but results in a misidentification of the logical state.  
\subsection{Conventional Algorithm: Joint Belief Propagation Decoding over Binary Vectors}\label{subsec:205908_25May25}
The idea of the joint BP decoding algorithm was, to the best of the authors' knowledge, proposed in~\cite{mackay2004sparse}.  
Although the algorithm was not explicitly described in that work, it was later formulated as a message-passing algorithm in~\cite{komoto2024quantumerrorcorrectionnear}, or perhaps even earlier in other sources.  

The algorithm is employed to simultaneously estimate the noise vectors \(\xU\) and \(\zU\). In that formulation, the noise components \(x_j\) and \(z_j\) are treated as random variables over binary vectors, and the decoding algorithm is driven by a factorization of the joint posterior distribution. Specifically, the factor graph is constructed based on a function proportional to the posterior probability, which is used to guide the message passing in the BP decoder.

The posterior probability of \(\xU, \zU\) given the syndromes \(\sU, \tU\) can be sparsely factorized as follows:
\begin{align}
 p(\xU,\zU|\sU,\tU)
 &\propto \I\bigl[H_Z\xU=\sU\bigr]\I\bigl[H_X\zU=\tU\bigr]p(\xU,\zU)
 \\&=
 \Bigl(\prod_{i\in[M]} \I\bigl[\sum_{j\in[N]}(H_Z)_{ij}x_j=s_i\bigr] \Bigr)
 \notag\\&\quad\times \Bigl( \prod_{i\in[M]} \I\bigl[\sum_{j\in[N]}(H_X)_{ij}z_j=t_i\bigr] \Bigr)
 \notag\\&\quad\times \Bigl( \prod_{j\in[N]} p\left(x_j, z_j\right) \Bigr),
  \label{align_p(x,z|s,t)}
\end{align}
where $(H_X)_{ij}$ and $(H_Z)_{ij}$ denote the $(i,j)$-th $(e \times e)$ submatrices of $H_X$ and $H_Z$, respectively. These submatrices are defined by the companion matrix mapping:
\begin{align}
(H_X)_{ij} = A(\gamma_{ij}), \qquad
(H_Z)_{ij} = A(\delta_{ij})^\top,
\end{align}
where $\gamma_{ij}, \delta_{ij} \in \Fb_q$ are the entries of the non-binary parity-check matrices $H_\Gamma$ and $H_\Delta$, respectively, and $A(\cdot)$ denotes the $e \times e$ binary companion matrix associated with each field element in $\Fb_q$ (see Appendices~\ref{234926_5Apr25} and \ref{234956_5Apr25} for details).

We denote $\I[\cdot]$ as 1 if the inside the brackets is true, and 0 otherwise.
For each \( i \), note that there are \( L \) values of \( j \) for which \( (H_X)_{ij} \) and \( (H_Z)_{ij} \) are non-zero, respectively.
In \cite{komoto2024quantumerrorcorrectionnear}, this factorization was used to formulate the joint BP decoding algorithm.
\subsection{Joint Belief Propagation Decoding over Finite Field}\label{sec:143825_6Jun25}
In the concluding part of Section~\ref{sec:171047_25May25}, we will propose a post-processing technique for the joint BP algorithm.  
For the sake of a concise formulation, this section describes a joint BP algorithm over a finite field, which is equivalent to the one over binary vectors defined in Section~\ref{subsec:205908_25May25}.

For $x_j, z_j \in \Fb_2^e$, we define \(\xi_j, \zeta_j \in \Fb_q\) such that  
\begin{align}
 \wU(\xi_j) = x_j \tAND \vU(\zeta_j) = z_j, \label{031931_14May25}
\end{align}
respectively. Similarly, we define \(\sigma_i, \tau_i \in \Fb_q\) such that \(\wU(\sigma_i) = s_i\) and \(\vU(\tau_i) = t_i\). Then, by Theorems~\ref{013653_26Mar25} and~\ref{wU_linear_relation_theorem} in the Appendices, we obtain the following equivalences:
\begin{align}
\left.
\begin{aligned}
H_\Delta\, \xiU &= \sigmaU 
\quad &&\Longleftrightarrow \quad 
H_Z\, \xU = \sU, \\
H_\Gamma\, \zetaU &= \tauU 
\quad &&\Longleftrightarrow \quad 
H_X\, \zU = \tU
\end{aligned}
\right\}.
\label{eq:binary_nonbinary_equiv}
\end{align}

Using the equivalence in \eqref{eq:binary_nonbinary_equiv}, the decoding constraint \eqref{align_p(x,z|s,t)} in terms of binary variables \(\xU, \zU\) can be equivalently rewritten using $\Fb_q$-valued variables \(\xiU, \zetaU \in \Fb_q^N\) as follows:
\begin{align}
 p(\xU,\zU|\sU,\tU)&\propto \I\bigl[H_Z\xU=\sU\bigr]\I\bigl[H_X\zU=\tU\bigr]p(\xU,\zU)
 \\ &= \I\bigl[H_\Delta\xiU=\sigmaU\bigr]\I\bigl[H_\Gamma\zetaU=\tauU\bigr]p(\xiU,\zetaU)
 \\&=
 \Bigl(\prod_{i\in[M]} \I\bigl[\sum_{j\in[N]}\delta_{ij}\xi_j=\sigma_i\bigr] \Bigr)
 \notag\\&\quad\times \Bigl(\prod_{i\in[M]} \I\bigl[\sum_{j\in[N]}\gamma_{ij}\zeta_j=\tau_i\bigr] \Bigr)
 \notag\\&\quad\times \Bigl(\prod_{j\in[N]} p\left(\xi_j, \zeta_j\right) \Bigr)\label{022413_12May25}
 \\&\propto p(\xiU,\zetaU|\sigmaU,\tauU).
  \label{233428_11May25}
\end{align}
Here, the pairwise prior $p\left(\xi_j, \zeta_j\right)$ is defined as $p\left(\xi_j, \zeta_j\right)\defeq p(x_j, z_j)$ for a pair $(x_j, z_j)$ satisfying the mapping condition described in~\eqref{031931_14May25}.

The joint BP algorithm that approximates the marginals of the posterior probability \(p(\xiU,\zetaU|\sigmaU,\tauU)\) is an algorithm that updates the following eight types of messages:
\begin{align}
& \mu^{(\ell),X}_{ji}(\xi_j),  \nu^{(\ell),X}_{ij}(\xi_j), \lambda^{(\ell),X}_j(\xi_j),  \kappa^{(\ell),X}_j(\xi_j), 
\\&\mu^{(\ell),Z}_{ji}(\zeta_j),\nu^{(\ell),Z}_{ij}(\zeta_j), \lambda^{(\ell),Z}_j(\zeta_j),\kappa^{(\ell),Z}_j(\zeta_j),
\end{align}
at the iteration round $\ell=0,1,\ldots$. 
These messages are probability vectors of length~$q$. 
Note that although we use Greek letters to denote these messages, they represent probability distributions and not elements of the finite field.
The following messages are initialized to the uniform distribution over $\xi_j, \zeta_j\in \Fb_q$, respectively:
\begin{align}
& \lambda_j^{(0),X}(\xi_j) = 1/q, \\
& \lambda_j^{(0),Z}(\zeta_j) = 1/q,\\
& \nu_{ij}^{(0),X}(\xi_j) = 1/q, \\
& \nu_{ij}^{(0),Z}(\zeta_j) = 1/q.
\end{align}

Each message is updated at each iteration round $\ell = 0, 1, 2, \ldots$ using the following update equations.
After each update, the messages are normalized to ensure they form probability distributions.
\begin{align}
 \kappa_{j}^{(\ell),X}(\xi_j)&=\sum_{\zeta_j}p(\xi_j,\zeta_j)\lambda_j^{(\ell),Z}(\zeta_j)\label{013318_27May25},
\\ \kappa_{j}^{(\ell),Z}(\zeta_j)&=\sum_{\xi_j}p(\xi_j,\zeta_j)\lambda_j^{(\ell),X}(\xi_j)\label{013321_27May25},
\\\mu^{(\ell),X}_{ji}(\xi_j)&=\kappa^{(\ell),X}_j(\xi_j)\prod_{i'\in \partial_Z(j)\setminus i}\nu_{i'j}^{(\ell),X}(\xi_j),\label{noise_X}
\\\mu^{(\ell),Z}_{ji}(\zeta_j)&=\kappa^{(\ell),Z}_j(\zeta_j)\prod_{i'\in \partial_X(j)\setminus i}\nu_{i'j}^{(\ell),Z}(\zeta_j),\label{noise_Z}
\\  \nu_{ij}^{(\ell+1),X}(\xi_j) &=\sum_{(\xi_{j'}):{j'\in \partial_Z(i)}\setminus j}
    \I\Bigl[\sum_{j'\in \partial_Z(i)}\delta_{ij'}\xi_{j'}=\sigma_i\Bigr]
    \notag\\&\qquad\times \prod_{j'\in \partial_Z(i)\setminus j}\mu_{j'i}^{(\ell),X}(\xi_{j'}),\label{parity_check_X}
\\  \nu_{ij}^{(\ell+1),Z}(\zeta_j)&=\sum_{(\zeta_{j'}):{j'\in \partial_X(i)}\setminus j}
    \I\Bigl[\sum_{j'\in \partial_X(i)}\gamma_{ij'}\zeta_{j'}=\tau_i\Bigr]
    \notag\\&\qquad\times \prod_{j'\in \partial_X(i)\setminus j}\mu_{j'i}^{(\ell),Z}(\zeta_{j'}),\label{parity_check_Z}
\\ \lambda_{j}^{(\ell+1),X}(\xi_j)&=\prod_{i\in \partial_Z(j)}\nu_{ij}^{(\ell+1),X}(\xi_j), 
\\ \lambda_{j}^{(\ell+1),Z}(\zeta_j)&=\prod_{i\in \partial_X(j)}\nu_{ij}^{(\ell+1),Z}(\zeta_j),
\end{align}
where $\partial_X(i)\subset[N]$ and $\partial_Z(i)\subset[N]$ are the set of column index $j$ such that $\gamma_{ij}\neq 0$ and $\delta_{ij}\neq 0$, respectively.
Similarly,  $\partial_X(j)\subset[M]$ and $\partial_Z(j)\subset[M]$ are the set of row index $i$ such that $\gamma_{ij}\neq 0$ and $\delta_{ij}\neq 0$, respectively.
It holds that 
$\#\partial_X(i)=\#\partial_Z(i)=L$ and 
$\#\partial_X(j)=\#\partial_Z(j)=J$.
The outer summation in \eqref{parity_check_X} is taken over all vectors $(\xi_{j'})$ of length $L - 1$,  
where $j' \in \partial_Z(i) \setminus \{j\}$ and each component $\xi_{j'}$ ranges over $\Fb_q$.  
In this summation, the variable $\xi_j$ is fixed, and the sum is taken over all other variables involved in the parity check at node $i$.
The interpretation of the summation in \eqref{parity_check_Z} is analogous, with $\zeta_j$ fixed and the summation taken over all $\zeta_{j'}$ for $j' \in \partial_X(i) \setminus \{j\}$.

We define the estimated noise at iteration $\ell$ as
\begin{align}
  &  \xiH_j^{(\ell)} \defeq \argmax_{\xi_j}     \kappa^{(\ell),X}_j(\xi_j)  \prod_{i \in \partial_Z(j)} \nu_{ij}^{(\ell),X}(\xi_j), \quad
\\&  \zetaH_j^{(\ell)} \defeq \argmax_{\zeta_j} \kappa^{(\ell),Z}_j(\zeta_j)\prod_{i \in \partial_X(j)} \nu_{ij}^{(\ell),Z}(\zeta_j).
\end{align}
At each iteration $\ell$, we check whether these estimated syndromes match the true syndromes $\sigmaU$ and $\tauU$.
If both conditions $H_\Delta \xiUH^{(\ell)} = \sigmaU$ and $H_\Gamma \zetaUH^{(\ell)} = \tauU$ are satisfied, then $\xiUH^{(\ell)}$ and $\zetaUH^{(\ell)}$ are accepted as the estimated noise, and the algorithm terminates.
If either condition is not satisfied, the decoding continues.

The messages \eqref{parity_check_X} and \eqref{parity_check_Z} can be computed using FFT \cite{DaveyMacKayGFq}, requiring \(O(Lq\log q)\) operations per message. 
On the other hand the computation for messages \eqref{013318_27May25} and \eqref{013321_27May25} require \(O(q^2)\) operations per message.  
Other messages computation require only $O(Jq)$ operations per message. 
Moreover, since the message updates at each node do not share memory, they can be computed in parallel across all nodes.  
For fixed \(q,J,L\), one BP iteration costs \(O(n)\), where \(n=ePL\). With a maximum of \(T\) iterations, the BP stage costs \(O(Tn)\); this is linear in \(n\) when \(T\) is independent of the block length.  
In the experiments conducted in this paper, we set \(q = 2^e\) with \(e = 8\).

\begin{example}\label{ex:030406_16May25}
Based on the $16 \times 48$ binary matrices $H_\Gamma$ and $H_\Delta$ given in Example~\ref{035757_2Jun25},  
we illustrate the factor graph associated with Eq.~\eqref{022413_12May25}, as shown in Fig.~\ref{025851_16May25}.  
The Tanner graphs on the left and right correspond to $H_\Delta$ and $H_\Gamma$, respectively.  
The edges that form the cycle included in the UTCBC presented in Example~\ref{035757_2Jun25} are highlighted with thicker lines in their corresponding colors.  
Each side contains 16 parity-check nodes.  
The variable nodes representing the noise variables $\xi_j$ and $\zeta_j$ are placed on each side, with 48 nodes per side.  
The central factor nodes represent the joint distribution $p(\xi_j, \zeta_j)$ of the noise variables.  
The directions of the messages passed along the edges are indicated at the bottom of the graph.  
\end{example}

\definecolor{c00}{rgb}{1.00,0.00,0.00} %
\definecolor{c10}{rgb}{1.00,0.50,0.00} %
\definecolor{c01}{rgb}{1.00,1.00,0.00} %
\definecolor{c11}{rgb}{0.00,1.00,0.00} %
\definecolor{c02}{rgb}{0.00,1.00,1.00} %
\definecolor{c12}{rgb}{0.00,0.00,1.00} %
\definecolor{c03}{rgb}{0.50,0.00,1.00} %
\definecolor{c13}{rgb}{1.00,0.00,1.00} %
\definecolor{c04}{rgb}{0.50,0.50,0.50} %
\definecolor{c14}{rgb}{0.75,0.25,0.25} %
\definecolor{c05}{rgb}{0.25,0.75,0.25} %
\definecolor{c15}{rgb}{0.25,0.25,0.75} %

\input{factor_graph}

Using the conventional code construction method described in Section~\ref{192200_2Jun25},  we consider codes defined by $(H_X, H_Z)$, or equivalently $(H_\Gamma, H_\Delta)$. 
Figure~\ref{fig:comparison_ISIT} shows the decoding performance  by the conventional decoding method presented in Section~\ref{sec:143825_6Jun25}.
In this experiment, degeneracy was not taken into account, and decoding is considered successful only when the estimated noise exactly matches the true noise, i.e., \((\xiUH = \xiU, \zetaUH = \zetaU)\).  
For codes with \(L \ge 8\) and rate \(R \ge 0.5\), no significant error floor is observed, and the performance approaches the hashing bound.  
On the other hand, for \(L = 6\) and rate \(R = 1/3\), a high error floor emerges as the code length increases.  
\rev{For reference, the corresponding hashing-bound benchmark values are $p_D^* = 7.439\%$ for $R=1/2$, $p_D^* = 5.598\%$ for $R=3/5$, and $p_D^* = 4.454\%$ for $R=2/3$.}
In the next section, we discuss the structural properties of the codes that lead to this error floor.  

\begin{figure*}[t]
  \centering
  \includegraphics[width=0.95\linewidth]{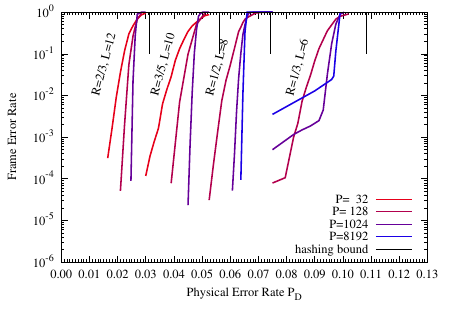}
  \caption{Decoding performance of the conventional decoder applied to the conventional code proposed in~\cite{komoto2024quantumerrorcorrectionnear}
  for several row weights $L$ and coding rates $R = 1 - 2J/L$.
  The field size is $q = 2^e$ with $e = 8$, and the number of physical qubits is given by $n = eN=ePL$.
  When $L > 6$ and $R > 1/3$, no observable error floor appears down to a frame error rate (FER) of $10^{-4}$ as the blocklength increases.
  In contrast, for $L = 6$ and $R = 1/3$, a pronounced error floor emerges as the blocklength increases.}
  \label{fig:comparison_ISIT}
\end{figure*}

\subsection{Observations of Decoder Behavior in the Error Floor Regime}\label{150214_6Jun25}
To characterize the typical behavior of the joint BP algorithm in the error floor regime, we begin by defining the following metrics.  
The decoder attempts to find the noise vector estimates $\xiUH$ and $\zetaUH$ such that 
\begin{align}
\text{$\sigmaU = H_\Delta \xiUH$ and $\tauU = H_\Gamma \zetaUH$.  }
\end{align}
In the following, we focus on the estimation of $X$-noise $\xiUH$.  
The estimation of $Z$-noise $\zetaUH$ can be handled in an analogous manner.  

For the output of the joint BP algorithm at the $\ell$-th iteration, the set of indices at which the estimated noise $\hat{\xiU}^{(\ell)}$ is incorrect is defined by  
\begin{align}
K_\mathrm{err}^{(\ell)} = \{ j \in [N] \mid \xiH_j^{(\ell)} \neq \xi_j \}.
\end{align}
Write \(\hat{\sigmaU}^{(\ell)}=H_\Delta\xiUH^{(\ell)}\). The set of unsatisfied checks is
\begin{align}
I_\mathrm{err}^{(\ell)}=\{i\in[M]\mid\hat\sigma_i^{(\ell)}\neq\sigma_i\}.
\end{align}
For a history length \(d\ge1\), define
\[
\mathcal W_\ell=\{r\in\mathbb Z:\max(1,\ell-d+1)\le r\le\ell\}.
\]
At initialization, \(\mathcal W_0\) is empty. Using the initialized hard decision \(\xiUH^{(0)}\), the set of variable nodes that changed during this window is
\begin{align}
K_d^{(\ell)}=\{j\in[N]\mid
\xiH_j^{(r-1)}\neq\xiH_j^{(r)}
\text{ for some }r\in\mathcal W_\ell\}.
\label{022343_27May25}
\end{align}
The set of checks that were unsatisfied at least once during the same window is
\begin{align}
I_d^{(\ell)}=\bigcup_{r\in\mathcal W_\ell} I_\mathrm{err}^{(r)}.
\label{022915_27May25}
\end{align}
Note that among the quantities defined above, the decoder cannot directly observe \(K_\mathrm{err}^{(\ell)}\).  
In contrast, all the other variables-$K_d^{(\ell)}$, $I_\mathrm{err}^{(\ell)}$, and $I_d^{(\ell)}$-can be computed from the internal state and history of the joint BP decoder.

\input{decoding_stats_table}

Table~\ref{tab:full_decoding_stats} provides an example of a decoding state transition.  
The example is based on a code constructed by the proposed method and decoded using the conventional joint BP algorithm over a depolarizing channel with physical error rate $p_D = 9.435\%$.  
Further details are provided in Section~\ref{134114_16Jun25}.  
The performance for other values of $p_D$ is shown as the solid blue curve in Fig.~\ref{fig:comparison}.  
The decoder is applied to a code with parameters $J = 2$, $L = 6$, and $P = 6500$, and aims to estimate both $\xiU$ and $\zetaU$.  

In the estimation of Z-noise vector $\zetaU$, the sizes of $K_\mathrm{err}^{(\ell)}$ and $I_\mathrm{err}^{(\ell)}$ both decrease monotonically to zero as the iteration index $\ell$ increases. The decoder successfully estimates the noise at iteration $\ell = 47$. Although typical cases exhibit similar success in estimating X-noise vector $\xiU$, this particular example illustrates a rare failure of joint BP decoding to correctly estimate the noise vector $\xiU$, occurring with probability on the order of $10^{-2}$.

In the estimation of $\xiU$, the sizes of $K_\mathrm{err}^{(\ell)}$ and $I_\mathrm{err}^{(\ell)}$ initially decrease monotonically. However, at a certain point, this decrease stagnates: although the number of incorrect estimates continues to decline, a small number of errors persist and never vanish, even after additional iterations. In this example, as $\ell$ increases, both $K_d^{(\ell)}$ and $I_d^{(\ell)}$ converge to a fixed set of size six. Within these sets, the values of $K_\mathrm{err}^{(\ell)}$ and $I_\mathrm{err}^{(\ell)}$ continue to fluctuate across iterations.
\begin{align}
K_d^{(\ell)} &= \{[4062, 0]_\Lrm, [10410, 1]_\Lrm,\ [13890, 2]_\Lrm,\notag\\&\quad [25420, 0]_\Rrm, [31699, 1]_\Rrm, [35508, 2]_\Rrm\}, \\
I_d^{(\ell)} &= \{[2853, 0], [2922, 0], [3490, 0],\notag\\&\quad [6662, 1], [9201, 1], [12902, 1]\}.
\end{align}
The numbers in parentheses indicate the corresponding block index for each variable or check node.

We examined the positions of the six variable nodes that failed to be correctly decoded and the corresponding unsatisfied check nodes in the matrix $H_\Delta$. Specifically, we analyzed the submatrix of $H_\Delta$ formed by the rows indexed by $I_\mathrm{err}^{(\ell)}$ and the columns indexed by $K_\mathrm{err}^{(\ell)}$. The entries of this submatrix are as follows:
\begin{figure*}[!t]
\begin{align}
\begin{array}{ccccccc}
 &[4062,0]_\Lrm   &[10410,1]_\Lrm&[13890,2]_\Lrm&[25420,0]_\Rrm&[31699,1]_\Rrm  & [35508,2]_\Rrm \\%
\relax[2853,0]&\tikzmark{4A}\alpha^{77}&0&0&0&\tikzmark{4B}\alpha^{79}&0\\
 \relax[2922,0]& 0&\tikzmark{4E}\alpha^{119}&  0&\tikzmark{4F}\alpha^{168}&  0&  0     \\
 \relax[3490,0]& 0&  0&\tikzmark{4I}\alpha^{179}&  0&  0&\tikzmark{4J}\alpha^{235}     \\
 \relax[6662,1]&\tikzmark{4L}\alpha^{65}&  0&  0&  0&  0&\tikzmark{4K}\alpha^{145}     \\
 \relax[9201,1]& 0&\tikzmark{4D}\alpha^{121}&  0&  0&\tikzmark{4C}\alpha^{174}&  0     \\
 \relax[12902,1]& 0&  0& \tikzmark{4H}\alpha^{92}& \tikzmark{4G}\alpha^{66}&  0&  0     
\end{array}
\begin{tikzpicture}[remember picture, overlay]
  \draw[thick, blue!50, ->] (pic cs:4A) -- (pic cs:4B);
  \draw[thick, blue!50, ->] (pic cs:4B) -- (pic cs:4C);
  \draw[thick, blue!50, ->] (pic cs:4C) -- (pic cs:4D);
  \draw[thick, blue!50, ->] (pic cs:4D) -- (pic cs:4E);
  \draw[thick, blue!50, ->] (pic cs:4E) -- (pic cs:4F);
  \draw[thick, blue!50, ->] (pic cs:4F) -- (pic cs:4G);
  \draw[thick, blue!50, ->] (pic cs:4G) -- (pic cs:4H);
  \draw[thick, blue!50, ->] (pic cs:4H) -- (pic cs:4I);
  \draw[thick, blue!50, ->] (pic cs:4I) -- (pic cs:4J);
  \draw[thick, blue!50, ->] (pic cs:4J) -- (pic cs:4K);
  \draw[thick, blue!50, ->] (pic cs:4K) -- (pic cs:4L);
  \draw[thick, blue!50, ->] (pic cs:4L) -- (pic cs:4A);
\end{tikzpicture}
\end{align}
\end{figure*}
Here, the notations $[i, b]=bP+i$, $[i, b]_\Lrm=bP+i$, and $[i, b]_\Rrm=(b+L/2)P+i$ are the same as those introduced in the proof of Theorem~\ref{053905_28Apr25}.  
This submatrix forms a cycle of length 12 $(=2L)$. 
From the order in which the blocks appear, we can identify this as a Tanner cycle belonging to an unavoidable TCBC $\uU(1)$ in $H_\Delta$. 

We have also observed other failure patterns, such as compound graphs formed by two length-12 cycles, or configurations trapped in length-16 cycles composed of \(8 \times 8\) submatrices, which occur with a probability roughly two orders of magnitude lower.
The objective of this section is to propose a post-processing algorithm that rescues the joint BP decoder from such cycle-induced failures where decoding does not succeed due to being trapped in one or more cycles.

In the remainder of this section, we consider a scenario in which the estimate of the X-noise has not yet converged-namely, it is trapped in a combination of a few length-12 cycles-resulting in
\begin{align}
 H_\Delta \xiUH^{(\ell)} \neq \sigmaU.
\end{align}
Failures in estimating Z-noise can be handled in a symmetric manner.

Let $\Cc^{(2L)}$ denote the set of all Tanner cycles of length $2L$ contained in $H_\Delta$ (the definition for $H_\Gamma$ is analogous). Based on Condition~\ref{000952_6Jun25} in Section~\ref{sec:000508_6Jun25}, the matrix $H_\Delta$ constructed by the proposed method is designed so that all such cycles belong to the unavoidable TCBCs $\uU(k)$ for $k \in \{0,1,2\}$. In other words, $\Cc^{(2L)}$ consists solely of the Tanner cycles contained in the unavoidable TCBCs $\uU(k)$ for $k \in \{0,1,2\}$. Since there are $P$ cycles for each value of $k$, the total number of such cycles is $3P$, and hence the size of $\Cc^{(2L)}$ is $3P$.

For a cycle $\Ct \in \Cc^{(2L)}$ of length $2L$, let $K(\Ct)$ and $I(\Ct)$ denote its column index set and row index set, respectively.  
For a collection $\Cc$ of such cycles, we define the combined column and row index sets as  
\begin{align}
  K(\Cc) &\triangleq \bigcup_{\Ct \in \Cc} K(\Ct) = K\Bigl( \bigcup_{\Ct \in \Cc} \Ct \Bigr), \\
  I(\Cc) &\triangleq \bigcup_{\Ct \in \Cc} I(\Ct) = I\Bigl( \bigcup_{\Ct \in \Cc} \Ct \Bigr).
\end{align}
We say that joint BP at iteration $\ell$ or the estimated noise $\hat{\xiU}^{(\ell)}$ is \emph{trapped} in the union of cycles $\Cc \subset \Cc^{(2L)}$ if the following holds:
\begin{align}
  K_\mathrm{err}^{(\ell)} \subset K(\Cc). \label{015037_27May25}
\end{align}
That is, all indices where the estimated noise is incorrect lie within the column support of the cycles in $\Cc$,  
that is,
\begin{align}
  \xi_j \neq \xiH_j^{(\ell)} \implies j \in K(\Cc).
\end{align}
Conversely, all indices outside $K(\Cc)$ correspond to correctly estimated noise values,  
that is,
\begin{align}
  j \notin K(\Cc) \implies \xi_j = \xiH_j^{(\ell)}.
\end{align}
In~\cite{kasai2025efficient}, only the special case where $\Cc$ consists of a single cycle (i.e., $|\Cc| = 1$) was considered.  

In the error floor region, experiments using the conventional decoding method on the proposed codes revealed the following observation:  
when the decoder stagnated, the estimated noise was typically trapped in the union of a small number of cycles, most of which had length $2L$, and occasionally of length $2(L+2)$ with low probability.  
In this work, we focus on the case where the estimated noise is trapped in a union of length-$2L$ cycles, and propose a post-processing method to address this situation.  
In the next section, we address the problem of estimating the set of cycles in which the decoder becomes stuck and trapped during decoding.  
\subsection{Estimation of the Cycle Set $\Cc$}\label{150240_6Jun25}
In this section, we describe an algorithm for estimating a small subset $\Cc \subset \Cc^{(2L)}$ intended to satisfy condition~\eqref{015037_27May25}, to be applied when the joint BP decoder stagnates after failing to find a noise estimate that satisfies the syndrome condition. Naturally, the decoder does not have access to the true noise $\xiU$ or to the set $K_\mathrm{err}^{(\ell)}$ of misestimated indices.
Instead, we use $K_d^{(\ell)}$, the set of variable nodes whose estimated values have changed during the last $d$ iterations (see~\eqref{022343_27May25}), which can be computed by the decoder.  
Based on $K_d^{(\ell)}$, we estimate the cycle set $\Cc$ according to the following strategy:
\begin{enumerate}
  \item If $K_d^{(\ell)}$ and $K(\Ct)$ share at least two elements, then include $\Ct$ in $\hat{\Cc}$.
  \item Limit the number of cycles in $\hat{\Cc}$ to at most $u$.
\end{enumerate}
The parameter $u$ caps the number of returned cycles.

Algorithm~\ref{021116_27May25} gives the procedure, with $u=2$ in the numerical setting considered here. Since a union of at most $u$ cycles contains at most $uL$ columns, an input with $|K_d^{(\ell)}|>uL$ can be rejected immediately without changing the possible outputs. With precomputed variable-to-cycle incidence lists, the remaining work is bounded independently of $n$ for fixed $u,L$: each column belongs to one cycle from each of the three UTCBC families. Constructing these incidence lists and maintaining the BP history require global work and are separate from the local search.
\begin{algorithm}
\caption{Cycle Set Estimation Based on Recent Variable Node Changes}\label{021116_27May25}
\begin{algorithmic}[1]
\If{$K_d^{(\ell)}=\varnothing$ or $|K_d^{(\ell)}|>uL$}
  \State \Return \textbf{null}
\EndIf
\State $\hat{\Cc}\gets\varnothing$
\ForAll{$j \in$ \texttt{$K^{(\ell)}_d$}}
  \ForAll{$\Ct \in \Cc^{(2L)}$ such that $j\in K(\Ct)$}
    \If{$|K(\Ct)\ \cap K^{(\ell)}_d|\ge 2$}
      \State add $\Ct$ to $\hat{\Cc}$
      \If{$|\hat{\Cc}|>u$}
        \State \Return \textbf{null}
      \EndIf
      \If{$K_d^{(\ell)}\subset K(\hat{\Cc})$}
        \State \Return $\hat{\Cc}$
      \EndIf
    \EndIf
  \EndFor
\EndFor
\State \Return \textbf{null}
\end{algorithmic}
\end{algorithm}
\subsection{Proposed Post-Processing Decoder}\label{150247_6Jun25}
In this section, we propose a post-processing decoding algorithm that supplements the noise estimation when the joint BP decoder stagnates after failing to find a noise estimate that satisfies the syndrome condition.
Note that even if the estimate differs from the true noise  $\xiU \neq \xiUH$, it may still be harmless for codeword recovery.  
Whether decoding is successful is determined by whether the condition $\xiU + \xiUH \in C_\Gamma^\perp$ is satisfied.

In the previous section, we identified a set of cycles $\hat{\Cc}$ from the estimated noise $\xiUH^{(\ell)}$ and its history.  
For brevity, we define $K := K(\hat{\Cc})$. 
Let $\KO := [N] \setminus K$, and denote by $\xiUH^{(\ell)}_{\KO}$ and $(H_\Delta)_{\KO}$ the restriction of $\xiUH^{(\ell)}$ and $H_\Delta$ to the components indexed by $\KO$, respectively.
We have:
\begin{align}
  &\sigmaU = H_\Delta \xiU = (H_\Delta)_K \xiU_K + (H_\Delta)_{\KO} \xiU_{\KO} \label{201851_26May25}
\end{align}

Suppose the estimation $\hat{\Cc}$ is correct; that is, the joint BP algorithm is trapped in $K$, and the noise has been correctly estimated on the complement set $\KO$, i.e., $K_\mathrm{err}^{(\ell)} \subset K$. Equivalently, the following holds:
\begin{align}
 \xiUH^{(\ell)}_{\KO} = \xiU_{\KO} \label{210643_16Mar25}
\end{align}
Then, from \eqref{201851_26May25}, it follows that:
\begin{align}
  &\sigmaU = (H_\Delta)_K \xiU_K + (H_\Delta)_{\KO} \xiUH_{\KO}^{(\ell)}\label{222023_26May25}
\end{align}
Our proposed post-processing decoder solves this linear equation for $\xiU_K$, and assigns the solution as $\xiUH_K$.  
Together with $\xiUH_{\KO}:=\xiUH_{\KO}^{(\ell)}$, this yields the final estimate $\xiUH$.  

For a support \(K\subseteq[N]\), define the zero-extended restricted kernel
\begin{align}
\mathcal N_K(H_\Delta)=\{\mathbf r\in\Fb_q^N:
H_\Delta\mathbf r=\zeroU,\ 
\operatorname{supp}(\mathbf r)\subseteq K\}.
\label{eq:restricted_kernel}
\end{align}

\begin{pp}[Conditional correctness of local refinement]
\label{prop:conditional_refinement}
Suppose that \(\xiUH^{(\ell)}_{\KO}=\xiU_{\KO}\). Every solution of \eqref{222023_26May25}, completed by this fixed complement, is logically equivalent to the true noise if and only if
\begin{align}
\mathcal N_K(H_\Delta)\subseteq C_\Gamma^\perp.
\label{eq:kernel_stabilizer_condition}
\end{align}
The analogous statement holds with \(\Gamma,\Delta\) and \(\zetaU,\xiU\) interchanged.
\end{pp}
\begin{proof}
The true restricted noise is a solution because the complement is correct. Any other solution differs from it by an element of \(\mathcal N_K(H_\Delta)\). Under \eqref{eq:kernel_stabilizer_condition}, every such difference is a stabilizer by \eqref{eq:dual_binary_images}, so every solution is logically correct. Conversely, if \(\mathbf r\in\mathcal N_K(H_\Delta)\setminus C_\Gamma^\perp\), adding \(\mathbf r\) to the true noise gives another solution with the same complement and a nontrivial logical residual.
\end{proof}

From Theorem~\ref{192838_2Jun25}, it is known that any cycle $\Ct$ in the UTCBCs $\uU(0)$ and $\uU(1)$ does not contain any logical errors; that is, $N(\Ct;H_\Delta) \subset C_\Gamma^\perp$.  
As discussed in Discussion~\ref{234642_5Jun25}, the proposed code construction assigns the $\Fb_q$-valued entries so that any cycle in the UTCBC $\uU(2)$ also avoids logical errors; specifically, $N(\Ct;H_\Delta) = \{\zeroU\}$.  
These individual-cycle properties do not extend automatically to a union of cycles.

The total number of equations in \eqref{222023_26May25} that involve any of the variables $\xi_j$ for $j \in K$ is at most $|K|$.  
The size of $K$ is upper bounded by $|K| \le uL$, where $u$ is the number of cycles in $\Cc$.  
In the experiments described later, we set $u = 2$.  
Thus, the size of the resulting system of equations is quite small compared to the total number of physical qubits $n$.  
Both the algorithm for checking whether a solution exists and the algorithm for computing one such solution can be executed with computational complexity $O(|K|^3)$.  
The linear equation admits a solution when \eqref{210643_16Mar25} holds. Logical correctness of every solution additionally requires \eqref{eq:kernel_stabilizer_condition}.  

The right-hand side can be formed from the residual syndrome already computed by BP:
\begin{align}
\sigmaU+(H_\Delta)_{\KO}\xiUH_{\KO}^{(\ell)}
&=\sigmaU+H_\Delta\xiUH^{(\ell)}\notag\\
&\quad +(H_\Delta)_K\xiUH_K^{(\ell)}.
\end{align}

\rev{The roles of the three ingredients considered in this paper are different. The joint BP stage performs global soft-information exchange on the entire Tanner graph and is responsible for the main linear-in-$n$ decoding process. We did not implement BP+OSD as a numerical baseline in this work; the following comparison is an algorithmic and complexity-level comparison with the natural OSD-type reprocessing one would apply to the present non-binary CSS setting. The central distinction is global versus local post-processing. A natural BP+OSD decoder for the present code would be a global, non-local method: after BP produces reliability information, it would reorder variables across the whole block, perform Gaussian elimination on a block-scale parity-check or generator representation, and enumerate test patterns on an information set~\cite{fossorier2002soft,panteleev2021degenerate}. The elimination step alone is cubic in the block length in a straightforward implementation, i.e., $O(n^3)$, before accounting for the additional OSD candidate enumeration. By contrast, the proposed degeneracy-aware refinement is a local repair tailored to the observed BP failure mechanism. For the proposed code with $L=6$, the BP-only error-floor events examined here include trapping on length-$12$ and length-$16$ cycles. The refinement does not reprocess the whole block; it first identifies the trapped cycle set from the BP history and then solves a linear system only on the corresponding restricted support $K(\hat{\Cc})$. Since $|K|\le uL$ and $J$, $L$, and $u$ are fixed constants in the present family, the restricted linear-solve cost is $O(|K|^3)=O(1)$ with respect to the block length $n$.}

Algorithm~\ref{003024_3Jun25} outlines the complete decoding procedure, which combines joint BP decoding with the proposed post-processing step. The algorithm begins by running standard joint BP iterations. If a valid noise estimate satisfying the syndrome condition is found, the algorithm terminates and outputs the estimate. Otherwise, it continues iterating.
If the decoder stagnates during the estimation of X-noise or Z-noise, the algorithm invokes a post-processing step. In this step, the decoder estimates a set of cycles $\hat{\Cc}$ responsible for trapping, and attempts to solve a linear system restricted to the support $K(\hat{\Cc})$. If a solution is found, it is assigned to the corresponding part of the noise estimate. The process continues until a valid estimate is obtained or a prescribed iteration limit $T$ is reached.

\begin{algorithm}
\caption{Joint Belief Propagation + Post-Processing Decoding Algorithm}\label{003024_3Jun25}
\begin{algorithmic}[1]
\Require Maximum iterations $T\ge1$, history length $d\ge1$, cycle cap $u\ge1$
\State $\ell \gets 0$
\State Initialize with the joint BP algorithm, including hard decisions at iteration $0$
\While{$\ell<T$}
  \State $\ell \gets \ell + 1$
  \State Perform the $\ell$-th joint BP iteration%
  \State $(\xiUH,\zetaUH)\gets(\xiUH^{(\ell)},\zetaUH^{(\ell)})$
  \If{$H_\Delta \xiUH^{(\ell)} = \sigmaU$ and $H_\Gamma \zetaUH^{(\ell)} = \tauU$}
    \State Output the estimated noise $(\xiUH:=\xiUH^{(\ell)}, \zetaUH:=\zetaUH^{(\ell)})$ and terminate
  \EndIf
  \If{$H_\Delta \xiUH^{(\ell)} \neq \sigmaU$ and joint BP algorithm stagnates in X-noise estimation}
    \State Determine $\hat{\Cc}$ using the cycle-set estimation algorithm for $H_\Delta$
    \If{$\hat{\Cc} \neq \texttt{null}$}
      \State $K := K(\hat{\Cc})$, $\KO := [N]\setminus K$
      \If{$\sigmaU = (H_\Delta)_K \xiU_K + (H_\Delta)_{\KO} \xiUH_{\KO}^{(\ell)}$ has a solution $\xiU_K$}
        \State Assign one such solution to $\xiUH_K$ and $\xiUH_{\KO}:=\xiUH_{\KO}^{(\ell)}$
      \EndIf
    \EndIf
  \EndIf
  \If{$H_\Gamma \zetaUH^{(\ell)} \neq \tauU$ and joint BP algorithm stagnates in Z-noise estimation}
    \State Determine $\hat{\Cc}$ using the cycle-set estimation algorithm for $H_\Gamma$
    \If{$\hat{\Cc} \neq \texttt{null}$}
      \State $K := K(\hat{\Cc})$, $\KO := [N]\setminus K$
      \If{$\tauU = (H_\Gamma)_K \zetaU_K + (H_\Gamma)_{\KO} \zetaUH_{\KO}^{(\ell)}$ has a solution $\zetaU_K$}
        \State Assign one such solution to $\zetaUH_K$ and $\zetaUH_{\KO}:=\zetaUH_{\KO}^{(\ell)}$
      \EndIf
    \EndIf
  \EndIf
  \If{$H_\Delta \xiUH = \sigmaU$ and $H_\Gamma \zetaUH = \tauU$}
    \State Output the estimated noise $(\xiUH, \zetaUH)$ and terminate
  \EndIf
\EndWhile
\State \Return decoding failure
\end{algorithmic}
\end{algorithm}

The proposed post-processing targets a cycle set $\Cc$ composed of cycles of length $2L$. The next shortest cycles are those of length $2(L+2)$. These longer cycles may not be full-rank and could potentially contain logical errors. As will be discussed in the next section, when decoding the proposed codes with the proposed decoder, it is possible for the decoder to stagnate after becoming trapped in cycles of length $2(L+2)$. In such cases, one might consider adopting a decoding strategy that selects the most likely noise consistent with the syndrome. Such a strategy may, by chance, correctly match the true noise and succeed. However, it also risks resulting in an undetected decoding failure. Declining to refine these longer cycles does not guarantee that all failures are detectable.

\section{Results}\label{043151_25Apr25}
This section presents the results of the proposed code construction and decoding method.  
We begin by describing the construction of the code matrices.  
We then evaluate the decoding performance of the proposed degeneracy-aware decoder through numerical simulations over depolarizing channels.  
Finally, we analyze the minimum distance of the proposed code by enumerating short cycles in the Tanner graph.  

\subsection{Code Construction}
We constructed APMs $\fU$ and $\gU$, along with the corresponding matrices $\HH_X$ and $\HH_Z$, by randomized searches satisfying the constraints in Algorithm~\ref{alg:fg_construction}, for several choices of $P$, with fixed parameters $J = 2$ and $L = 6$. Although it would be preferable to select $P$ as a power of two for implementation purposes, we found that doing so made it extremely difficult to simultaneously satisfy both conditions (a) and (b). In our searches, valid sequences $\fU$ and $\gU$ were found more readily when $P$ contained multiple small prime factors. In particular, we were able to identify such sequences more easily when $P$ was a multiple of $384 = 2^7 \cdot 3$, which represents a relatively simple composite case. 
The resulting pairs $(\fU, \gU)$ of APMs are listed in Table~\ref{tab:015542_17Apr25}. 

For the finite field construction, we set $e = 8$ and $q = 2^e$, and defined the field $\mathbb{F}_q$ using the primitive polynomial $x^8 + x^4 + x^3 + x^2 + 1$ over $\mathbb{F}_2$.  
The matrices $H_\Gamma$ and $H_\Delta$ satisfy the construction constraints in Algorithms~\ref{alg:HGamma} and~\ref{alg:HDelta}.
We used the instance with $P = 6500$, which was the largest successful value of $P$ in Table~\ref{tab:015542_17Apr25}.
Specifically, the searches succeeded for $P=2^i\cdot3$ with $i=7,8,\ldots,11$, but did not find a solution for $i=12$ within the search performed.

Each nonzero entry in $H_\Gamma$ and $H_\Delta$ is then replaced by its corresponding $e \times e$ binary companion matrix (or its transpose), following the method described in Section \ref{153729_6Jun25}, to obtain the binary matrices $H_X$ and $H_Z$.

\subsection{Numerical Simulation}\label{134114_16Jun25}
We conducted numerical experiments over depolarizing channels with physical error rate  $p_D$. 
The proposed decoding method combines joint sum-product BP with the local post-processing in Algorithm~\ref{003024_3Jun25}. We use code-capacity noise, with independent Pauli errors and error-free syndrome measurements.

\begin{figure*}
  \centering
  \includegraphics[width=0.95\linewidth]{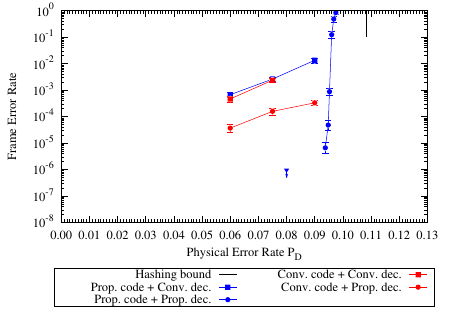}
\caption{Comparison of decoding performance between the proposed $[[312000, 104000, \le 14]]$ code and the conventional $[[312000, 104000, \le 10]]$ code~\cite{komoto2024quantumerrorcorrectionnear}, both having the same code length and rate $R = 1/3$. The conventional decoder is joint sum-product BP; the proposed decoder adds the local refinement. At $p_D=9.48\%$, the proposed code and refinement have 18 recorded failures in 374,932 trials, giving $4.80\times10^{-5}$. Error bars are nominal 95\% Wilson score intervals computed from the recorded counts. The zero-failure observation at $p_D=0.08$ is shown separately by a downward upper-limit marker at the upper endpoint of its interval. The vertical line is the rate-$1/3$ hashing-bound benchmark.}
 \label{fig:comparison}
\end{figure*}

Figure~\ref{fig:comparison} shows the decoding performance of the proposed code with $P = 6500$ using the degeneracy-aware decoder.  
The number of physical qubits is $n = eLP = 312000$.  
At a coding rate of $R=1/3$ and physical error rate $p_D=9.48\%$, 18 failures were recorded in 374,932 trials, giving an empirical frame error rate (FER) of $4.80\times10^{-5}$ and a nominal 95\% Wilson interval of $[3.04\times10^{-5},7.59\times10^{-5}]$.  
\rev{For this coding rate, the corresponding hashing-bound benchmark is $p_D^* \simeq 0.10835$ ($10.835\%$). Hence, the present result should be interpreted as finite-length evidence that the operating point moves toward this asymptotic benchmark, while a nontrivial gap still remains.}
Moreover, for fixed field size, degrees, and iteration limit, the BP complexity scales linearly with the number of physical qubits $n$.  
\rev{To clarify the low-error-rate scope of the numerical evidence, the simulations include a measured FER of $6.61\times10^{-6}$ at $p_D=0.09375$ and a zero-failure observation at $p_D=0.08$. Going below this level would require substantially more Monte Carlo trials, since the number of trials needed to obtain a statistically meaningful number of failures grows roughly inversely with the target FER. We therefore did not include plotted points below this range, and we do not claim that the present simulation data establish performance at $10^{-7}$ or below.}
\rev{At the lowest-noise point shown for the proposed code with the proposed decoder, $p_D=0.08$, no frame error was observed in 3,967,705 trials. The upper endpoint of its nominal 95\% Wilson score interval is $9.68\times10^{-7}$. Since a zero empirical FER cannot be displayed directly on the logarithmic axis, the downward marker in Fig.~\ref{fig:comparison} is placed at this upper endpoint.}
The recorded failure counts may include stabilizer-equivalent outcomes rejected by the sufficient stabilizer-membership test, making the reported FER conservative with respect to this test. The nominal Wilson intervals are calculated from these counts.

\rev{We emphasize that this operating point should not be interpreted as a complete FTQC-level performance demonstration. For a full fault-tolerant quantum-computing architecture, one ultimately needs substantially smaller logical failure probabilities, often at the level of $10^{-7}$ or below depending on the target computation, together with additional system-level requirements such as circuit-level noise modeling, repeated syndrome extraction, logical-gate protocols, decoding latency, and space-time overhead. The present work addresses a different but complementary QEC objective: it demonstrates that high-rate non-binary LDPC-based CSS codes can substantially reduce the error floor by exploiting degeneracy and local cycle-structured refinement. Thus, the reported FER $=4.80\times10^{-5}$ result and the low-noise observations are intended as finite-length high-rate QEC benchmarks and as evidence for the effectiveness of the proposed construction/decoder mechanism, not as a claim that all FTQC requirements have already been met.}
\rev{The same point is important when comparing the observed BP-based behavior with a bounded-distance-decoding interpretation. A BDD guarantee is worst-case: if the minimum distance is $d$, it guarantees correction only up to $\lfloor(d-1)/2\rfloor$ errors and, in general, cannot guarantee correction once the error weight exceeds this radius. BP is different. It is an iterative graph-based decoder whose success is not determined solely by the Hamming weight of the error or by whether the error overlaps with more than half of the support of a minimum-distance logical operator. Richardson's error-floor analysis of LDPC codes makes this distinction explicit: for a one-step ML decoder the relevant trapping sets are precisely nonzero codewords, whereas for iterative decoders the dominant error events are decoder-dependent trapping sets or near-codewords governed by the local graph and by messages from the surrounding variables and checks~\cite{richardson2003errorfloors}. Consequently, if the surrounding part of the factor graph supplies sufficiently reliable messages, BP may correct an error pattern even when its restriction to a particular low-weight logical support exceeds half of that support. Conversely, a smaller pattern may still cause failure if it forms a harmful trapping structure. In this sense, a guide proportional to $p_D^7$ should be read only as a low-$p_D$ leading-order approximation for rare minimum-distance-type events. It should not be extrapolated as a prediction over the whole plotted range: at larger $p_D$, higher-weight error patterns occur more often, and BP can still correct many of them by using reliability messages from the rest of the factor graph. This is why we interpret Fig.~\ref{fig:comparison} as finite-length, decoder-dependent simulation evidence rather than as a BDD or minimum-distance-only prediction.}
\rev{This interpretation also explains why BP+OSD is not included as an experimental curve in Fig.~\ref{fig:comparison}. As discussed above, BP+OSD is a global full-block post-processing method, whereas the proposed refinement is a local repair of the observed length-$12$ cycle-trapping events. Including BP+OSD would therefore shift the discussion to a broader comparison among global post-decoders, rather than isolating the intended global-versus-local contrast. The gain in Fig.~\ref{fig:comparison} should be understood as the effect of this cycle-structured local refinement, whose restricted linear-solve cost is $O(1)$ in the block length for fixed $q,L,u$, with support-search preprocessing and BP history maintenance as described above.}

For comparison, we also plot the decoding performance of a code constructed using the method of~\cite{komoto2024quantumerrorcorrectionnear} with the same parameters $(P, q, J, L)$.  
The conventional decoding method corresponds to joint BP decoding without any post-processing.

The proposed method exhibits a slightly inferior performance in the waterfall region compared to the conventional one.
Trade-offs between waterfall performance and error-floor behavior have also been observed in classical LDPC code design; one example is the effect of column weight discussed in~\cite[slide~9]{lin2004structured}.

The observed error-floor events include errors trapped in cycles of length $2(L+2)$. Such cycles can support logical errors. In particular, satisfying the syndrome condition does not by itself exclude a residual in $C_\Delta\setminus C_\Gamma^\perp$ or $C_\Gamma\setminus C_\Delta^\perp$, and the individual-cycle properties do not guarantee logical correctness on a union of cycles.

\subsection{Computing Upper Bound of Minimum Distance}
For the proposed code $(C_X, C_Z)$, we evaluate an upper bound on the minimum distance in this section.  
Cycle-based descriptions are useful for binary column-weight-two LDPC codes~\cite{1715529}. Here, we obtain upper bounds by explicitly constructing nontrivial logical operators supported on selected short cycles of the nonbinary Tanner graphs.  

Our approach for computing an upper bound on the minimum distance $d_X$ is as follows.  
The same procedure is applied to compute an upper bound on $d_Z$.  
\begin{enumerate}
 \item First, we enumerate short Tanner cycles $\Ct$ in $H_\Gamma$. 
 \item Next, for each cycle $\Ct$, we enumerate nonzero codewords $\xiU \in N(\Ct;H_\Gamma)$ that are not contained in $C_\Delta^\perp$.
 \item Among all such codewords, we take the smallest bitwise Hamming weight as an upper bound on $d_X$.
\end{enumerate}

Both $C_X$ and $C_Z$ are designed to have girth 12.  
That is, cycles of length 12 are the shortest cycles in the Tanner graphs of the codes.  
For any cycle $\Ct$ of length 12, we have constructed the code such that the corresponding subcode $N(\Ct;H_\Gamma)$ is either contained in $C_\Delta^\perp$ or has dimension zero.  
Hence, such cycles do not contribute to the minimum distance $d_X$ of $C_X$ (see \eqref{def:d_X}).  
The same holds for $d_Z$ of $C_Z$.

We now consider the next shortest cycles, those of length 16.  
We first enumerated all length-16 cycles $\Ct$ in the Tanner graph of $C_X$.  
For each cycle, we computed the rank over $\mathbb{F}_q$ of the corresponding $8 \times 8$ submatrix.  
This rank is either 8 or 7.  
In the former case, the subcode $N(\Ct;H_\Gamma) \subset C_\Gamma$ contains only the zero vector;  
in the latter case, it contains $q - 1$ nonzero codewords.  
We take the minimum bitwise Hamming weight among these nonzero codewords as an upper bound on $d_X$.  
An upper bound on $d_Z$ of $C_Z$ is computed in the same manner.  

Using this approach, we evaluated an upper bound on the minimum distance of the proposed code.  
In total, we identified 1396 and 1402 length-16 cycles whose corresponding submatrices have rank 7 in $H_\Gamma$ and $H_\Delta$, respectively.  
By computing the bitwise weights of the corresponding codewords-excluding those that lie in $C_\Delta^\perp$ and $C_\Gamma^\perp$, respectively-we obtained the weight distributions $A_X(w)$ and $A_Z(w)$, as shown in Table~\ref{weight_dist}.  

For comparison, we also computed $A_X(w)$ and $A_Z(w)$ for the conventional codes used in Section~\ref{134114_16Jun25}.  
In the conventional code, there exist length-12 cycles $\Ct$ in the UTCBC $\uU(2)$ in $H_\Gamma$ such that $N(\Ct;H_\Gamma)$ contains nonzero codewords that are not included in $C_\Delta^\perp$.  
These codewords were enumerated and included in the weight distribution $A_X(w)$.  
The distribution $A_Z(w)$ was computed in the same manner.  

In Table~\ref{weight_dist}, we list the weight distributions of the proposed and conventional codes.  
The upper bounds on the minimum bitwise distance are 14 and 10 for the proposed and conventional codes, respectively.  
The conventional code contains many low-weight codewords, which are likely responsible for the observed error floor.
In contrast, the proposed code includes fewer low-weight codewords arising from short cycles, which may contribute to its lower observed error floor.

We note that the minimum bitwise weight observed among the enumerated length-16 cycles is 14.
While this value provides an upper bound on the minimum distance, it may not be exact:
it is possible that codewords arising from longer cycles or from combinations of cycles in the Tanner graph have lower bitwise Hamming weights.

\begin{table*}[t]
  \centering
\caption{Weight distribution of codewords formed from the short cycles. Here, $A_X(w)$ and $A_Z(w)$ denote the number of codewords of weight $w$ in $C_X \setminus C_Z^\perp$ and $C_Z \setminus C_X^\perp$, respectively.}
\label{weight_dist}
  \renewcommand{\arraystretch}{0.95}
\begin{tabular}{c|c|c||c|c}
&\multicolumn{2}{c}{\text{Prop. }}&\multicolumn{2}{c}{\text{Conv. }}  \\\hline
    \hline
    $w$ & $A_X(w)$ & $A_Z(w)$& $A_X(w)$ & $A_Z(w)$\\
    \hline
10 & 0 & 0 & 4 & 2 \\
11 & 0 & 0 & 4 & 8 \\
12 & 0 & 0 & 33 & 35 \\
13 & 0 & 0 & 91 & 109 \\
14 & 2 & 0 & 262 & 257 \\
15 & 6 & 0 & 647 & 593 \\
16 & 9 & 4 & 1354 & 1255 \\
17 & 12 & 9 & 2520 & 2674 \\
18 & 46 & 33 & 4747 & 4751 \\
19 & 114 & 111 & 7721 & 7530 \\
20 & 274 & 266 & 11249 & 11363 \\
21 & 611 & 643 & 15276 & 15059 \\
22 & 1276 & 1277 & 18985 & 18985 \\
23 & 2360 & 2509 & 21733 & 21517 \\
24 & 4154 & 4394 & 22916 & 22499 \\
25 & 7245 & 7350 & 21693 & 21497 \\
26 & 11102 & 10952 & 19543 & 19119 \\
27 & 15624 & 15725 & 15993 & 15532 \\
28 & 21190 & 21109 & 11757 & 12019 \\
29 & 26526 & 26060 & 8259 & 8061 \\
30 & 30931 & 31305 & 5156 & 5246 \\
31 & 34512 & 34721 & 2941 & 2989 \\
32 & 35851 & 36401 & 1642 & 1614 \\
33 & 34931 & 35065 & 746 & 793 \\
34 & 32092 & 32146 & 344 & 346 \\
35 & 27441 & 27484 & 152 & 132 \\
36 & 22069 & 22073 & 49 & 45 \\
37 & 16730 & 16718 & 14 & 15 \\
38 & 11859 & 11865 & 7 & 6 \\
39 & 7856 & 8108 & 1 & 3 \\
40 & 5071 & 4962 & 0 & 1 \\
41 & 2919 & 2926 & 1 & 0 \\
42 & 1666 & 1667 & 0 & 0 \\
43 & 821 & 898 & 0 & 0 \\
44 & 393 & 412 & 0 & 0 \\
45 & 181 & 199 & 0 & 0 \\
46 & 64 & 73 & 0 & 0 \\
47 & 26 & 28 & 0 & 0 \\
48 & 9 & 9 & 0 & 0 \\
49 & 5 & 6 & 0 & 0 \\
50 & 2 & 2 & 0 & 0 \\
\hline
\text{total} & $255\times$1396 & $255\times$1402 & $255\times$768 & $255\times$761 \\
\end{tabular}
\end{table*}
\section{Conclusion and Future Work}\label{011317_4Jun25}

In this paper, we proposed a construction and decoding method for quantum error correction using non-binary LDPC codes, with a focus on explicitly exploiting degeneracy.  
We introduced key design criteria to control cycle structure in the parity-check matrices and developed a degeneracy-aware post-processing decoder to address decoding stagnation.  

Our construction utilizes APMs to ensure algebraic structure while maintaining randomness through selective sampling.  
This decoder dynamically estimates the subset of cycles responsible for the stagnation and corrects errors via constrained linear solving.  

Simulation results demonstrated that our code achieves an empirical FER of $4.80\times10^{-5}$ at a physical error rate of $9.48\%$ for rate-$1/3$ codes.
Furthermore, cycle enumeration gives minimum-distance upper bounds of 14 and 10 for the proposed and conventional codes, respectively.
For fixed field size, degrees, and iteration limit, the BP decoding complexity scales linearly with the number of physical qubits, while the local solve acts on at most $uL$ symbols.

Future work includes extending the proposed framework to support lower-rate quantum LDPC codes and exploring the feasibility of applying the proposed degeneracy-aware decoding strategy to binary LDPC codes.  
These directions may further broaden the applicability of our approach and enable practical deployment in systems where qubit and complexity constraints are stringent.  

 \section*{Acknowledgment}
 This study was carried out using the TSUBAME4.0 supercomputer at Institute of Science Tokyo.
\appendices
\section{Companion Matrix Representation over Finite Fields}\label{234926_5Apr25}
This section introduces the companion matrix representation of finite field elements over $\mathbb{F}_{2^e}$. We begin by reviewing the vector representation of field elements via the mapping $\vU$, and then define the companion matrix $A(\alpha)$ associated with a primitive element $\alpha$. The main result is that multiplication in $\mathbb{F}_{2^e}$ can be translated into matrix-vector multiplication over $\mathbb{F}_2$, providing a linear-algebraic framework that is crucial for decoding operations in quantum LDPC codes. Several theorems are presented to formalize the algebraic properties of these representations, including linearity, multiplicativity, and the equivalence between field equations and their matrix representations.

Let $q = 2^e$. It is well-known that elements of the finite field $\mathbb{F}_q$ can be commonly represented either as binary vectors of length $e$ or as polynomials of degree less than $e$ with binary coefficients.
Let $\alpha \in \mathbb{F}_q$ be a primitive element. Any element $\gamma \in \mathbb{F}_q$ can be uniquely written in the form 
\begin{align}
 \gamma = \sum_{j=0}^{e-1} g_j \alpha^j.
\end{align} 
We define the mapping $\vU : \mathbb{F}_q \to \mathbb{F}_2^e$ by
\[
\vU(\gamma) := (g_0, g_1, \ldots, g_{e-1})^{\top} \in \mathbb{F}_2^e.
\]

\begin{example}\label{ex:023559_22Mar25}
For $e=8$, using the primitive polynomial $x^8 + x^4 + x^3 + x^2 + 1$ over $\mathbb{F}_2$, we have:
\begin{align}
 \begin{array}{ll}
 \vU(0) &= (00000000)^\top, \\
 \vU(1) &= (10000000)^\top, \\
 \vU(\alpha) &= (01000000)^\top, \\
 \vU(\alpha^{e-1}) &= (00000001)^\top, \\
 \vU(\alpha^{e}) &= (10111000)^\top, \\
 \end{array}
\end{align}
\end{example}

\begin{teiri}
The map $\vU$ defines an isomorphism of vector spaces over $\mathbb{F}_2$, and also a group isomorphism under addition:
\[
\vU(\gamma_1 + \gamma_2) = \vU(\gamma_1) + \vU(\gamma_2), \quad \forall \gamma_1, \gamma_2 \in \mathbb{F}_q.
\]
\end{teiri}

\begin{proof}
Let $\gamma_1, \gamma_2 \in \mathbb{F}_q$. Since $\mathbb{F}_q = \mathbb{F}_{2^e}$ is an $\mathbb{F}_2$-vector space with basis $\{1, \alpha, \ldots, \alpha^{e-1}\}$, each $\gamma_i$ has a unique representation
\[
\gamma_1 = \sum_{j=0}^{e-1} g_j^{(1)} \alpha^j, \quad
\gamma_2 = \sum_{j=0}^{e-1} g_j^{(2)} \alpha^j \quad (g_j^{(i)} \in \mathbb{F}_2).
\]
Thus,
\[
\gamma_1 + \gamma_2 = \sum_{j=0}^{e-1} ( g_j^{(1)} + g_j^{(2)} ) \alpha^j.
\]
Hence,
\[
\vU(\gamma_1 + \gamma_2) = ( g_0^{(1)} + g_0^{(2)}, \ldots, g_{e-1}^{(1)} + g_{e-1}^{(2)} )^{\top}
= \vU(\gamma_1) + \vU(\gamma_2).
\]
Therefore, $\vU$ is a homomorphism with respect to addition, and it is a linear isomorphism over $\mathbb{F}_2$, as it extracts coefficients from a linear combination and both $\mathbb{F}_q$ and $\mathbb{F}_2^e$ have the same dimension.
\end{proof}

They can also be naturally represented using $e \times e$ binary matrices.
\begin{df}
Let $a(x) = a_0 + a_1 x + \cdots + a_e x^e$ with $a_0 = a_e = 1$ be the primitive polynomial of $\alpha \in \mathbb{F}_q$. The {\itshape companion matrix} $A(\alpha)$ \cite{macwilliams77} is defined as:
\begin{align}
    A(\alpha) &\defeq
\bigl(\vU(\alpha^1),\vU(\alpha^2),\ldots,\vU(\alpha^{e})\bigr)\label{153337_11May25}
\\&=
    \begin{pmatrix}
        0 & 0 & 0 & 0 & a_0 \\
        1 & 0 & 0 & 0 & a_1 \\
        0 & 1 & 0 & 0 & a_2 \\
        \vdots & \vdots & \ddots & \vdots & \vdots \\
        0 & 0 & 0 & 1 & a_{e-1}
    \end{pmatrix}.\label{031734_22Mar25}
\end{align}
We define the mapping \(A: \Fb_q \to \Fb_2^{e \times e}\) by
\begin{align}
    A(0) &\defeq O, \\
    A(\alpha^l) &\defeq A(\alpha)^l, \quad \text{for } l = 0, \ldots, q-2.
\end{align}
\end{df}

The mapping \(A: \mathbb{F}_q \to \mathbb{F}_2^{e \times e}\) introduced above has important algebraic properties. Specifically, it preserves the algebraic structure of the finite field $\mathbb{F}_q$, as demonstrated in the following theorem.
\begin{teiri}\label{163515_22Mar25}
As shown in \cite{macwilliams77}, the map $A$ is injective, and its image $A(\mathbb{F}_q)$ forms a field isomorphic to $\mathbb{F}_q$. The map satisfies the following properties for all $\gamma_1, \gamma_2, \delta \in \mathbb{F}_q$: multiplicativity, commutativity, and additivity:
\begin{align}
&     A(\gamma_1 \gamma_2) = A(\gamma_1)A(\gamma_2) = A(\gamma_2)A(\gamma_1), \label{044508_22Mar25}\\
&     A(\gamma)\vU(\delta) = \vU(\gamma\delta), \label{044514_22Mar25}\\
&     A(\gamma_1 + \gamma_2) = A(\gamma_1) + A(\gamma_2). \label{044517_22Mar25}
\end{align}
\end{teiri}

\begin{proof}
The case for the zero element is trivial, so we assume that $\gamma_1, \gamma_2, \delta \in \mathbb{F}_q^\times$.
Let $\eU^{(i)}$ denote the standard basis vector of $\mathbb{F}_2^e$ whose $i$th component is 1 and the others are 0. By the definition of $\vU$, we have $\vU(\alpha^i) = \eU^{(i)}$. Furthermore, by the definition of the companion matrix $A(\alpha)$, its action on $\eU^{(i)}$ corresponds to the $i$th column of $A(\alpha)$. Hence, we have
\begin{align}
  A(\alpha)\vU(\alpha^i) = \vU(\alpha^{i+1}) \quad \text{for all } i \in [e]. \label{025939_22Mar25}
\end{align}
By the $\Fb_2$-linearity of $\vU$, equation~\eqref{025939_22Mar25} implies $A(\alpha)\vU(\delta)=\vU(\alpha\delta)$ for every $\delta\in\Fb_q$. Induction therefore gives $A(\alpha)^k\vU(\delta)=\vU(\alpha^k\delta)$ for $k\ge0$. In particular, $A(\alpha)^{q-1}=I$, since it fixes every binary vector. Thus powers may be interpreted modulo $q-1$.
Let $\gamma_1 = \alpha^i$ and $\gamma_2 = \alpha^j$. Then we have
\[
\begin{aligned}
A(\gamma_1) A(\gamma_2) &= A(\alpha)^i A(\alpha)^j = A(\alpha)^{i+j} \\
&= A(\alpha^{i+j}) = A(\gamma_1 \gamma_2),
\end{aligned}
\]
which establishes both multiplicativity and commutativity.
To prove \eqref{044514_22Mar25}, suppose $\gamma = \alpha^j$ and $\delta = \alpha^i$. Then,
\[
\begin{aligned}
A(\gamma)\vU(\delta) &= A(\alpha^j)\vU(\alpha^i) = A(\alpha)^j \vU(\alpha^i) \\
&= \vU(\alpha^{i+j}) = \vU(\gamma \delta),
\end{aligned}
\]
where we used the action on arbitrary $\delta$ established above.
To prove \eqref{044517_22Mar25}, note that for any $\delta \in \mathbb{F}_q$, we compute:
\[
\begin{aligned}
&(A(\gamma_1) + A(\gamma_2))\vU(\delta)\\
&= A(\gamma_1)\vU(\delta) + A(\gamma_2)\vU(\delta)\\
&= \vU(\gamma_1 \delta) + \vU(\gamma_2 \delta)\\
&= \vU((\gamma_1 + \gamma_2)\delta)\\
&= A(\gamma_1 + \gamma_2)\vU(\delta).
\end{aligned}
\]
Since $\vU$ is surjective onto $\Fb_2^e$, the equality
\[
A(\gamma_1 + \gamma_2) = A(\gamma_1) + A(\gamma_2)
\]
follows. Finally, $A(\gamma)\vU(1)=\vU(\gamma)$ proves that $A$ is injective. Its image contains $I=A(1)$ and the inverse $A(\gamma^{-1})$ of each nonzero $A(\gamma)$, so it is a field isomorphic to $\Fb_q$.
\end{proof}

The matrix representation $A(\gamma)$ not only captures the algebraic operations in $\mathbb{F}_q$, but also provides a useful tool for performing field multiplication and addition in terms of binary vector and matrix operations. In particular, the mapping $A$ allows us to translate scalar equations over $\mathbb{F}_q$ into linear-algebraic expressions over $\mathbb{F}_2$ via the map $\vU$.

The following theorem illustrates a basic but important connection between the vector representation $\vU(\alpha^i)$ and the matrix representation $A(\alpha^i)$.
\begin{teiri}
The vector $\vU(\alpha^i)$ agrees with the leftmost column of $A(\alpha^i)$.
\end{teiri}
\begin{proof}
From the multiplication property,
\[
A(\gamma_1)\vU(\gamma_2) = \vU(\gamma_1 \gamma_2).
\]
Let $\gamma_1 = \alpha^i$ and $\gamma_2 = 1 = \alpha^0$. Then
\[
A(\alpha^i)\vU(1) = \vU(\alpha^i).
\]
Since $\vU(1) = (1, 0, \ldots, 0)^{\top}$, this implies that $\vU(\alpha^i)$ agrees with the first column of $A(\alpha^i)$, which completes the proof.
\end{proof}

The next theorem provides an equivalence between scalar linear combinations over $\mathbb{F}_q$ and their binary vector representations using matrix multiplication. This property will be useful when solving equations or decoding in vectorized form.
\begin{teiri}\label{013653_26Mar25}
The following equivalence holds for any \(\delta_j, \xi_j, \sigma \in \Fb_q\):
\begin{align}
\sum_j \delta_j \xi_j = \sigma 
\quad \text{ if and only if }\quad 
\sum_j A(\delta_j)\, \vU(\xi_j) = \vU(\sigma).
\end{align}
\end{teiri}

\begin{proof}
Using the additivity of \(\vU\) and the multiplicative property \(A(\gamma)\vU(\delta) = \vU(\gamma \delta)\) from \eqref{044514_22Mar25}, we compute:
\begin{align}
\sum_j A(\delta_j)\, \vU(\xi_j) 
&= \sum_j \vU(\delta_j \xi_j) \notag\\
&= \vU\Bigl(\sum_j \delta_j \xi_j\Bigr). \notag
\end{align}
Therefore, the identity \(\sum_j A(\delta_j)\, \vU(\xi_j) = \vU(\sigma)\) holds  
if and only if \(\vU\left(\sum_j \delta_j \xi_j\right) = \vU(\sigma)\).  
Since \(\vU\) is injective, this is equivalent to \(\sum_j \delta_j \xi_j = \sigma\).
\end{proof}

\section{Finite Field Generated by \(A(\alpha)^\T\)}\label{234956_5Apr25}

In the previous section, we studied the algebraic structure of the field generated by the companion matrix \(A(\alpha)\), which faithfully represents multiplication in $\Fb_q$. In this section, we turn our attention to the transpose matrix \(A(\alpha)^\T\), which also generates a field isomorphic to \(\Fb_q\).

Our goal is to examine the algebraic properties of this transposed field representation and to define a new binary mapping \(\wU(\cdot)\) that works analogously with \(A(\cdot)^\T\), in a manner similar to \(\vU(\cdot)\) and \(A(\cdot)\) discussed previously. This will be important, for example, in dual code constructions or applications where right multiplication by vectors is preferable.
Specifically, we aim to identify the corresponding map \(\wU(\cdot)\) that yields similar algebraic properties when used in conjunction with \(A(\cdot)^\T\), analogous to those established in Theorem~\ref{163515_22Mar25}.

\begin{df}
Define the mapping \(A^\T: \Fb_q \to \Fb_2^{e \times e}\) as follows:
\begin{align}
    A^\T(0) &\defeq O, \\
    A^\T(\alpha^i) &\defeq (A(\alpha)^\T)^i = (A(\alpha)^i)^\T \quad \text{for } i = 0, \ldots, q-2.
\end{align}
\end{df}

The matrix mapping \( A^\T(\cdot) \) inherits many of the algebraic properties of \( A(\cdot) \), because transposition preserves addition and reverses products, while all matrices in $A(\Fb_q)$ commute. In particular, since \( A(\alpha) \) generates a field isomorphic to \( \Fb_q \), its transpose \( A(\alpha)^\T \) does as well.
The following theorem summarizes the basic algebraic structure of the set \( A^\T(\Fb_q) \), and shows that it forms a field isomorphic to \( \Fb_q \). This will be particularly useful when we consider right-multiplication by row vectors in the next section.
\begin{teiri}
The map \(A^\T\) is injective, and its image \(A^\T(\Fb_q) \subset \Fb_2^{e \times e}\) forms a field that is isomorphic to \(\Fb_q\) under the correspondence defined by \(A^\T\). The following properties hold for all \(\gamma_1, \gamma_2 \in \Fb_q\):
\begin{align}
& A^\T(\gamma_1 + \gamma_2) = A^\T(\gamma_1) + A^\T(\gamma_2), \label{221847_25Mar25}\\
& A^\T(\gamma_1 \gamma_2) = A^\T(\gamma_1) A^\T(\gamma_2). \label{012141_26Mar25}
\end{align}
\end{teiri}

\begin{proof}
Since \(A^\T(\gamma)\) is the transpose of \(A(\gamma)\), the result follows immediately from Theorem~\ref{163515_22Mar25}.
\end{proof}

We define \(\wU(0) \defeq (0, \ldots, 0) \in \Fb_2^e\).  
For \(i = 1, \ldots, q-1\), define \(\wU(\alpha^i) \in \Fb_2^e\) as the first column of \(A^\T(\alpha^i)\).
Note that this definition is analogous to the map \(\vU: \Fb_q \to \Fb_2^e\) introduced in the previous section, where \(\vU(\alpha^i)\) corresponds to the \(i\)th standard basis vector for \(i\in[e]\) and satisfies \(A(\gamma) \vU(\delta) = \vU(\gamma \delta)\).

In contrast, as we will prove in Theorem~\ref{161818_11May25}, the map \(\wU(\cdot)\) satisfies the relation
\[
A^\T(\gamma) \wU(\delta) = \wU(\gamma \delta),
\]
which mirrors the behavior of \(\vU(\cdot)\), but with respect to the transpose operator. In particular, while \(\vU(\gamma)\) is obtained as a column of \(A(\gamma)\), \(\wU(\gamma)\) is constructed from the row structure of \(A(\gamma)\) via transposition.

This symmetry between \(\vU\) and \(\wU\) allows us to characterize linear relations and matrix actions in terms of either map, depending on whether we prefer right or left multiplication.

The map \(\wU: \Fb_q \to \Fb_2^e\), defined as the first column of \(A^\T(\gamma)\), also preserves addition. That is, for all \(\gamma_1, \gamma_2 \in \Fb_q\), we have
\[
\wU(\gamma_1 + \gamma_2) = \wU(\gamma_1) + \wU(\gamma_2),
\]
which follows directly from the additivity of \(A^\T(\cdot)\) given in \eqref{221847_25Mar25}. If $\gamma\ne0$, then $A(\gamma)^\T$ is invertible, and hence $\wU(\gamma)=A(\gamma)^\T\eU^{(0)}\ne\zeroU$. Thus the additive map has trivial kernel. Since its domain and codomain have the same finite cardinality, $\wU$ is an $\Fb_2$-linear isomorphism.

\begin{teiri}\label{161818_11May25}
The following multiplicative property holds for all \(\gamma_1, \gamma_2 \in \Fb_q\):
\begin{align}
A^\T(\gamma_1)\, \wU(\gamma_2) = \wU(\gamma_1 \gamma_2).
\end{align}
\end{teiri}

\begin{proof}
If either \(\gamma_1\) or \(\gamma_2\) is zero, both sides are clearly zero, so the identity holds trivially. Suppose now that \(\gamma_1 = \alpha^j\) and \(\gamma_2 = \alpha^i\) for some \(i, j \in \mathbb{Z}_{\ge 0}\). 

By the definition of \(\wU(\cdot)\), we have
\begin{align}
\wU(\alpha^i) = A^\T(\alpha^i)\, \eU^{(0)} = A^\T(\alpha^i)\, \wU(\alpha^0) = A^\T(\alpha^i)\, \wU(1). \label{012303_26Mar25}
\end{align}
Using the multiplicativity of \(A^\T(\cdot)\) established in \eqref{012141_26Mar25}, we have:
\begin{align}
A^\T(\alpha^j)\, \wU(\alpha^i)
&\stkeq{\eqref{012303_26Mar25}} A^\T(\alpha^j)\, A^\T(\alpha^i)\, \wU(1) \notag\\
&= A^\T(\alpha^{i+j})\, \wU(1) \notag\\
&\stkeq{\eqref{012303_26Mar25}} \wU(\alpha^{i+j}),
\end{align}
where all exponents are interpreted modulo \(q - 1\). Hence, the identity holds for all \(\gamma_1, \gamma_2 \in \Fb_q^\times\), and the general case follows.
\end{proof}

To facilitate the expression of scalar relations over \(\Fb_q\) in terms of binary matrix operations involving \(A^\T(\cdot)\), we introduce the mapping \(\wU : \Fb_q \to \Fb_2^e\), which plays a role analogous to \(\vU(\cdot)\) used with \(A(\cdot)\). In particular, \(\wU(\cdot)\) is designed to interact naturally with the transposed field representation so that right-multiplication corresponds to field multiplication.
The following theorem establishes the fundamental equivalence between scalar linear combinations over \(\Fb_q\) and their binary vector representation via \(A^\T(\cdot)\) and \(\wU(\cdot)\).
\begin{teiri}\label{wU_linear_relation_theorem}
The following equivalence holds for any \(\gamma_j, \zeta_j, \tau\in \Fb_q\):
\begin{align}
\sum_j \gamma_j \zeta_j = \tau
\quad \Longleftrightarrow \quad 
\sum_j A^\T(\gamma_j)\, \wU(\zeta_j) = \wU(\tau).
\end{align}
\end{teiri}
\begin{proof}
This can be shown by an argument analogous to that of Theorem~\ref{013653_26Mar25}.  
Using the multiplicativity of \(A^\T(\cdot)\), we have \(A^\T(\gamma_j)\, \wU(\zeta_j) = \wU(\gamma_j \zeta_j)\). Therefore,
\begin{align}
\sum_j A^\T(\gamma_j)\, \wU(\zeta_j) 
= \sum_j \wU(\gamma_j \zeta_j) 
= \wU\bigl( \sum_j \gamma_j \zeta_j \bigr).
\end{align}
Thus, the left-hand side equals \(\wU(\tau)\) if and only if \(\sum_j \gamma_j \zeta_j = \tau\), since \(\wU\) is injective.
\end{proof}

\renewcommand{\arraystretch}{1.5} %
\begin{table*}[h]
    \caption{Companion matrices and binary representation for primitive element $\alpha\in \Fb_8$ with primitive polynomial $a(x)=1+x+x^3$.}
    \begin{align}
        \begin{array}{c|ccccccc}
            i & 0 & 1& 2&3&4&5&6
            \\\hline\hline
            \alpha^i
            &\alpha^0
            &\alpha^1
            &\alpha^2
            &\alpha^3
            &\alpha^4
            &\alpha^5
            &\alpha^6
            \\\hline    
            \vU(\alpha^i)&
            (100)^\T&
            (010)^\T&
            (001)^\T&
            (110)^\T&
            (011)^\T&
            (111)^\T&
            (101)^\T
            \\\hline
            A^i & \begin{pmatrix}
            100 \\
            010 \\
            001
            \end{pmatrix} & \begin{pmatrix}
            001 \\
            101 \\
            010
            \end{pmatrix} & \begin{pmatrix}
            010\\
            011\\
            101
            \end{pmatrix} & \begin{pmatrix}
            101\\
            111\\
            011
            \end{pmatrix} & \begin{pmatrix}
            011\\
            110\\
            111
            \end{pmatrix} & \begin{pmatrix}
            111\\
            100\\
            110
            \end{pmatrix} & \begin{pmatrix}
            110\\
            001\\
            100
            \end{pmatrix}
            \\\hline
            \wU(\alpha^i)&
            (100)^\T&
            (001)^\T&
            (010)^\T&
            (101)^\T&
            (011)^\T&
            (111)^\T&
            (110)^\T
            \\\hline
            (A^\T)^i & \begin{pmatrix}
            100 \\
            010 \\
            001
            \end{pmatrix} & \begin{pmatrix}
            010\\
            001\\
            110
            \end{pmatrix} & \begin{pmatrix}
            001\\
            110\\
            011
            \end{pmatrix} & \begin{pmatrix}
            110\\
            011\\
            111
            \end{pmatrix} & \begin{pmatrix}
            011\\
            111\\
            101
            \end{pmatrix} & \begin{pmatrix}
            111\\
            101\\
            100
            \end{pmatrix} & \begin{pmatrix}
            101\\
            100\\
            010
            \end{pmatrix}
        \end{array}
    \end{align}
\renewcommand{\arraystretch}{1.0} %
\end{table*}

\section{Proof for \texorpdfstring{$\eqref{224009_19May25}$}{(24)}}\label{164912_4Jun25}
The orthogonality of $H_X$ and $H_Z$ follows from the linearity and multiplicativity of the companion matrix mapping $A$, together with the orthogonality of the underlying non-binary matrices $H_\Gamma$ and $H_\Delta$:
\begin{align}
    \left(H_X H_Z^{\top}\right)_{i, j} 
    &= \sum_{k} A\left(\gamma_{i, k}\right) A\left(\delta_{j, k}\right) \\
    &= \sum_{k} A\left(\gamma_{i, k} \delta_{j, k}\right) \\
    &= A\left(\sum_{k} \gamma_{i, k} \delta_{j, k}\right) \\
    &= A\left(\left(H_{\Gamma} H_{\Delta}^{\top}\right)_{i, j}\right) \\
    &= A(0) \\
    &= O.
\end{align}

\section{Proof for \texorpdfstring{$\eqref{015905_13Feb25}$ and $\eqref{095733_27May25}$}{(24)}}\label{224018_19May25}
Assume the nonzero-label and simple-cycle support conditions stated in Section~\ref{sec:231308_5Jun25}. In this appendix, we show that if the orthogonality condition \eqref{212247_11Feb25}:
\begin{align}
H_\Gamma (H_\Delta)^\top = O 
\end{align}
holds, then the congruence relations \eqref{015905_13Feb25} and \eqref{095733_27May25}:
\begin{align}
\left(
\setlength{\arraycolsep}{2pt}\begin{array}{c|c|c|c}
-\HH_Z^{(\Lrm)} & \HH_Z^{(\Rrm)} & \HH_Z^{(\Lrm)} & -\HH_Z^{(\Rrm)}
\end{array}
\right)
 \log \gammaU &= \zeroU \pmod{q - 1}, \\
\left(
\setlength{\arraycolsep}{2pt}\begin{array}{c|c|c|c}
-\HH_X^{(\Lrm)} & \HH_X^{(\Rrm)} & \HH_X^{(\Lrm)} & -\HH_X^{(\Rrm)}
\end{array}
\right)
 \log \deltaU &= \zeroU \pmod{q - 1} 
\end{align}
necessarily follow.

Let $\deltaU_{jr}^\top$ denote the $r$-th row of $H_\Delta^{(j)}$, where $H_\Delta^{(j)}$ is the $j$-th row block of $H_\Delta$. 
By the orthogonality condition, we have $H_\Gamma \deltaU_{jr} = \zeroU$.  
Note that $H_\Gamma$ shares the same support (i.e., positions of nonzero entries) as $\HH_X$, which is structured in a specific pattern. 
According to Theorem~\ref{053905_28Apr25}, the restriction of $H_\Gamma$ to the support of $\deltaU_{jr}$, after deleting zero rows, forms a square matrix of size $L$, and it constitutes a Tanner cycle $\Ct$ contained in the UTCBC $\uU(j)$ (defined in Section~\ref{sec:235443_19May25}) in $H_\Gamma$.  

As shown in the proof of Theorem~\ref{053905_28Apr25}, the Tanner cycle $\Ct$ starts from the $(0,0)$ block and traverses four elements in a clockwise direction, repeating this pattern for $L/2$ rounds.
Label the elements visited in the order of the cycle as follows:
\begin{align}
&\gamma^{(0)}_0 \to \gamma^{(0)}_1 \to \gamma^{(0)}_2 \to \gamma^{(0)}_3 \to \\
&\gamma^{(1)}_0 \to \gamma^{(1)}_1 \to \gamma^{(1)}_2 \to \gamma^{(1)}_3 \to \\
&\quad \vdots \\
&\gamma^{(L/2-1)}_0 \to \gamma^{(L/2-1)}_1 \to \gamma^{(L/2-1)}_2 \to \gamma^{(L/2-1)}_3 \to \gamma^{(0)}_0.\label{151512_2Jun25}
\end{align}
Note that for each $\ell \in [L/2]$, the elements $\gamma^{(\ell)}_1$ and $\gamma^{(\ell)}_2$ lie in the same column in the right half of $H_\Gamma$, while $\gamma^{(\ell)}_3$ and $\gamma^{(\ell+1)}_0$ lie in the same column in the left half of $H_\Gamma$.  
The set of all such columns coincides with the support (i.e., nonzero positions) of $\deltaU_{jr}$.

The determinant of this submatrix $\Ct$ is given by
\begin{align}
\prod_{\ell \in [L/2]} \gamma^{(\ell)}_0 \gamma^{(\ell)}_2 
+ \prod_{\ell \in [L/2]} \gamma^{(\ell)}_1 \gamma^{(\ell)}_3. \label{184118_30May25}
\end{align}
Since we work over a finite field of characteristic 2, i.e., $q = 2^e$, the sign does not affect the result.  
Moreover, since $\Ct$ is orthogonal to the vector extracted from the nonzero entries of $\deltaU_{jr}$, the rank of $\Ct$ must be deficient, that is, strictly less than $L$, and hence~\eqref{184118_30May25} must equal zero.  
Therefore, the following equation holds:  
\[
\prod_{\ell \in [L/2]} \frac{\gamma^{(\ell)}_1 \gamma^{(\ell)}_3}{\gamma^{(\ell)}_0 \gamma^{(\ell)}_2} = 1.
\]  
Applying the discrete logarithm $\log_\alpha(\cdot)$ (with the base $\alpha$ omitted), we obtain  
\begin{align}
&\sum_{\ell \in [L/2]} \big( -\log \gamma^{(\ell)}_0 + \log \gamma^{(\ell)}_1 + \log \gamma^{(\ell)}_3 - \log \gamma^{(\ell)}_2 \big)\notag\\&\qquad = 0 \pmod{q - 1}. \label{151702_2Jun25}
\end{align}

Observe that the columns in which the variables $\gamma^{(\ell)}_0$ for $\ell\in[L/2]$ appear in $H_\Gamma$ coincide with the positions of nonzero entries in the $r$-th row of the $j$-th row block of $\HH_Z^{(\Lrm)}$.  
The same holds for $\gamma^{(\ell)}_3$.
Likewise, the variables $\gamma^{(\ell)}_i$ for $i = 1, 2$ appear in the same columns as the nonzero entries in the $r$-th row of the $j$-th row block of $\HH_Z^{(\Rrm)}$.

Let $\hU_{jr}^\top$ denote the $r$-th row in the $j$-th row block of $\HH_Z$.  
We write the left and right halves of $\hU_{jr}^\top$ as $\hU_{jr,\Lrm}^\top$ and $\hU_{jr,\Rrm}^\top$, respectively.  
Recall the definition of $\gammaU = (\gamma_0, \ldots, \gamma_{2N-1})$ given in~\eqref{235116_14Jun25}.  
The first and second halves of the vector $\gammaU$ correspond to the $N$ symbols $\gamma_{ij}$ that appear in the upper and lower rows of $H_\Gamma$, respectively, arranged from left to right.  
Recall that for each $\ell \in [L/2]$, the elements $\gamma^{(\ell)}_1$ and $\gamma^{(\ell)}_2$ lie in the same column in the right half of $H_\Gamma$, while $\gamma^{(\ell)}_3$ and $\gamma^{(\ell+1)}_0$ lie in the same column in the left half of $H_\Gamma$.  
Based on this structure, equation~\eqref{151702_2Jun25} can be rewritten as follows.  
\begin{align}
\left(
\setlength{\arraycolsep}{2pt}\begin{array}{c|c|c|c}
-\hU_{jr,\Lrm}^\top & \hU_{jr,\Rrm}^\top & \hU_{jr,\Lrm}^\top & -\hU_{jr,\Rrm}^\top 
\end{array}
\right)
\log \gammaU = \zeroU \pmod{q - 1}.
\end{align}
Since this argument applies to all $j \in [J]$ and $r \in [P]$, we can represent it in matrix-vector form as:
\begin{align}
\left(
\setlength{\arraycolsep}{2pt}\begin{array}{c|c|c|c}
-\HH_Z^{(\Lrm)} & \HH_Z^{(\Rrm)} & \HH_Z^{(\Lrm)} & -\HH_Z^{(\Rrm)}
\end{array}
\right)
\log \gammaU = \zeroU \pmod{q - 1}.
\end{align}
This establishes~\eqref{015905_13Feb25}. Conversely, under these support conditions, the congruence makes each selected cycle matrix singular. A cycle with nonzero edge labels has rank at least $L-1$, so its kernel is one-dimensional. The check equations around the cycle propagate a nonzero coordinate to nonzero values at every other position. Choosing a nonzero scalar multiple of this kernel generator for each prescribed row therefore constructs a support-preserving $H_\Delta$ satisfying $H_\Gamma H_\Delta^\top=O$.
The identity~\eqref{095733_27May25} can be proved in the same manner.

\section{Proof of \texorpdfstring{$\eqref{160947_1Jun25}$ and $\eqref{160950_1Jun25}$}{(24)}}\label{161030_1Jun25}

In this section, we prove that Equation~\eqref{160947_1Jun25} is equivalent to the condition that every Tanner cycle $\Ct$ contained in the UTCBC $\uU(2)$ in $H_\Gamma$ satisfies that $N(\Ct;H_\Gamma)$ contains no nonzero codewords.

Let $\Ct$ be the Tanner cycle in $H_\Gamma$ selected by the support of the $r$-th row of the extended binary block $\HH_Z^{(2)}$. Following the discussion in Appendix~\ref{224018_19May25}, we label each variable in $\Ct$ as in Equation~\eqref{151512_2Jun25}.
The dimension of $N(\Ct;H_\Gamma)$ is zero if and only if $\Ct$ has full rank, i.e., rank $L$. The condition that $\Ct$ is full rank can be formulated analogously to the argument for non-full-rank cycles in Equation~\eqref{151702_2Jun25}, as follows:
\begin{align}
 &\sum_{\ell \in [L/2]} \big( -\log \gamma^{(\ell)}_0 + \log \gamma^{(\ell)}_1 + \log \gamma^{(\ell)}_3 - \log \gamma^{(\ell)}_2 \big)\notag\\&\qquad =  c_r \pmod{q - 1}, 
\end{align}
for some $c_r \neq 0$.

Since this condition must hold for all $r \in [P]$, we can rewrite the expression in matrix-vector form:  
\begin{align}
\left(
\setlength{\arraycolsep}{1pt}\begin{array}{c|c|c|c}
-\HH_Z^{(2,\Lrm)} & \HH_Z^{(2,\Rrm)} & \HH_Z^{(2,\Lrm)} & -\HH_Z^{(2,\Rrm)}
\end{array}
\right)
\log \gammaU = \cU \pmod{q - 1},
\end{align}
where $\cU$ is a vector with $c_i \neq 0$ for all $i \in [P]$.  
This establishes equation~\eqref{160947_1Jun25}.  
Equation~\eqref{160950_1Jun25} can be derived in the same manner.  

\bibliographystyle{IEEEtran}
\bibliography{IEEEabrv,references} 
\par\noindent\begin{minipage}{\columnwidth}
\input{biography}
\end{minipage}
\end{document}

%% file: notation.tex
\usepackage{bbm}
\usepackage{mathrsfs}

\makeatletter
\newcommand\RedeclareMathOperator{%
  \@ifstar{\def\rmo@s{m}\rmo@redeclare}{\def\rmo@s{o}\rmo@redeclare}%
}
\newcommand\rmo@redeclare[2]{%
  \begingroup \escapechar\m@ne\xdef\@gtempa{{\string#1}}\endgroup
  \expandafter\@ifundefined\@gtempa
     {\@latex@error{\noexpand#1undefined}\@ehc}%
     \relax
  \expandafter\rmo@declmathop\rmo@s{#1}{#2}}
\newcommand\rmo@declmathop[3]{%
  \DeclareRobustCommand{#2}{\qopname\newmcodes@#1{#3}}%
}
\@onlypreamble\RedeclareMathOperator
\makeatother

\definecolor{grey}{rgb}{0.7, 0.75, 0.71}

\def\dark-red#1{\textcolor[rgb]{0.7,0.0,0.0}{#1}}

\definecolor{amber}{rgb}{1.0, 0.75, 0.0}

\def\...{\dotsc}

\def\intT2{\int_{-T/2}^{T/2}}

\def\sumi1n{\sum_{i=1}^{n}}
\def\sumi1N{\sum_{i=1}^{N}}
\def\sumi0N--{\sum_{i=0}^{N-1}}

\def\ccc{\cdots}

\def\stkeq#1{\stackrel{{\scriptsize #1}}{=}}

\def\=def{\overset{\text{\small def}}{=}}
\newcommand{\defeq}{\overset{\text{\small def}}{=}}

\newcommand{\argmax}{\operatornamewithlimits{argmax}}

\DeclareMathOperator{\supp}{supp}

\DeclarePairedDelimiterX{\inp}[2]{\langle}{\rangle}{#1, #2}

\def\Fb{\mathbb{F}}

\def\Zb{\mathbb{Z}}

\def\Ac{\mathcal{A}}

\def\Cc{\mathcal{C}}

\def\Fc{\mathcal{F}}

\def\Qc{\mathcal{Q}}

\def\<{\langle}
\def\>{\rangle}

 \def\mat4#1#2#3#4{
\begin{pmatrix}
 #1&\ccc&#2\\
 \vdots&&\vdots\\
 #3&\ccc&#4
\end{pmatrix}}

\def\0sf{\mathsf{0}}
\def\1sf{\mathsf{1}}

\def\brm{\mathrm{b}}

\def\frm{\mathrm{f}}

\def\Lrm{\mathrm{L}}

\def\Rrm{\mathrm{R}}

\def\0BS{\boldsymbol{0}}
\def\1BS{\boldsymbol{1}}

\def\0B{\mathbf{0}}
\def\1B{\mathbf{1}}

\def\0H{\hat{0}}
\def\1H{\hat{1}}

\def\xH{\hat{x}}

\def\zH{\hat{z}}

\def\EH{\hat{E}}

\def\HH{\hat{H}}

\def\+TT{\texttt{+}}
\def\-{\texttt{-}}
\def\+KB{|+\> \<+|}
\def\-KB{|-\> \<-|}

\def\q0{|0\>}

\def\zetaH{\hat{\zeta}}

\def\xiH{\hat{\xi}}

\def\zeroU{\underline{0}}
\def\0U{\underline{0}}
\def\1U{\underline{1}}

\def\aU{\underline{a}}

\def\cU{\underline{c}}
\def\dU{\underline{d}}
\def\eU{\underline{e}}
\def\fU{\underline{f}}
\def\gU{\underline{g}}
\def\hU{\underline{h}}

\def\sU{\underline{s}}
\def\tU{\underline{t}}
\def\uU{\underline{u}}
\def\vU{\underline{v}}
\def\wU{\underline{w}}
\def\xU{\underline{x}}

\def\zU{\underline{z}}

\def\0UH{\underline{\0H}}
\def\1UH{\underline{\1H}}

\def\xUH{\underline{\xH}}

\def\zUH{\underline{\zH}}

\def\KO{\overline{K}}

\def\tauU{\underline{\tau}}

\def\gammaU{\underline{\gamma}}
\def\gammaU{\underline{\gamma}}

\def\deltaU{\underline{\delta}}

\def\sigmaU{\underline{\sigma}}

\def\zetaU{\underline{\zeta}}

\def\xiU{\underline{\xi}}

\def\zetaUH{\hat{\underline{\zeta}}}

\def\xiUH{\hat{\xiU}}

 \def\Ct{\mathtt{C}}
 \def\Gt{\mathtt{G}}

\def\tAND{\mbox{ and }}

\RedeclareMathOperator{\Im}{Im}

\def\s{\underline{s}}

\def\C{\mathbb{C}}

\newcommand{\T}{\mathsf{T}}

\newcommand{\I}{\mathbbm{1}}

%% file: matrix.tex
\definecolor{c00}{rgb}{1.00,0.00,0.00} %
\definecolor{c10}{rgb}{1.00,0.50,0.00} %
\definecolor{c01}{rgb}{1.00,1.00,0.00} %
\definecolor{c11}{rgb}{0.00,1.00,0.00} %
\definecolor{c02}{rgb}{0.00,1.00,1.00} %
\definecolor{c12}{rgb}{0.00,0.00,1.00} %
\definecolor{c03}{rgb}{0.50,0.00,1.00} %
\definecolor{c13}{rgb}{1.00,0.00,1.00} %
\definecolor{c04}{rgb}{0.50,0.50,0.50} %
\definecolor{c14}{rgb}{0.75,0.25,0.25} %
\definecolor{c05}{rgb}{0.25,0.75,0.25} %
\definecolor{c15}{rgb}{0.25,0.25,0.75} %
\begin{figure*}[!t]
\begin{align}
 &\HH_X=
\left(
\begin{array}{c}
\HH^{(0)}_X \\
\HH^{(1)}_X 
\end{array}
\right)
=
\left(
\begin{array}{c|c|c||c|c|c}
\phantom{0}\phantom{0}\phantom{0}\phantom{0}\phantom{0}\tikzmark{2A}{\setlength{\fboxsep}{1pt}\colorbox{red!70!black}{\textcolor{white}{1}}}\phantom{0}\phantom{0}&\phantom{0}1\phantom{0}\phantom{0}\phantom{0}\phantom{0}\phantom{0}\phantom{0}&\phantom{0}\phantom{0}1\phantom{0}\phantom{0}\phantom{0}\phantom{0}\phantom{0}&\phantom{0}\phantom{0}\phantom{0}\phantom{0}\phantom{0}1\phantom{0}\phantom{0}&\phantom{0}\phantom{0}\phantom{0}\phantom{0}\phantom{0}\phantom{0}\phantom{0}1&\phantom{0}\phantom{0}\phantom{0}\phantom{0}\phantom{0}\tikzmark{2B}{\setlength{\fboxsep}{1pt}\colorbox{red!70!black}{\textcolor{white}{1}}}\phantom{0}\phantom{0}\\
\phantom{0}\phantom{0}1\phantom{0}\phantom{0}\phantom{0}\phantom{0}\phantom{0}&\phantom{0}\phantom{0}\phantom{0}\phantom{0}\phantom{0}\phantom{0}1\phantom{0}&\phantom{0}\phantom{0}\phantom{0}1\phantom{0}\phantom{0}\phantom{0}\phantom{0}&\phantom{0}\phantom{0}1\phantom{0}\phantom{0}\phantom{0}\phantom{0}\phantom{0}&\phantom{0}\phantom{0}\phantom{0}\phantom{0}1\phantom{0}\phantom{0}\phantom{0}&\phantom{0}\phantom{0}1\phantom{0}\phantom{0}\phantom{0}\phantom{0}\phantom{0}\\
\phantom{0}\phantom{0}\phantom{0}\phantom{0}\phantom{0}\phantom{0}\phantom{0}1&\phantom{0}\phantom{0}\phantom{0}\tikzmark{2E}{\setlength{\fboxsep}{1pt}\colorbox{red!70!black}{\textcolor{white}{1}}}\phantom{0}\phantom{0}\phantom{0}\phantom{0}&\phantom{0}\phantom{0}\phantom{0}\phantom{0}1\phantom{0}\phantom{0}\phantom{0}&\phantom{0}\phantom{0}\phantom{0}\phantom{0}\phantom{0}\phantom{0}\phantom{0}1&\phantom{0}\tikzmark{2F}{\setlength{\fboxsep}{1pt}\colorbox{red!70!black}{\textcolor{white}{1}}}\phantom{0}\phantom{0}\phantom{0}\phantom{0}\phantom{0}\phantom{0}&\phantom{0}\phantom{0}\phantom{0}\phantom{0}\phantom{0}\phantom{0}\phantom{0}1\\
\phantom{0}\phantom{0}\phantom{0}\phantom{0}1\phantom{0}\phantom{0}\phantom{0}&1\phantom{0}\phantom{0}\phantom{0}\phantom{0}\phantom{0}\phantom{0}\phantom{0}&\phantom{0}\phantom{0}\phantom{0}\phantom{0}\phantom{0}\tikzmark{2I}{\setlength{\fboxsep}{1pt}\colorbox{red!70!black}{\textcolor{white}{1}}}\phantom{0}\phantom{0}&\phantom{0}\phantom{0}\phantom{0}\phantom{0}\tikzmark{2J}{\setlength{\fboxsep}{1pt}\colorbox{red!70!black}{\textcolor{white}{1}}}\phantom{0}\phantom{0}\phantom{0}&\phantom{0}\phantom{0}\phantom{0}\phantom{0}\phantom{0}\phantom{0}1\phantom{0}&\phantom{0}\phantom{0}\phantom{0}\phantom{0}1\phantom{0}\phantom{0}\phantom{0}\\
\phantom{0}1\phantom{0}\phantom{0}\phantom{0}\phantom{0}\phantom{0}\phantom{0}&\phantom{0}\phantom{0}\phantom{0}\phantom{0}\phantom{0}1\phantom{0}\phantom{0}&\phantom{0}\phantom{0}\phantom{0}\phantom{0}\phantom{0}\phantom{0}1\phantom{0}&\phantom{0}1\phantom{0}\phantom{0}\phantom{0}\phantom{0}\phantom{0}\phantom{0}&\phantom{0}\phantom{0}\phantom{0}1\phantom{0}\phantom{0}\phantom{0}\phantom{0}&\phantom{0}1\phantom{0}\phantom{0}\phantom{0}\phantom{0}\phantom{0}\phantom{0}\\
\phantom{0}\phantom{0}\phantom{0}\phantom{0}\phantom{0}\phantom{0}1\phantom{0}&\phantom{0}\phantom{0}1\phantom{0}\phantom{0}\phantom{0}\phantom{0}\phantom{0}&\phantom{0}\phantom{0}\phantom{0}\phantom{0}\phantom{0}\phantom{0}\phantom{0}1&\phantom{0}\phantom{0}\phantom{0}\phantom{0}\phantom{0}\phantom{0}1\phantom{0}&1\phantom{0}\phantom{0}\phantom{0}\phantom{0}\phantom{0}\phantom{0}\phantom{0}&\phantom{0}\phantom{0}\phantom{0}\phantom{0}\phantom{0}\phantom{0}1\phantom{0}\\
\phantom{0}\phantom{0}\phantom{0}1\phantom{0}\phantom{0}\phantom{0}\phantom{0}&\phantom{0}\phantom{0}\phantom{0}\phantom{0}\phantom{0}\phantom{0}\phantom{0}1&1\phantom{0}\phantom{0}\phantom{0}\phantom{0}\phantom{0}\phantom{0}\phantom{0}&\phantom{0}\phantom{0}\phantom{0}1\phantom{0}\phantom{0}\phantom{0}\phantom{0}&\phantom{0}\phantom{0}\phantom{0}\phantom{0}\phantom{0}1\phantom{0}\phantom{0}&\phantom{0}\phantom{0}\phantom{0}1\phantom{0}\phantom{0}\phantom{0}\phantom{0}\\
{\setlength{\fboxsep}{1pt}\colorbox{blue!70!black}{\textcolor{white}{1}}}\phantom{0}\phantom{0}\phantom{0}\phantom{0}\phantom{0}\phantom{0}\phantom{0}&\phantom{0}\phantom{0}\phantom{0}\phantom{0}{\setlength{\fboxsep}{1pt}\colorbox{blue!70!black}{\textcolor{white}{1}}}\phantom{0}\phantom{0}\phantom{0}&\phantom{0}{\setlength{\fboxsep}{1pt}\colorbox{blue!70!black}{\textcolor{white}{1}}}\phantom{0}\phantom{0}\phantom{0}\phantom{0}\phantom{0}\phantom{0}&{\setlength{\fboxsep}{1pt}\colorbox{blue!70!black}{\textcolor{white}{1}}}\phantom{0}\phantom{0}\phantom{0}\phantom{0}\phantom{0}\phantom{0}\phantom{0}&\phantom{0}\phantom{0}{\setlength{\fboxsep}{1pt}\colorbox{blue!70!black}{\textcolor{white}{1}}}\phantom{0}\phantom{0}\phantom{0}\phantom{0}\phantom{0}&{\setlength{\fboxsep}{1pt}\colorbox{blue!70!black}{\textcolor{white}{1}}}\phantom{0}\phantom{0}\phantom{0}\phantom{0}\phantom{0}\phantom{0}\phantom{0}\\
\hline
\phantom{0}\phantom{0}1\phantom{0}\phantom{0}\phantom{0}\phantom{0}\phantom{0}&\phantom{0}\phantom{0}\phantom{0}\phantom{0}\phantom{0}1\phantom{0}\phantom{0}&\phantom{0}1\phantom{0}\phantom{0}\phantom{0}\phantom{0}\phantom{0}\phantom{0}&\phantom{0}\phantom{0}\phantom{0}\phantom{0}\phantom{0}1\phantom{0}\phantom{0}&\phantom{0}\phantom{0}\phantom{0}\phantom{0}\phantom{0}1\phantom{0}\phantom{0}&\phantom{0}\phantom{0}\phantom{0}\phantom{0}\phantom{0}\phantom{0}\phantom{0}1\\
\phantom{0}\phantom{0}\phantom{0}1\phantom{0}\phantom{0}\phantom{0}\phantom{0}&\phantom{0}\phantom{0}1\phantom{0}\phantom{0}\phantom{0}\phantom{0}\phantom{0}&\phantom{0}\phantom{0}\phantom{0}\phantom{0}\phantom{0}\phantom{0}1\phantom{0}&\phantom{0}\phantom{0}1\phantom{0}\phantom{0}\phantom{0}\phantom{0}\phantom{0}&\phantom{0}\phantom{0}1\phantom{0}\phantom{0}\phantom{0}\phantom{0}\phantom{0}&\phantom{0}\phantom{0}\phantom{0}\phantom{0}1\phantom{0}\phantom{0}\phantom{0}\\
\phantom{0}\phantom{0}\phantom{0}\phantom{0}1\phantom{0}\phantom{0}\phantom{0}&\phantom{0}\phantom{0}\phantom{0}\phantom{0}\phantom{0}\phantom{0}\phantom{0}1&\phantom{0}\phantom{0}\phantom{0}1\phantom{0}\phantom{0}\phantom{0}\phantom{0}&\phantom{0}\phantom{0}\phantom{0}\phantom{0}\phantom{0}\phantom{0}\phantom{0}1&\phantom{0}\phantom{0}\phantom{0}\phantom{0}\phantom{0}\phantom{0}\phantom{0}1&\phantom{0}1\phantom{0}\phantom{0}\phantom{0}\phantom{0}\phantom{0}\phantom{0}\\
\phantom{0}\phantom{0}\phantom{0}\phantom{0}\phantom{0}\tikzmark{2L}{\setlength{\fboxsep}{1pt}\colorbox{red!70!black}{\textcolor{white}{1}}}\phantom{0}\phantom{0}&\phantom{0}\phantom{0}\phantom{0}\phantom{0}1\phantom{0}\phantom{0}\phantom{0}&1\phantom{0}\phantom{0}\phantom{0}\phantom{0}\phantom{0}\phantom{0}\phantom{0}&\phantom{0}\phantom{0}\phantom{0}\phantom{0}\tikzmark{2K}{\setlength{\fboxsep}{1pt}\colorbox{red!70!black}{\textcolor{white}{1}}}\phantom{0}\phantom{0}\phantom{0}&\phantom{0}\phantom{0}\phantom{0}\phantom{0}1\phantom{0}\phantom{0}\phantom{0}&\phantom{0}\phantom{0}\phantom{0}\phantom{0}\phantom{0}\phantom{0}1\phantom{0}\\
\phantom{0}\phantom{0}\phantom{0}\phantom{0}\phantom{0}\phantom{0}1\phantom{0}&\phantom{0}1\phantom{0}\phantom{0}\phantom{0}\phantom{0}\phantom{0}\phantom{0}&\phantom{0}\phantom{0}\phantom{0}\phantom{0}\phantom{0}\tikzmark{2H}{\setlength{\fboxsep}{1pt}\colorbox{red!70!black}{\textcolor{white}{1}}}\phantom{0}\phantom{0}&\phantom{0}1\phantom{0}\phantom{0}\phantom{0}\phantom{0}\phantom{0}\phantom{0}&\phantom{0}\tikzmark{2G}{\setlength{\fboxsep}{1pt}\colorbox{red!70!black}{\textcolor{white}{1}}}\phantom{0}\phantom{0}\phantom{0}\phantom{0}\phantom{0}\phantom{0}&\phantom{0}\phantom{0}\phantom{0}1\phantom{0}\phantom{0}\phantom{0}\phantom{0}\\
\phantom{0}\phantom{0}\phantom{0}\phantom{0}\phantom{0}\phantom{0}\phantom{0}1&\phantom{0}\phantom{0}\phantom{0}\phantom{0}\phantom{0}\phantom{0}1\phantom{0}&\phantom{0}\phantom{0}1\phantom{0}\phantom{0}\phantom{0}\phantom{0}\phantom{0}&\phantom{0}\phantom{0}\phantom{0}\phantom{0}\phantom{0}\phantom{0}1\phantom{0}&\phantom{0}\phantom{0}\phantom{0}\phantom{0}\phantom{0}\phantom{0}1\phantom{0}&1\phantom{0}\phantom{0}\phantom{0}\phantom{0}\phantom{0}\phantom{0}\phantom{0}\\
1\phantom{0}\phantom{0}\phantom{0}\phantom{0}\phantom{0}\phantom{0}\phantom{0}&\phantom{0}\phantom{0}\phantom{0}\tikzmark{2D}{\setlength{\fboxsep}{1pt}\colorbox{red!70!black}{\textcolor{white}{1}}}\phantom{0}\phantom{0}\phantom{0}\phantom{0}&\phantom{0}\phantom{0}\phantom{0}\phantom{0}\phantom{0}\phantom{0}\phantom{0}1&\phantom{0}\phantom{0}\phantom{0}1\phantom{0}\phantom{0}\phantom{0}\phantom{0}&\phantom{0}\phantom{0}\phantom{0}1\phantom{0}\phantom{0}\phantom{0}\phantom{0}&\phantom{0}\phantom{0}\phantom{0}\phantom{0}\phantom{0}\tikzmark{2C}{\setlength{\fboxsep}{1pt}\colorbox{red!70!black}{\textcolor{white}{1}}}\phantom{0}\phantom{0}\\
\phantom{0}1\phantom{0}\phantom{0}\phantom{0}\phantom{0}\phantom{0}\phantom{0}&1\phantom{0}\phantom{0}\phantom{0}\phantom{0}\phantom{0}\phantom{0}\phantom{0}&\phantom{0}\phantom{0}\phantom{0}\phantom{0}1\phantom{0}\phantom{0}\phantom{0}&1\phantom{0}\phantom{0}\phantom{0}\phantom{0}\phantom{0}\phantom{0}\phantom{0}&1\phantom{0}\phantom{0}\phantom{0}\phantom{0}\phantom{0}\phantom{0}\phantom{0}&\phantom{0}\phantom{0}1\phantom{0}\phantom{0}\phantom{0}\phantom{0}\phantom{0}
\end{array}
\right),
\end{align}
\begin{tikzpicture}[remember picture, overlay]
  \draw[thick, red!50, ->] (pic cs:2A) -- (pic cs:2B);
  \draw[thick, red!50, ->] (pic cs:2B) -- (pic cs:2C);
  \draw[thick, red!50, ->] (pic cs:2C) -- (pic cs:2D);
  \draw[thick, red!50, ->] (pic cs:2D) -- (pic cs:2E);
  \draw[thick, red!50, ->] (pic cs:2E) -- (pic cs:2F);
  \draw[thick, red!50, ->] (pic cs:2F) -- (pic cs:2G);
  \draw[thick, red!50, ->] (pic cs:2G) -- (pic cs:2H);
  \draw[thick, red!50, ->] (pic cs:2H) -- (pic cs:2I);
  \draw[thick, red!50, ->] (pic cs:2I) -- (pic cs:2J);
  \draw[thick, red!50, ->] (pic cs:2J) -- (pic cs:2K);
  \draw[thick, red!50, ->] (pic cs:2K) -- (pic cs:2L);
  \draw[thick, red!50, ->] (pic cs:2L) -- (pic cs:2A);
\end{tikzpicture}
\end{figure*}
\begin{figure*}[!t]
\begin{align}
&\HH_Z=
\left(
\begin{array}{c}
\HH^{(0)}_Z \\
\HH^{(1)}_Z 
\end{array}
\right)
=
\left(\begin{array}{c|c|c||c|c|c}
\phantom{0}\phantom{0}\phantom{0}\phantom{0}\phantom{0}\phantom{0}\phantom{0}1&\phantom{0}\phantom{0}\phantom{0}\phantom{0}\phantom{0}\phantom{0}\phantom{0}1&\phantom{0}\phantom{0}\phantom{0}\phantom{0}\phantom{0}1\phantom{0}\phantom{0}&\phantom{0}\phantom{0}\phantom{0}\phantom{0}\phantom{0}\phantom{0}\phantom{0}1&\phantom{0}\phantom{0}\phantom{0}\phantom{0}\phantom{0}\phantom{0}1\phantom{0}&\phantom{0}\phantom{0}\phantom{0}1\phantom{0}\phantom{0}\phantom{0}\phantom{0}\\
\phantom{0}\phantom{0}\phantom{0}\phantom{0}1\phantom{0}\phantom{0}\phantom{0}&\phantom{0}\phantom{0}\phantom{0}\phantom{0}\tikzmark{3E}{\setlength{\fboxsep}{1pt}\colorbox{blue!71!black}{\textcolor{white}{1}}}\phantom{0}\phantom{0}\phantom{0}&\phantom{0}\phantom{0}1\phantom{0}\phantom{0}\phantom{0}\phantom{0}\phantom{0}&\phantom{0}\phantom{0}\phantom{0}\phantom{0}1\phantom{0}\phantom{0}\phantom{0}&\phantom{0}\phantom{0}\phantom{0}\phantom{0}\phantom{0}\phantom{0}\phantom{0}1&\tikzmark{3F}{\setlength{\fboxsep}{1pt}\colorbox{blue!71!black}{\textcolor{white}{1}}}\phantom{0}\phantom{0}\phantom{0}\phantom{0}\phantom{0}\phantom{0}\phantom{0}\\
\phantom{0}1\phantom{0}\phantom{0}\phantom{0}\phantom{0}\phantom{0}\phantom{0}&\phantom{0}1\phantom{0}\phantom{0}\phantom{0}\phantom{0}\phantom{0}\phantom{0}&\phantom{0}\phantom{0}\phantom{0}\phantom{0}\phantom{0}\phantom{0}\phantom{0}1&\phantom{0}1\phantom{0}\phantom{0}\phantom{0}\phantom{0}\phantom{0}\phantom{0}&1\phantom{0}\phantom{0}\phantom{0}\phantom{0}\phantom{0}\phantom{0}\phantom{0}&\phantom{0}\phantom{0}\phantom{0}\phantom{0}\phantom{0}1\phantom{0}\phantom{0}\\
\phantom{0}\phantom{0}\phantom{0}\phantom{0}\phantom{0}\phantom{0}1\phantom{0}&\phantom{0}\phantom{0}\phantom{0}\phantom{0}\phantom{0}\phantom{0}1\phantom{0}&\phantom{0}\phantom{0}\phantom{0}\phantom{0}1\phantom{0}\phantom{0}\phantom{0}&\phantom{0}\phantom{0}\phantom{0}\phantom{0}\phantom{0}\phantom{0}1\phantom{0}&\phantom{0}1\phantom{0}\phantom{0}\phantom{0}\phantom{0}\phantom{0}\phantom{0}&\phantom{0}\phantom{0}1\phantom{0}\phantom{0}\phantom{0}\phantom{0}\phantom{0}\\
\phantom{0}\phantom{0}\phantom{0}1\phantom{0}\phantom{0}\phantom{0}\phantom{0}&\phantom{0}\phantom{0}\phantom{0}1\phantom{0}\phantom{0}\phantom{0}\phantom{0}&\phantom{0}\tikzmark{3I}{\setlength{\fboxsep}{1pt}\colorbox{blue!71!black}{\textcolor{white}{1}}}\phantom{0}\phantom{0}\phantom{0}\phantom{0}\phantom{0}\phantom{0}&\phantom{0}\phantom{0}\phantom{0}1\phantom{0}\phantom{0}\phantom{0}\phantom{0}&\phantom{0}\phantom{0}\tikzmark{3J}{\setlength{\fboxsep}{1pt}\colorbox{blue!71!black}{\textcolor{white}{1}}}\phantom{0}\phantom{0}\phantom{0}\phantom{0}\phantom{0}&\phantom{0}\phantom{0}\phantom{0}\phantom{0}\phantom{0}\phantom{0}\phantom{0}1\\
\tikzmark{3A}{\setlength{\fboxsep}{1pt}\colorbox{blue!71!black}{\textcolor{white}{1}}}\phantom{0}\phantom{0}\phantom{0}\phantom{0}\phantom{0}\phantom{0}\phantom{0}&1\phantom{0}\phantom{0}\phantom{0}\phantom{0}\phantom{0}\phantom{0}\phantom{0}&\phantom{0}\phantom{0}\phantom{0}\phantom{0}\phantom{0}\phantom{0}1\phantom{0}&\tikzmark{3B}{\setlength{\fboxsep}{1pt}\colorbox{blue!71!black}{\textcolor{white}{1}}}\phantom{0}\phantom{0}\phantom{0}\phantom{0}\phantom{0}\phantom{0}\phantom{0}&\phantom{0}\phantom{0}\phantom{0}1\phantom{0}\phantom{0}\phantom{0}\phantom{0}&\phantom{0}\phantom{0}\phantom{0}\phantom{0}1\phantom{0}\phantom{0}\phantom{0}\\
\phantom{0}\phantom{0}\phantom{0}\phantom{0}\phantom{0}1\phantom{0}\phantom{0}&\phantom{0}\phantom{0}\phantom{0}\phantom{0}\phantom{0}1\phantom{0}\phantom{0}&\phantom{0}\phantom{0}\phantom{0}1\phantom{0}\phantom{0}\phantom{0}\phantom{0}&\phantom{0}\phantom{0}\phantom{0}\phantom{0}\phantom{0}1\phantom{0}\phantom{0}&\phantom{0}\phantom{0}\phantom{0}\phantom{0}1\phantom{0}\phantom{0}\phantom{0}&\phantom{0}1\phantom{0}\phantom{0}\phantom{0}\phantom{0}\phantom{0}\phantom{0}\\
\phantom{0}\phantom{0}1\phantom{0}\phantom{0}\phantom{0}\phantom{0}\phantom{0}&\phantom{0}\phantom{0}1\phantom{0}\phantom{0}\phantom{0}\phantom{0}\phantom{0}&1\phantom{0}\phantom{0}\phantom{0}\phantom{0}\phantom{0}\phantom{0}\phantom{0}&\phantom{0}\phantom{0}1\phantom{0}\phantom{0}\phantom{0}\phantom{0}\phantom{0}&\phantom{0}\phantom{0}\phantom{0}\phantom{0}\phantom{0}1\phantom{0}\phantom{0}&\phantom{0}\phantom{0}\phantom{0}\phantom{0}\phantom{0}\phantom{0}1\phantom{0}\\
\hline
\phantom{0}\phantom{0}\phantom{0}\phantom{0}\phantom{0}1\phantom{0}\phantom{0}&\phantom{0}\phantom{0}\phantom{0}\phantom{0}\phantom{0}\phantom{0}\phantom{0}1&\phantom{0}\phantom{0}\phantom{0}\phantom{0}\phantom{0}\phantom{0}\phantom{0}1&\phantom{0}\phantom{0}\phantom{0}1\phantom{0}\phantom{0}\phantom{0}\phantom{0}&\phantom{0}\phantom{0}\phantom{0}\phantom{0}\phantom{0}\phantom{0}\phantom{0}1&\phantom{0}\phantom{0}\phantom{0}\phantom{0}\phantom{0}\phantom{0}1\phantom{0}\\
\phantom{0}\phantom{0}1\phantom{0}\phantom{0}\phantom{0}\phantom{0}\phantom{0}&\phantom{0}\phantom{0}\phantom{0}\phantom{0}\tikzmark{3D}{\setlength{\fboxsep}{1pt}\colorbox{blue!71!black}{\textcolor{white}{1}}}\phantom{0}\phantom{0}\phantom{0}&\phantom{0}\phantom{0}\phantom{0}\phantom{0}1\phantom{0}\phantom{0}\phantom{0}&\tikzmark{3C}{\setlength{\fboxsep}{1pt}\colorbox{blue!71!black}{\textcolor{white}{1}}}\phantom{0}\phantom{0}\phantom{0}\phantom{0}\phantom{0}\phantom{0}\phantom{0}&\phantom{0}\phantom{0}\phantom{0}\phantom{0}1\phantom{0}\phantom{0}\phantom{0}&\phantom{0}\phantom{0}\phantom{0}\phantom{0}\phantom{0}\phantom{0}\phantom{0}1\\
\phantom{0}\phantom{0}\phantom{0}\phantom{0}\phantom{0}\phantom{0}\phantom{0}1&\phantom{0}1\phantom{0}\phantom{0}\phantom{0}\phantom{0}\phantom{0}\phantom{0}&\phantom{0}\tikzmark{3H}{\setlength{\fboxsep}{1pt}\colorbox{blue!71!black}{\textcolor{white}{1}}}\phantom{0}\phantom{0}\phantom{0}\phantom{0}\phantom{0}\phantom{0}&\phantom{0}\phantom{0}\phantom{0}\phantom{0}\phantom{0}1\phantom{0}\phantom{0}&\phantom{0}1\phantom{0}\phantom{0}\phantom{0}\phantom{0}\phantom{0}\phantom{0}&\tikzmark{3G}{\setlength{\fboxsep}{1pt}\colorbox{blue!71!black}{\textcolor{white}{1}}}\phantom{0}\phantom{0}\phantom{0}\phantom{0}\phantom{0}\phantom{0}\phantom{0}\\
\phantom{0}\phantom{0}\phantom{0}\phantom{0}1\phantom{0}\phantom{0}\phantom{0}&\phantom{0}\phantom{0}\phantom{0}\phantom{0}\phantom{0}\phantom{0}1\phantom{0}&\phantom{0}\phantom{0}\phantom{0}\phantom{0}\phantom{0}\phantom{0}1\phantom{0}&\phantom{0}\phantom{0}1\phantom{0}\phantom{0}\phantom{0}\phantom{0}\phantom{0}&\phantom{0}\phantom{0}\phantom{0}\phantom{0}\phantom{0}\phantom{0}1\phantom{0}&\phantom{0}1\phantom{0}\phantom{0}\phantom{0}\phantom{0}\phantom{0}\phantom{0}\\
\phantom{0}1\phantom{0}\phantom{0}\phantom{0}\phantom{0}\phantom{0}\phantom{0}&\phantom{0}\phantom{0}\phantom{0}1\phantom{0}\phantom{0}\phantom{0}\phantom{0}&\phantom{0}\phantom{0}\phantom{0}1\phantom{0}\phantom{0}\phantom{0}\phantom{0}&\phantom{0}\phantom{0}\phantom{0}\phantom{0}\phantom{0}\phantom{0}\phantom{0}1&\phantom{0}\phantom{0}\phantom{0}1\phantom{0}\phantom{0}\phantom{0}\phantom{0}&\phantom{0}\phantom{0}1\phantom{0}\phantom{0}\phantom{0}\phantom{0}\phantom{0}\\
\phantom{0}\phantom{0}\phantom{0}\phantom{0}\phantom{0}\phantom{0}1\phantom{0}&1\phantom{0}\phantom{0}\phantom{0}\phantom{0}\phantom{0}\phantom{0}\phantom{0}&1\phantom{0}\phantom{0}\phantom{0}\phantom{0}\phantom{0}\phantom{0}\phantom{0}&\phantom{0}\phantom{0}\phantom{0}\phantom{0}1\phantom{0}\phantom{0}\phantom{0}&1\phantom{0}\phantom{0}\phantom{0}\phantom{0}\phantom{0}\phantom{0}\phantom{0}&\phantom{0}\phantom{0}\phantom{0}1\phantom{0}\phantom{0}\phantom{0}\phantom{0}\\
\phantom{0}\phantom{0}\phantom{0}1\phantom{0}\phantom{0}\phantom{0}\phantom{0}&\phantom{0}\phantom{0}\phantom{0}\phantom{0}\phantom{0}1\phantom{0}\phantom{0}&\phantom{0}\phantom{0}\phantom{0}\phantom{0}\phantom{0}1\phantom{0}\phantom{0}&\phantom{0}1\phantom{0}\phantom{0}\phantom{0}\phantom{0}\phantom{0}\phantom{0}&\phantom{0}\phantom{0}\phantom{0}\phantom{0}\phantom{0}1\phantom{0}\phantom{0}&\phantom{0}\phantom{0}\phantom{0}\phantom{0}1\phantom{0}\phantom{0}\phantom{0}\\
\tikzmark{3L}{\setlength{\fboxsep}{1pt}\colorbox{blue!71!black}{\textcolor{white}{1}}}\phantom{0}\phantom{0}\phantom{0}\phantom{0}\phantom{0}\phantom{0}\phantom{0}&\phantom{0}\phantom{0}1\phantom{0}\phantom{0}\phantom{0}\phantom{0}\phantom{0}&\phantom{0}\phantom{0}1\phantom{0}\phantom{0}\phantom{0}\phantom{0}\phantom{0}&\phantom{0}\phantom{0}\phantom{0}\phantom{0}\phantom{0}\phantom{0}1\phantom{0}&\phantom{0}\phantom{0}\tikzmark{3K}{\setlength{\fboxsep}{1pt}\colorbox{blue!71!black}{\textcolor{white}{1}}}\phantom{0}\phantom{0}\phantom{0}\phantom{0}\phantom{0}&\phantom{0}\phantom{0}\phantom{0}\phantom{0}\phantom{0}1\phantom{0}\phantom{0}
\end{array}
\right).
\end{align}
\begin{tikzpicture}[remember picture, overlay]
  \draw[thick, blue!50, ->] (pic cs:3A) -- (pic cs:3B);
  \draw[thick, blue!50, ->] (pic cs:3B) -- (pic cs:3C);
  \draw[thick, blue!50, ->] (pic cs:3C) -- (pic cs:3D);
  \draw[thick, blue!50, ->] (pic cs:3D) -- (pic cs:3E);
  \draw[thick, blue!50, ->] (pic cs:3E) -- (pic cs:3F);
  \draw[thick, blue!50, ->] (pic cs:3F) -- (pic cs:3G);
  \draw[thick, blue!50, ->] (pic cs:3G) -- (pic cs:3H);
  \draw[thick, blue!50, ->] (pic cs:3H) -- (pic cs:3I);
  \draw[thick, blue!50, ->] (pic cs:3I) -- (pic cs:3J);
  \draw[thick, blue!50, ->] (pic cs:3J) -- (pic cs:3K);
  \draw[thick, blue!50, ->] (pic cs:3K) -- (pic cs:3L);
  \draw[thick, blue!50, ->] (pic cs:3L) -- (pic cs:3A);
\end{tikzpicture}
\end{figure*}

%% file: blue_row.tex
 &\hU_{kr}^\top=\left(
 \alpha^{200}\hspace{1mm}0 \hspace{1mm}0 \hspace{1mm}0 \hspace{1mm}0 \hspace{1mm}0 \hspace{1mm}0 \hspace{1mm}0 \hspace{1mm}0 \hspace{1mm}0 \hspace{1mm}0 \hspace{1mm}0 \hspace{1mm}\alpha^{200}\hspace{1mm}0 \hspace{1mm}0 \hspace{1mm}0 \hspace{1mm}0 \hspace{1mm}\alpha^{170}\hspace{1mm}0 \hspace{1mm}0 \hspace{1mm}0 \hspace{1mm}0 \hspace{1mm}0 \hspace{1mm}0 \hspace{1mm}\alpha^{238}\hspace{1mm}0 \hspace{1mm}0 \hspace{1mm}0 \hspace{1mm}0 \hspace{1mm}0 \hspace{1mm}0 \hspace{1mm}0 \hspace{1mm}0 \hspace{1mm}0 \hspace{1mm}\alpha^{167}\hspace{1mm}0 \hspace{1mm}0 \hspace{1mm}0 \hspace{1mm}0 \hspace{1mm}0 \hspace{1mm}\alpha^{95}\hspace{1mm}0 \hspace{1mm}0 \hspace{1mm}0 \hspace{1mm}0 \hspace{1mm}0 \hspace{1mm}0 \hspace{1mm}0 
 \right).

%% file: blue_matrix.tex
\begin{align}
\left(
\begin{array}{c@{\hspace{1mm}}c@{\hspace{1mm}}c@{\hspace{1mm}}c@{\hspace{1mm}}c@{\hspace{1mm}}c@{\hspace{1mm}}}
&{\setlength{\fboxsep}{1pt}\colorbox{blue!70!black}{\textcolor{white}{$\mathtt{7C}$}}}&&&&{\setlength{\fboxsep}{1pt}\colorbox{blue!70!black}{\textcolor{white}{$\mathtt{E5}$}}}\\
&&{\setlength{\fboxsep}{1pt}\colorbox{blue!70!black}{\textcolor{white}{$\mathtt{12}$}}}&&{\setlength{\fboxsep}{1pt}\colorbox{blue!70!black}{\textcolor{white}{$\mathtt{15}$}}}\\
{\setlength{\fboxsep}{1pt}\colorbox{blue!70!black}{\textcolor{white}{$\mathtt{3E}$}}}&&&{\setlength{\fboxsep}{1pt}\colorbox{blue!70!black}{\textcolor{white}{$\mathtt{18}$}}}\\
&{\setlength{\fboxsep}{1pt}\colorbox{blue!70!black}{\textcolor{white}{$\mathtt{1D}$}}}&&{\setlength{\fboxsep}{1pt}\colorbox{blue!70!black}{\textcolor{white}{$\mathtt{F6}$}}}\\
&&{\setlength{\fboxsep}{1pt}\colorbox{blue!70!black}{\textcolor{white}{$\mathtt{C2}$}}}&&&{\setlength{\fboxsep}{1pt}\colorbox{blue!70!black}{\textcolor{white}{$\mathtt{0E}$}}}\\
{\setlength{\fboxsep}{1pt}\colorbox{blue!70!black}{\textcolor{white}{$\mathtt{6F}$}}}&&&&{\setlength{\fboxsep}{1pt}\colorbox{blue!70!black}{\textcolor{white}{$\mathtt{90}$}}}
\end{array}\right)
\left(\begin{array}{c}
 {\setlength{\fboxsep}{1pt}\colorbox{blue!70!black}{\textcolor{white}{$\mathtt{C8}$}}}\\
 {\setlength{\fboxsep}{1pt}\colorbox{blue!70!black}{\textcolor{white}{$\mathtt{C8}$}}}\\
 {\setlength{\fboxsep}{1pt}\colorbox{blue!70!black}{\textcolor{white}{$\mathtt{AA}$}}}\\
 {\setlength{\fboxsep}{1pt}\colorbox{blue!70!black}{\textcolor{white}{$\mathtt{EE}$}}}\\
 {\setlength{\fboxsep}{1pt}\colorbox{blue!70!black}{\textcolor{white}{$\mathtt{A7}$}}}\\
 {\setlength{\fboxsep}{1pt}\colorbox{blue!70!black}{\textcolor{white}{$\mathtt{5F}$}}}
\end{array}
\right)
=
\zeroU
\end{align}

%% file: red_matrix.tex
\begin{align}
\left(
\begin{array}{c@{\hspace{1mm}}c@{\hspace{1mm}}c@{\hspace{1mm}}c@{\hspace{1mm}}c@{\hspace{1mm}}c@{\hspace{1mm}}}
{\setlength{\fboxsep}{1pt}\colorbox{red!70!black}{\textcolor{white}{$\mathtt{3F}$}}}&&&&&{\setlength{\fboxsep}{1pt}\colorbox{red!70!black}{\textcolor{white}{$\mathtt{F1}$}}}\\
&{\setlength{\fboxsep}{1pt}\colorbox{red!70!black}{\textcolor{white}{$\mathtt{4C}$}}}&&&{\setlength{\fboxsep}{1pt}\colorbox{red!70!black}{\textcolor{white}{$\mathtt{78}$}}}\\
&&{\setlength{\fboxsep}{1pt}\colorbox{red!70!black}{\textcolor{white}{$\mathtt{E1}$}}}&{\setlength{\fboxsep}{1pt}\colorbox{red!70!black}{\textcolor{white}{$\mathtt{89}$}}}&\\
{\setlength{\fboxsep}{1pt}\colorbox{red!70!black}{\textcolor{white}{$\mathtt{8C}$}}}&&&{\setlength{\fboxsep}{1pt}\colorbox{red!70!black}{\textcolor{white}{$\mathtt{16}$}}}\\
&&{\setlength{\fboxsep}{1pt}\colorbox{red!70!black}{\textcolor{white}{$\mathtt{A4}$}}}&&{\setlength{\fboxsep}{1pt}\colorbox{red!70!black}{\textcolor{white}{$\mathtt{85}$}}}&\\
&{\setlength{\fboxsep}{1pt}\colorbox{red!70!black}{\textcolor{white}{$\mathtt{F3}$}}}&&&&{\setlength{\fboxsep}{1pt}\colorbox{red!70!black}{\textcolor{white}{$\mathtt{D4}$}}}
\end{array}\right)
\end{align}

%% file: factor_graph.tex
\begin{figure*}[htbp]
\centering
\begin{tikzpicture}[scale=1,
  roundnode/.style={circle, draw=black, fill=grey, minimum size=4pt, inner sep=1pt},
  squarenode/.style={rectangle, draw=black, fill=grey, minimum size=4pt, inner sep=1pt},
  channelnode/.style={rectangle, draw=black, fill=white, minimum size=4pt, inner sep=1pt},
  >=stealth
]

\foreach \i in {0,...,15} {
  \pgfmathsetmacro{\yy}{\i * 5.5 / 8}
  \node[squarenode] (A\i) at (-2, -\yy) {};
}

\foreach \i in {0,...,47} {
  \pgfmathsetmacro{\yy}{\i * 5.5 / 24}
  \node[roundnode]  (B\i) at (2, -\yy) {};
  \node[channelnode] (C\i) at (4, -\yy) {};
  \node[roundnode]  (D\i) at (6, -\yy) {};
}

\foreach \i in {0,...,15} {
  \pgfmathsetmacro{\yy}{\i * 5.5 / 8}
  \node[squarenode] (E\i) at (10, -\yy) {};
}

\draw[grey!50!black, thin] (A0.east) -- (B7.west);
\draw[grey!50!black, thin] (A0.east) -- (B15.west);
\draw[grey!50!black, thin] (A0.east) -- (B21.west);
\draw[grey!50!black, thin] (A0.east) -- (B31.west);
\draw[grey!50!black, thin] (A0.east) -- (B38.west);
\draw[grey!50!black, thin] (A0.east) -- (B43.west);
\draw[grey!50!black, thin] (A1.east) -- (B4.west);
\draw[blue!70!black, very thick] (A1.east) -- (B12.west);
\draw[grey!50!black, thin] (A1.east) -- (B18.west);
\draw[grey!50!black, thin] (A1.east) -- (B28.west);
\draw[grey!50!black, thin] (A1.east) -- (B39.west);
\draw[blue!70!black, very thick] (A1.east) -- (B40.west);
\draw[grey!50!black, thin] (A2.east) -- (B1.west);
\draw[grey!50!black, thin] (A2.east) -- (B9.west);
\draw[grey!50!black, thin] (A2.east) -- (B23.west);
\draw[grey!50!black, thin] (A2.east) -- (B25.west);
\draw[grey!50!black, thin] (A2.east) -- (B32.west);
\draw[grey!50!black, thin] (A2.east) -- (B45.west);
\draw[grey!50!black, thin] (A3.east) -- (B6.west);
\draw[grey!50!black, thin] (A3.east) -- (B14.west);
\draw[grey!50!black, thin] (A3.east) -- (B20.west);
\draw[grey!50!black, thin] (A3.east) -- (B30.west);
\draw[grey!50!black, thin] (A3.east) -- (B33.west);
\draw[grey!50!black, thin] (A3.east) -- (B42.west);
\draw[grey!50!black, thin] (A4.east) -- (B3.west);
\draw[grey!50!black, thin] (A4.east) -- (B11.west);
\draw[blue!70!black, very thick] (A4.east) -- (B17.west);
\draw[grey!50!black, thin] (A4.east) -- (B27.west);
\draw[blue!70!black, very thick] (A4.east) -- (B34.west);
\draw[grey!50!black, thin] (A4.east) -- (B47.west);
\draw[blue!70!black, very thick] (A5.east) -- (B0.west);
\draw[grey!50!black, thin] (A5.east) -- (B8.west);
\draw[grey!50!black, thin] (A5.east) -- (B22.west);
\draw[blue!70!black, very thick] (A5.east) -- (B24.west);
\draw[grey!50!black, thin] (A5.east) -- (B35.west);
\draw[grey!50!black, thin] (A5.east) -- (B44.west);
\draw[grey!50!black, thin] (A6.east) -- (B5.west);
\draw[grey!50!black, thin] (A6.east) -- (B13.west);
\draw[grey!50!black, thin] (A6.east) -- (B19.west);
\draw[grey!50!black, thin] (A6.east) -- (B29.west);
\draw[grey!50!black, thin] (A6.east) -- (B36.west);
\draw[grey!50!black, thin] (A6.east) -- (B41.west);
\draw[grey!50!black, thin] (A7.east) -- (B2.west);
\draw[grey!50!black, thin] (A7.east) -- (B10.west);
\draw[grey!50!black, thin] (A7.east) -- (B16.west);
\draw[grey!50!black, thin] (A7.east) -- (B26.west);
\draw[grey!50!black, thin] (A7.east) -- (B37.west);
\draw[grey!50!black, thin] (A7.east) -- (B46.west);
\draw[grey!50!black, thin] (A8.east) -- (B5.west);
\draw[grey!50!black, thin] (A8.east) -- (B15.west);
\draw[grey!50!black, thin] (A8.east) -- (B23.west);
\draw[grey!50!black, thin] (A8.east) -- (B27.west);
\draw[grey!50!black, thin] (A8.east) -- (B39.west);
\draw[grey!50!black, thin] (A8.east) -- (B46.west);
\draw[grey!50!black, thin] (A9.east) -- (B2.west);
\draw[blue!70!black, very thick] (A9.east) -- (B12.west);
\draw[grey!50!black, thin] (A9.east) -- (B20.west);
\draw[blue!70!black, very thick] (A9.east) -- (B24.west);
\draw[grey!50!black, thin] (A9.east) -- (B36.west);
\draw[grey!50!black, thin] (A9.east) -- (B47.west);
\draw[grey!50!black, thin] (A10.east) -- (B7.west);
\draw[grey!50!black, thin] (A10.east) -- (B9.west);
\draw[blue!70!black, very thick] (A10.east) -- (B17.west);
\draw[grey!50!black, thin] (A10.east) -- (B29.west);
\draw[grey!50!black, thin] (A10.east) -- (B33.west);
\draw[blue!70!black, very thick] (A10.east) -- (B40.west);
\draw[grey!50!black, thin] (A11.east) -- (B4.west);
\draw[grey!50!black, thin] (A11.east) -- (B14.west);
\draw[grey!50!black, thin] (A11.east) -- (B22.west);
\draw[grey!50!black, thin] (A11.east) -- (B26.west);
\draw[grey!50!black, thin] (A11.east) -- (B38.west);
\draw[grey!50!black, thin] (A11.east) -- (B41.west);
\draw[grey!50!black, thin] (A12.east) -- (B1.west);
\draw[grey!50!black, thin] (A12.east) -- (B11.west);
\draw[grey!50!black, thin] (A12.east) -- (B19.west);
\draw[grey!50!black, thin] (A12.east) -- (B31.west);
\draw[grey!50!black, thin] (A12.east) -- (B35.west);
\draw[grey!50!black, thin] (A12.east) -- (B42.west);
\draw[grey!50!black, thin] (A13.east) -- (B6.west);
\draw[grey!50!black, thin] (A13.east) -- (B8.west);
\draw[grey!50!black, thin] (A13.east) -- (B16.west);
\draw[grey!50!black, thin] (A13.east) -- (B28.west);
\draw[grey!50!black, thin] (A13.east) -- (B32.west);
\draw[grey!50!black, thin] (A13.east) -- (B43.west);
\draw[grey!50!black, thin] (A14.east) -- (B3.west);
\draw[grey!50!black, thin] (A14.east) -- (B13.west);
\draw[grey!50!black, thin] (A14.east) -- (B21.west);
\draw[grey!50!black, thin] (A14.east) -- (B25.west);
\draw[grey!50!black, thin] (A14.east) -- (B37.west);
\draw[grey!50!black, thin] (A14.east) -- (B44.west);
\draw[blue!70!black, very thick] (A15.east) -- (B0.west);
\draw[grey!50!black, thin] (A15.east) -- (B10.west);
\draw[grey!50!black, thin] (A15.east) -- (B18.west);
\draw[grey!50!black, thin] (A15.east) -- (B30.west);
\draw[blue!70!black, very thick] (A15.east) -- (B34.west);
\draw[grey!50!black, thin] (A15.east) -- (B45.west);
\foreach \i in {0,...,47} {
  \draw[black, thin] (B\i) -- (C\i);
  \draw[black, thin] (C\i) -- (D\i);
}

\draw[red!70!black, very thick] (D5.east) -- (E0.west);
\draw[grey!50!black, thin] (D9.east) -- (E0.west);
\draw[grey!50!black, thin] (D18.east) -- (E0.west);
\draw[grey!50!black, thin] (D29.east) -- (E0.west);
\draw[grey!50!black, thin] (D39.east) -- (E0.west);
\draw[red!70!black, very thick] (D45.east) -- (E0.west);
\draw[grey!50!black, thin] (D2.east) -- (E1.west);
\draw[grey!50!black, thin] (D14.east) -- (E1.west);
\draw[grey!50!black, thin] (D19.east) -- (E1.west);
\draw[grey!50!black, thin] (D26.east) -- (E1.west);
\draw[grey!50!black, thin] (D36.east) -- (E1.west);
\draw[grey!50!black, thin] (D42.east) -- (E1.west);
\draw[grey!50!black, thin] (D7.east) -- (E2.west);
\draw[red!70!black, very thick] (D11.east) -- (E2.west);
\draw[grey!50!black, thin] (D20.east) -- (E2.west);
\draw[grey!50!black, thin] (D31.east) -- (E2.west);
\draw[red!70!black, very thick] (D33.east) -- (E2.west);
\draw[grey!50!black, thin] (D47.east) -- (E2.west);
\draw[grey!50!black, thin] (D4.east) -- (E3.west);
\draw[grey!50!black, thin] (D8.east) -- (E3.west);
\draw[red!70!black, very thick] (D21.east) -- (E3.west);
\draw[red!70!black, very thick] (D28.east) -- (E3.west);
\draw[grey!50!black, thin] (D38.east) -- (E3.west);
\draw[grey!50!black, thin] (D44.east) -- (E3.west);
\draw[grey!50!black, thin] (D1.east) -- (E4.west);
\draw[grey!50!black, thin] (D13.east) -- (E4.west);
\draw[grey!50!black, thin] (D22.east) -- (E4.west);
\draw[grey!50!black, thin] (D25.east) -- (E4.west);
\draw[grey!50!black, thin] (D35.east) -- (E4.west);
\draw[grey!50!black, thin] (D41.east) -- (E4.west);
\draw[grey!50!black, thin] (D6.east) -- (E5.west);
\draw[grey!50!black, thin] (D10.east) -- (E5.west);
\draw[grey!50!black, thin] (D23.east) -- (E5.west);
\draw[grey!50!black, thin] (D30.east) -- (E5.west);
\draw[grey!50!black, thin] (D32.east) -- (E5.west);
\draw[grey!50!black, thin] (D46.east) -- (E5.west);
\draw[grey!50!black, thin] (D3.east) -- (E6.west);
\draw[grey!50!black, thin] (D15.east) -- (E6.west);
\draw[grey!50!black, thin] (D16.east) -- (E6.west);
\draw[grey!50!black, thin] (D27.east) -- (E6.west);
\draw[grey!50!black, thin] (D37.east) -- (E6.west);
\draw[grey!50!black, thin] (D43.east) -- (E6.west);
\draw[blue!70!black, very thick] (D0.east) -- (E7.west);
\draw[blue!70!black, very thick] (D12.east) -- (E7.west);
\draw[blue!70!black, very thick] (D17.east) -- (E7.west);
\draw[blue!70!black, very thick] (D24.east) -- (E7.west);
\draw[blue!70!black, very thick] (D34.east) -- (E7.west);
\draw[blue!70!black, very thick] (D40.east) -- (E7.west);
\draw[grey!50!black, thin] (D2.east) -- (E8.west);
\draw[grey!50!black, thin] (D13.east) -- (E8.west);
\draw[grey!50!black, thin] (D17.east) -- (E8.west);
\draw[grey!50!black, thin] (D29.east) -- (E8.west);
\draw[grey!50!black, thin] (D37.east) -- (E8.west);
\draw[grey!50!black, thin] (D47.east) -- (E8.west);
\draw[grey!50!black, thin] (D3.east) -- (E9.west);
\draw[grey!50!black, thin] (D10.east) -- (E9.west);
\draw[grey!50!black, thin] (D22.east) -- (E9.west);
\draw[grey!50!black, thin] (D26.east) -- (E9.west);
\draw[grey!50!black, thin] (D34.east) -- (E9.west);
\draw[grey!50!black, thin] (D44.east) -- (E9.west);
\draw[grey!50!black, thin] (D4.east) -- (E10.west);
\draw[grey!50!black, thin] (D15.east) -- (E10.west);
\draw[grey!50!black, thin] (D19.east) -- (E10.west);
\draw[grey!50!black, thin] (D31.east) -- (E10.west);
\draw[grey!50!black, thin] (D39.east) -- (E10.west);
\draw[grey!50!black, thin] (D41.east) -- (E10.west);
\draw[red!70!black, very thick] (D5.east) -- (E11.west);
\draw[grey!50!black, thin] (D12.east) -- (E11.west);
\draw[grey!50!black, thin] (D16.east) -- (E11.west);
\draw[red!70!black, very thick] (D28.east) -- (E11.west);
\draw[grey!50!black, thin] (D36.east) -- (E11.west);
\draw[grey!50!black, thin] (D46.east) -- (E11.west);
\draw[grey!50!black, thin] (D6.east) -- (E12.west);
\draw[grey!50!black, thin] (D9.east) -- (E12.west);
\draw[red!70!black, very thick] (D21.east) -- (E12.west);
\draw[grey!50!black, thin] (D25.east) -- (E12.west);
\draw[red!70!black, very thick] (D33.east) -- (E12.west);
\draw[grey!50!black, thin] (D43.east) -- (E12.west);
\draw[grey!50!black, thin] (D7.east) -- (E13.west);
\draw[grey!50!black, thin] (D14.east) -- (E13.west);
\draw[grey!50!black, thin] (D18.east) -- (E13.west);
\draw[grey!50!black, thin] (D30.east) -- (E13.west);
\draw[grey!50!black, thin] (D38.east) -- (E13.west);
\draw[grey!50!black, thin] (D40.east) -- (E13.west);
\draw[grey!50!black, thin] (D0.east) -- (E14.west);
\draw[red!70!black, very thick] (D11.east) -- (E14.west);
\draw[grey!50!black, thin] (D23.east) -- (E14.west);
\draw[grey!50!black, thin] (D27.east) -- (E14.west);
\draw[grey!50!black, thin] (D35.east) -- (E14.west);
\draw[red!70!black, very thick] (D45.east) -- (E14.west);
\draw[grey!50!black, thin] (D1.east) -- (E15.west);
\draw[grey!50!black, thin] (D8.east) -- (E15.west);
\draw[grey!50!black, thin] (D20.east) -- (E15.west);
\draw[grey!50!black, thin] (D24.east) -- (E15.west);
\draw[grey!50!black, thin] (D32.east) -- (E15.west);
\draw[grey!50!black, thin] (D42.east) -- (E15.west);
\node at (-2,0.5) {$\I[\sum_j \delta_{ij}\xi_j = \sigma_i]$};
\node at (2,0.5) {$\xi_j$};
\node at (4,0.5) {$p(\xi_j,\zeta_j)$};
\node at (6,0.5) {$\zeta_j$};
\node at (10,0.5) {$\I[\sum_j \gamma_{ij}\zeta_j = \tau_i]$};
\node at (0.0,-11.5) {$\overrightarrow{\nu^{(\ell),X}_{ji}(\xi_j)}$};\node at (3.0,-11.5) {$\overrightarrow{\lambda^{(\ell),X}_{j}(\xi_j)}$};\node at (5.0,-11.5) {$\overleftarrow{\lambda^{(\ell),Z}_{j}(\zeta_j)}$};\node at (8.0,-11.5) {$\overleftarrow{\nu^{(\ell),Z}_{ji}(\zeta_j)}$};
\node at (0.0,-12.5) {$\overleftarrow{\mu^{(\ell),X}_{ij}(\xi_j)}$};\node at (3.0,-12.5) {$\overleftarrow{\kappa^{(\ell),X}_{j}(\xi_j)}$};\node at (5.0,-12.5) {$\overrightarrow{\kappa^{(\ell),Z}_{j}(\zeta_j)}$};\node at (8.0,-12.5) {$\overrightarrow{\mu^{(\ell),Z}_{ij}(\zeta_j)}$};
\end{tikzpicture}
\caption{Factor graph corresponding to Eq.~\eqref{022413_12May25}.  
The joint BP algorithm can be interpreted as a message-passing algorithm executed on this graph.  
The directions of the messages are indicated by arrows.}
\label{025851_16May25}
\end{figure*}

%% file: decoding_stats_table.tex
\begin{table*}
 \caption{Decoding Statistics over Iterations for $J=2,L=6,P=6500,N=LP=39000,M=JP=13000, p_D=9.435\%$ and $d=8$.}
 \label{tab:full_decoding_stats}
\begin{center}
  \begin{tabular}{c||c|c|c|c||c|c|c|c}
 &\multicolumn{4}{c||}{\text{Estimation for $\xiU$}}&\multicolumn{4}{c}{\text{Estimation for $\zetaU$}}\\\hline
 $\ell$&    $|K_{\mathrm{err}}^{(\ell)}|$ &$|K_d^{(\ell)}|$ &    $|I_{\mathrm{err}}^{(\ell)}|$ &$|I_d^{(\ell)}|$        &    $|K_{\mathrm{err}}^{(\ell)}|$ &$|K_d^{(\ell)}|$ &    $|I_{\mathrm{err}}^{(\ell)}|$ &$|I_d^{(\ell)}|$ \\
 \hline\hline
  0&14944 &    0 & 9689 & 9689&15017 &    0 & 9741 & 9741 \\
  1&13731 & 4270 & 8618 &10371&13845 & 4165 & 8677 &10399 \\
  2&12875 & 6986 & 7676 &10631&12959 & 6864 & 7791 &10656 \\
  3&12108 & 8776 & 7036 &10757&12306 & 8660 & 7178 &10791 \\
  4&11693 &10053 & 6558 &10852&11765 &10017 & 6717 &10883 \\
  5&11297 &11035 & 6221 &10907&11370 &11022 & 6304 &10941 \\
  6&10866 &11808 & 5862 &10951&11043 &11808 & 6028 &10986 \\
  7&10542 &12446 & 5640 &10974&10667 &12518 & 5745 &11027 \\
  8&10300 &12950 & 5464 &10044&10364 &13119 & 5537 &10141 \\
  9&10069 &11536 & 5216 & 9337&10099 &11796 & 5334 & 9442 \\
  \vdots&\vdots&\vdots&\vdots&\vdots&\vdots&\vdots&\vdots&\vdots\\
 41&  466 & 5682 &  462 & 3625&  846 & 5956 &  684 & 3755 \\
 42&  221 & 5088 &  204 & 3167&  473 & 5405 &  436 & 3421 \\
 43&   90 & 4337 &  103 & 2742&  227 & 4822 &  225 & 3053 \\
 44&   15 & 3633 &   25 & 2307&   81 & 4243 &   95 & 2664 \\
 45&    2 & 2980 &    2 & 1856&   21 & 3575 &   27 & 2210 \\
 46&    3 & 2257 &    4 & 1389&    5 & 2882 &    7 & 1755 \\
 47&    2 & 1595 &    2 &  909&    0 & 2197 &    0 & 1300 \\
 48&    2 &  973 &    2 &  565&    0 & 1538 &    0 &  897 \\
 49&    3 &  538 &    4 &  261&    0 &  998 &    0 &  531 \\
 50&    2 &  250 &    2 &  118&    0 &  542 &    0 &  264 \\
 51&    2 &  101 &    2 &   29&    0 &  256 &    0 &  108 \\
 52&    3 &   19 &    4 &    6&    0 &   87 &    0 &   30 \\
 53&    2 &    6 &    2 &    6&    0 &   23 &    0 &    7 \\
 54&    2 &    6 &    2 &    6&    0 &    5 &    0 &    0 \\
 55&    3 &    6 &    4 &    6&    0 &    0 &    0 &    0 \\
 56&    2 &    6 &    2 &    6&    0 &    0 &    0 &    0 \\
 57&    2 &    6 &    2 &    6&    0 &    0 &    0 &    0 \\
 58&    3 &    6 &    4 &    6&    0 &    0 &    0 &    0 \\
  \vdots&\vdots&\vdots&\vdots&\vdots&\vdots&\vdots&\vdots&\vdots\\
 \end{tabular}
\end{center} 
\end{table*}

%% file: biography.tex
\begin{IEEEbiographynophoto}{Kenta Kasai}
received the B.E., M.E., and Ph.D. degrees from Tokyo Institute of
Technology, Tokyo, Japan, in 2001, 2003, and 2006, respectively.
From 2012 to 2024, he was an Associate Professor at Tokyo Institute
of Technology. He is currently an Associate Professor with the
Department of Information and Communications Engineering, School
of Engineering, Institute of Science Tokyo, Tokyo, Japan.
His research interests include information theory, coding theory,
and quantum error correction.
\end{IEEEbiographynophoto}